\documentclass[11pt,letterpaper]{article}
\usepackage[in]{fullpage}
\usepackage{amsmath,amsthm,amsfonts,amssymb}
\usepackage{liyang}

\usepackage{framed}
\usepackage{verbatim}
\usepackage{enumitem}
\usepackage{array}
\usepackage{multirow}
\usepackage{afterpage}

\usepackage{caption}

\usepackage[usenames,dvipsnames]{xcolor}

\makeatletter
\newtheorem*{rep@theorem}{\rep@title}
\newcommand{\newreptheorem}[2]{
\newenvironment{rep#1}[1]{
 \def\rep@title{#2 \ref{##1}}
 \begin{rep@theorem}\itshape}
 {\end{rep@theorem}}}
\makeatother
\theoremstyle{plain}

\def\colorful{1}

\ifnum\colorful=1

\fi
\ifnum\colorful=0

\fi

\usepackage{boxedminipage}

\newcommand{\ignore}[1]{}

\newreptheorem{theorem}{Theorem}
\newtheorem*{theorem*}{Theorem}
\newreptheorem{lemma}{Lemma}
\newreptheorem{proposition}{Proposition}
\newtheorem*{noclaim*}{Claim}

\newcommand{\Bin}{\mathrm{Bin}}

\def\bs{\mathbf{s}}
\def\br{\mathbf{r}}
\newcommand{\Dy}{\calD_{\text{yes}}}
\newcommand{\Dn}{\calD_{\text{no}}}

\def\rr{\boldsymbol{r}}
\def\sperp{\hspace{0.07cm}\perp}

\title{Beyond Talagrand Functions: New Lower Bounds for Testing \\Monotonicity and Unateness\vspace{0.6cm}}

\author{
Xi Chen\thanks{Columbia University, email: \texttt{xichen@cs.columbia.edu}.}
\and
Erik Waingarten\thanks{Columbia University, email: \texttt{eaw@cs.columbia.edu}.}
\and
Jinyu Xie\thanks{Columbia University, email: \texttt{jinyu@cs.columbia.edu}}}
\begin{document}
\maketitle
\thispagestyle{empty}
\begin{abstract}
We prove a lower bound of $\tilde{\Omega}(n^{1/3})$ for the query complexity of
  any two-sided and \mbox{adaptive} algorithm that tests whether an unknown  Boolean
  function $f:\{0,1\}^n\rightarrow \{0,1\}$ is monotone~or far from monotone.
This improves the recent bound of $\tilde{\Omega}(n^{1/4})$ for the same problem
  by Belovs~and Blais \cite{BB15}.
Our result builds on a new family of random Boolean functions
  that can be viewed as a two-level \vspace{0.02cm}extension of Talagrand's random DNFs.

Beyond monotonicity, we also prove a lower bound of $\tilde{\Omega}(n^{2/3})$ for any two-sided and
  adaptive algorithm, and a lower bound of $\tilde{\Omega}(n)$ for any one-sided and
  non-adaptive algorithm for testing unateness, a natural generalization of monotonicity.
The latter matches the recent linear upper bounds by Khot and Shinkar \cite{KS15} and by
  Chakrabarty and Seshadhri \cite{CS16}.
\end{abstract}\newpage

\hypersetup{linkcolor=magenta}
\hypersetup{linktocpage}
\setcounter{tocdepth}{2}
\setcounter{totalnumber}{1}

\tableofcontents \thispagestyle{empty}

\newpage

\setcounter{page}{1}


\section{Introduction }

Over the last few decades, property testing has emerged as an
  important line of research in sublinear time algorithms{.} 
{The} goal is to understand abilities and limitations of
  randomized algorithms that {determine whether an unknown object} 
  has a specific property or is far from having the property,~by examining randomly
  a small portion of the object.
Over the years many different types~of objects and properties
  have been studied from this property testing perspective
  (see \cite{Ron:08testlearn,PropertyTestingICS,Ron:10FNTTCS} for overviews of
  contemporary property testing research).


In this paper we study the monotonicity testing of Boolean functions,
  one of the most basic~and natural problems that have been studied in the area of property testing
  for many years
  \cite{DGLRRS, GGL+00, EKK+00, FLNRRS, Fis04, BKR:04long, Ailon04,HK08, RS09c, BBM12, BCGM12, RRSW12,CS13a, CS13b, CS13c, BRY:14,CST144,KMS15,CDST15,BB15} with many exciting developments during the past few years.
Introduced by Goldreich, Goldwasser, Lehman, and Ron~\cite{GGL+98},~the problem is concerned with the (randomized) query complexity of determining whether an unknown Boolean function $f:\{0,1\}^n\rightarrow \{0,1\}$ is \emph{monotone} or \emph{far from monotone}. Recall that $f$ is monotone if $f(x) \le f(y)$ for all $x \prec y$
  (i.e., $x_i\le y_i$ for every $i\in [n]=\{1,\ldots,n\}$). We say that $f$ is $\eps$-close~to monotone if $\Pr\hspace{0.03cm}[\hspace{0.03cm}f(\bx)\ne g(\bx)\hspace{0.03cm}]\hspace{0.03cm} \le \eps$ for some monotone   function $g$ where the probability is taken~over a uniform draw of $\bx$ from $\{0,1\}^n$, and that $f$ is $\eps$-far from monotone otherwise.

We are interested in query-efficient randomized algorithms for the following task:
\begin{flushleft}\begin{quote}
{\it Given as input a distance parameter $\eps > 0$ and oracle access to an unknown Boolean\\ function $f:\{0,1\}^n\rightarrow \{0,1\}$, accept with probability at least $2/3$ if $f$ is monotone and\\ reject with probability at least $2/3$ if $f$ is $\eps$-far from monotone.}
\end{quote}\end{flushleft}

Beyond monotonicity, we also work on the testing of \emph{unateness}, a generalization of monotonicity.
Here a Boolean function $f:\{0,1\}^n\rightarrow \{0,1\}$ is unate iff
  there exists a string $r\in \{0,1\}^n$ such that $g(x)=f(x\oplus r)$~is~monotone (i.e.,
  $f$ is either monotone increasing
  or monotone decreasing in each coordinate), where we use
  $\oplus$ to denote the bitwise XOR of two strings.
We are interested~in~query-efficient randomized algorithms that determine whether
  an unknown $f$ is unate or far from unate.

\subsection{Previous work on monotonicity testing and unateness testing}

The work of Goldreich et al.~\cite{GGL+98,GGL+00} proposed a simple ``edge tester.''
For each round, the ``edge tester'' picks an $x\in \{0,1\}^n$ and an $i\in [n]$ uniformly at random and
  queries $f(x)$ and~$f(y)$ with $y=x^{(i)}$, where $\smash{x^{(i)}}$ denotes
  $x$ with its $i$th bit flipped.
If $(x,y)$ is a \emph{violating edge}, i.e., either~1) $x\prec y$ and $f(x)>f(y)$
  or 2) $y\prec x$ and $f(y)>f(x)$, the tester rejects $f$; the tester accepts $f$~if~no violating
   edge is found after a certain number of rounds.
The ``edge tester'' is both one-sided~(i.e.~it always accept when $f$ is monotone)
  and non-adaptive (i.e.\hspace{-0.04cm} its queries do not depend on the oracle's responses to previous queries).
\cite{GGL+00} showed that $O(n/\eps)$ rounds suffice for the
  ``edge tester'' to find a violating edge with high probability when $f$ is $\eps$-far from monotone.

\ignore{$O(n^2\log(1/\eps)/\eps)$ upper bound on its query complexity. This was subsequently improved to $O(n/\eps)$ in~the journal version~.}
Later Fischer et al.~\cite{FLNRRS} obtained the first lower bounds, showing that there is a constant distance parameter $\eps_0 > 0$ such that $\Omega(\log n)$ queries are necessary for any \emph{non-adaptive} algorithm  
and $\Omega(\sqrt{n})$ queries are necessary for any \emph{non-adaptive} and \emph{one-sided} algorithm.
  \ignore{, since any $q$-query adaptive tester can be simulated by a non-adaptive one that simply carries out all $2^q$ possible executions.}

These were the best known results on this problem for more than a decade, until Chakrabarty and Seshadhri  improved  the linear upper bound of Goldreich et al. \hspace{-0.06cm}to $\tilde{O}(n^{7/8}\eps^{-3/2})$ \cite{CS13a} using~a ``pair tester'' which is one-sided and non-adaptive.
Such a tester looks for a so-called \emph{violating pair} $(x,y)$ of $f$  satisfying $x\prec y$ and $f(x)>f(y)$.
Their analysis was later slightly refined by~Chen et~al. in \cite{CST144} to $\tilde{O}(n^{5/6}\eps^{-4})$.
  \cite{CST144} also gave an $\tilde{\Omega}(n^{1/5})$ lower bound for non-adaptive algorithms.

\def\dyes{\mathcal{D}_{\text{yes}}}
\def\dno{\mathcal{D}_{\text{no}}}

Further progress has been made during the past two years.
Chen et al. \cite{CDST15} gave a lower bound of $\Omega(n^{1/2-c})$ for non-adaptive algorithms  
  for any positive constant $c$.
Later an upper bound of $\tilde{O}(n^{1/2} /\epsilon^{2})$
  was obtained by Khot et al. \hspace{-0.05cm}in \cite{KMS15} via a deep analysis of the
  ``pair tester'' based on a new isoperimetric-type theorem
  for far-from-monotone Boolean functions.
These~results (almost)~resolved~the query complexity of {non-adaptive} monotonicity testing
  over Boolean functions.
Very recently Belovs and Blais \cite{BB15} made a breakthrough and gave
  an $\tilde{\Omega}(n^{1/4})$ lower~bound for \emph{adaptive} algorithms.
This is the first polynomial lower~bound for adaptive monotonicity testing.~We discuss the lower bound
 construction of \cite{BB15} in more detail in Section \ref{reviewsec}.

The problem of testing unateness was introduced in the same paper  \cite{GGL+00} by Goldreich~et al.
  where they  obtained a one-sided and non-adaptive algorithm with $O(n^{3/2}/\eps)$ queries.~The~first improvement after \cite{GGL+00} was made by Khot and Shinkar \cite{KS15} with
  a one-sided and \mbox{adaptive} $O(n\log n/\eps)$-query algorithm.
Baleshzar et al. \hspace{-0.05cm}\cite{BMPR16} extended the algorithm of \cite{KS15} to~testing unateness of
  functions $f:\{0,1\}^n\rightarrow \mathbb{R}$ with the same query complexity.~They~also~gave~a~lower bound of $\Omega(\sqrt{n}/\eps)$ for one-sided, non-adaptive
  algorithms over Boolean functions.
 Chakrabarty~and Seshadhri \cite{CS16} recently gave a~{one-sided}, non-adaptive
  algorithm of $O((n/\eps)\log (n/\eps))$ queries.

\subsection{Our results}\label{ourresult}

Our main result is an $\tilde{\Omega}(n^{1/3})$ lower bound for
  adaptive monotonicity testing of Boolean functions,
  improving the $\tilde{\Omega}(n^{1/4})$ lower bound of Belovs and Blais \cite{BB15}.

\begin{theorem}[Monotonicity]\label{monomain}
\begin{flushleft}There exists a constant $\eps_0>0$ such that
  any two-sided and adaptive algorithm for testing whether an unknown Boolean function
  $f:\{0,1\}^n\rightarrow \{0,1\}$
  is monotone or $\eps_0$-far from monotone must make $\Omega\hspace{0.02cm}(n^{1/3}/\log^2n)$ queries.\end{flushleft}
\end{theorem}

In \cite{BB15}, Belovs and Blais obtained their $\tilde{\Omega}(n^{1/4})$ lower bound
  using a family of random~functions known as \emph{Talagrand's random DNFs}
  (or simply as the Talagrand function) \cite{Talagrand}.
A function drawn from this family is the disjunction of $N\equiv 2^{\sqrt{n}}$ many monotone terms $T_i$
  with each $T_i$ being the conjunction of $\sqrt{n}$ variables sampled uniformly from $[n]$. So such a 
  function looks like
$$
f(x)=\bigvee_{i\in [N]} T_i(x)=\bigvee_{i\in [N]} \left(\bigwedge_{k\in S_i} x_k\right).
$$
However,
  it turns out that there is a matching $\tilde{O}(n^{1/4})$-query, one-sided
  algorithm for functions of \cite{BB15}. (See Section \ref{sec:alg} for a sketch of the algorithm.)
So the analysis of \cite{BB15} is  tight.

Our main contribution behind the lower bound of Theorem \ref{monomain}
  is a new and harder family~of~random functions~for monotonicity testing, which
  we call \emph{two-level Talagrand functions}.
This starts by reexamining the construction of \cite{BB15} from a slightly different angle,
  which leads to both natural generalizations and simpler analysis of such functions.
We review the construction of \cite{BB15} under this framework
  and describe our new  two-level Talagrand functions in Section \ref{reviewsec}.
We then~give~an overview of the proof of Theorem \ref{monomain} in Section \ref{sketch2}.
As far as we know, we are not aware of~the~two-level Talagrand functions in the literature
  and expect to see more interesting applications of them in the future.
On the other hand, the techniques developed in the proof of Theorem \ref{monomain} can be
  easily adapted~to~\mbox{prove} a tight $\tilde{\Omega}(n^{1/2})$ lower bound for
  non-adaptive monotonicity testing, removing~the $-c$ in the exponent of \cite{CDST15} (see Section~\ref{sec:non-mono}).


Next for testing unateness, we present an $\tilde{\Omega}(n^{1/2})$ lower bound against adaptive algorithms.

\begin{theorem}[Unateness]\label{unatemain}
\begin{flushleft}There exists a constant $\eps_0>0$ such that
  any two-sided and adaptive algori\-thm for testing whether an unknown Boolean function
  $f:\{0,1\}^n\rightarrow \{0,1\}$
  is unate versus $\eps_0$-far from unate must make $\Omega
  \hspace{0.02cm}(n^{2/3}/\log^{3}n)$ queries.\end{flushleft}
\end{theorem}

The lower bound construction behind Theorem \ref{unatemain} follows
  a similar framework.
Some of the new ideas and techniques developed for the monotonicity lower bound are adapted to
  prove Theorem~\ref{unatemain}
  though with~a few twists that~are unique~to unateness.
%

Moreover,
  we obtain a linear lower bound for one-sided and non-adaptive unateness algorithms.
This improves the $\Omega(\sqrt{n})$ lower bound of Baleshzar et al. \cite{BMPR16} 
  and matches the upper bound of Chakrabarty and Seshadhri \cite{CS16} for
  such algorithms.

\begin{figure}
\begin{center}
\renewcommand{\arraystretch}{1.4}
\begin{tabular}{ |c|c|c|c| }
 \hline
 &Best Upper Bound & Best Lower Bound & This Work \\
 \hline \hline
 Non-adaptive & & &\\
 \hline \hline
 \ Monotonicity\ \ & $\tilde{O}( {\sqrt{n}}/{\eps^2})$ \cite{KMS15} & $\tilde{\Omega}(n^{1/2 - c})$ \cite{CDST15} & $\tilde{\Omega}(\sqrt{n})$\\[-0.4ex]
 Unateness & $\tilde{O}(n/\eps)$ \cite{CS16} & \  $\Omega(\sqrt{n})$ (one-sided) \cite{BMPR16}\ \  & $\ \tilde{\Omega}(n)$ (one-sided)\ \ \\[0.2ex]
 \hline \hline
 Adaptive & & & \\
 \hline\hline
 Monotonicity & \ $\tilde{O}( {\sqrt{n}}/{\eps^2})$ \cite{KMS15}\ \  & $\tilde{\Omega}(n^{1/4})$ \cite{BB15} & $\tilde{\Omega}(n^{1/3})$\\ [-0.4ex]
 Unateness & $\tilde{O}(n/ \eps)$ \cite{KS15, CS16} & & $\tilde{\Omega}(n^{2/3})$ \\ [0.2ex]
 \hline
\end{tabular}\vspace{-0.35cm}
\end{center}
\caption{Previous work and our results on monotonicity testing and unateness testing.\vspace{-0.25cm}} \label{figblablabla}
\end{figure}
\def\Tal{\mathsf{Tal}} \def\Talpm{\mathsf{Tal}_{\pm}} \def\ff{\boldsymbol{f}} \def\gg{\boldsymbol{g}}
\def\SS{\mathsf{S}}

\begin{theorem}[One-sided and non-adaptive unateness]\label{nonadaptive}
\begin{flushleft}There exists a constant $\eps_0>0$ such that
  any one-sided and non-adaptive algori\-thm for testing whether an unknown Boolean function
  is unate versus $\eps_0$-far from unate must make $\Omega\hspace{0.02cm}(n /\log^2 n)$ queries.\end{flushleft}
\end{theorem}

We summarize previous work and our new results in Figure~\ref{figblablabla}.

\subsection{An overview of our construction for Theorem \ref{monomain}}\label{reviewsec}

We start by reviewing the hard functions used in \cite{BB15} (i.e.,  Talagrand's random DNFs), but~this time interpret
  them under the new framework that we will follow throughout the paper.
Employing Yao's minimax principle as usual,
  the goal of \cite{BB15} is to (1) construct a pair of distributions $\dyes^*$ and
  $\dno^*$ over Boolean functions
  from $\{0,1\}^n$ to $ \{0,1\}$ such that $\ff\sim \dyes^*$ is always monotone
  while $\gg\sim \dno^*$ is $\Omega(1)$-far from monotone with probability $\Omega(1)$;
  (2) show that no deterministic algorithm
  with a small number of queries can distinguish them (see equation (\ref{useuse2}) later).

Let $\smash{N=2^{\sqrt{n}}}$. A function $f$ from $\dyes^*$ is drawn using the following procedure.
We first sample~a sequence of $N$ random sub-hypercubes $H_i$ in $\{0,1\}^n$.
Each $H_i$ is defined by a random~term $T_i$ with $x\in  H_i$ if $T_i(x) = 1$,
  where $T_i$ is
  the conjunction of $\sqrt{n}$ random variables sampled uniformly~from~$[n]$
  (so each $H_i$ has dimension $n-\sqrt{n}$).
By a simple calculation most likely~the $H_i$'s have~little~overlap between each other and
  they together cover an $\Omega(1)$-fraction of $\{0,1\}^n$.
Informally we consider~$H_i$'s together as a random \emph{partition} of $\{0,1\}^n$
  where each $x\in \{0,1\}^n$ belongs to a unique $H_i$ (for now~do not worry about
  cases when $x$ lies in none or multiple $H_i$'s).
Next we sample for each~$H_i$ a random dictatorship function
  $h_i(x)= {x_\ell}$ with $\ell$ drawn uniformly from~$[n]$.
The final function is $f(x)=h_i(x)$ for each $x\in H_i$ (again
  do not worry about  cases when $x$ lies in none or multiple $H_i$'s).
A function $g$ from $\dno^*$ is drawn using the same procedure
  except that~each~$h_i$ is now a random anti-dictatorship function $h_i(x)=\overline{x_\ell}$
  with $\ell$ sampled uniformly from $[n]$.

Note that the~distributions sketched here are slightly different from \cite{BB15} (see Section \ref{sec:alg}).
For $\dno^*$ in particular, instead of associating each $H_i$ with an independent,
  random anti-dictatorship~$h_i$, {\cite{BB15} draws
 $\sqrt{n}$ anti-dictatorship functions~\emph{in total} and associates each
 $H_i$ with one of~them~randomly}.\footnote{Note that this is very close but also not exactly the same
   as the distributions used in \cite{BB15}; see Section \ref{sec:alg}.} While this gives a connection to the noise sensitivity results of \cite{MosselOdonnell:03}
   on Talagrand~functions, it makes the functions harder to analyze
 and generalize due to the correlation between $h_i$'s.
 
By definition, $f$ is always monotone.~On the other hand, $g$ is far from monotone as (intuitively)  
  $H_i$'s are mostly disjoint and within each $H_i$, 
 $g$ is anti-monotone due to the anti-dictatorship $h_i$.

At a high level one can view the terms $T_i$ together as an \emph{addressing function}
  in the construction of $\dyes^*$ and $\dno^*$,
  which maps each $x$ to one of the $N$ independent anti-dictatorship functions $h_i$,~by randomly partitioning $\{0,1\}^n$ using a long sequence of small hypercubes $H_i$.
\mbox{Conceptually}, this is the picture that we will follow to define our two-level Talagrand functions.
They will~also~be~built using a random partition of $\{0,1\}^n$ into a sequence of small(er) hypercubes,
  with the property that (i) if one places a   dictatorship function in each
  hypercube independently at random, the resulting function is monotone,
  and (ii) if one places a random anti-dictatorship function
  in each of them, the resulting function is far from monotone with $\Omega(1)$ probability.
The main difference lies in the way how the partition is done and how the hypercubes  are sampled.

Before introducing the two-level Talagrand function, we explain at a high-level
  why the pair~of distributions $\dyes^*$ and $\dno^*$ are hard to distinguish (this will allow us to compare
  them with our~new functions and see why the latter are harder).
Consider the
 situation when an algorithm~is~given~an $x\in H_i$'s with $h_i(x)=0$
  and would like to find a violating pair in $H_i$, by flipping some $1$'s of $x$ to $0$
  and hoping to see $g(y)=1$ in the new $y$ obtained.
The algorithm faces the following  dilemma:
\begin{flushleft}\begin{enumerate}
\item on the one hand,  the algorithm wants to flip as many $1$'s of $x$ as possible
  in order to~flip\\ the hidden anti-dictator variable $\ell$ of the anti-dictatorship function $h_i$;\vspace{-0.08cm}
\item on the other hand, it is very unlikely for the algorithm to flip many (say $\omega(\sqrt{n}\log n)$) $1$'s\\ of $x$
  without moving $y$ outside of $H_i$  (which happens if one of the $1$-entries flipped lies in $T_i$), and when this happens,
  $g(y)$ provides essentially no information about $\ell$.
\end{enumerate}\end{flushleft}
So $g$ is very resilient against such attacks.
However, consider the case when $x\in H_i$ and $h_i(x)=1$; then,
  the algorithm may try to find a violating pair in $H_i$
  by flipping $0$'s of $x$ to $1$, and this time there is no limitation on how many $0$'s of $x$ one can flip!
 In fact flipping $0$'s to $1$'s can~never move $y$ outside~of $H_i$.\footnote{While we tried to keep
  the high-level description here simple, there is indeed a truncation that is always
  applied on $g$, where one set $g(x)=1$ for $|x|>(n/2)+\sqrt{n}$, $g(x)=0$ for $|x|<(n/2)-\sqrt{n}$,
  and keep $g(x)$ the same only when $x$ lies in the middle layers with $|x|$ between $(n/2)-\sqrt{n}$
  and $(n/2)+\sqrt{n}$.
But even with the truncation in place, one can take advantage of this observation and
  find a violation in $g$ using $\tilde{O}(n^{1/4})$ queries. See details in Section \ref{sec:alg}}
In Section \ref{sec:alg}, we leverage this observation to
  find a violation  with $\tilde{O}(n^{1/4})$ queries.

Now we describe the two-level Talagrand function.
The random partitions we employ below~are more complex; they  allow us to
  upperbound not only the number of $1$'s of $x$ that an algorithm can flip (without moving outside
  of the hypercube) but also the
  number of $0$'s as well.
We use $\dyes$ and $\dno$ to denote the two distributions.

To draw a function $f$ from $\dyes$,
  we partition $\{0,1\}^n$ into $N^2$ random sub-hypercubes as follows. First we sample as before
    $N$ random $\sqrt{n}$-terms $T_i$ to obtain $H_i$.
After that, we further partition~each $H_i$,
  by independently sampling $N$ random $\sqrt{n}$-clauses $C_{i,j}$, with each of them
  being the disjunction of $\sqrt{n}$ random variables sampled from $[n]$ uniformly.
The terms $T_i$ and clauses $C_{i,j}$ together define $N^2$ sub-hypercubes $H_{i,j}$:
   \hspace{-0.05cm}$x \in H_{i, j}$ if $T_i(x)=1$ and $C_{i,j}(x) = 0$.
The rest is very similar.~We~sample~a random dictatorship function $h_{i,j}$ for each $H_{i,j}$; the final function $f$ has $f(x)=h_{i,j}(x)$ for~$x\in H_{i,j}$.\footnote{Again, do not worry about cases when $x$ lies in none or multiple $H_{i,j}$'s.}
A function $g$ from $\dno$ is drawn using the same procedure except that $h_{i,j}$'s are
  independent random anti-dictatorship functions.
We call such functions two-level Talagrand functions, as
  one can view each of them as a two-level structure with the top being
  a Talagrand DNF and the bottom being $N$   Talagrand CNFs, one attached with each term of the top DNF. See Figure \ref{fig:function} for a visual depiction.

By a simple calculation, (most likely) the $H_{i,j}$'s have little overlap and
  cover an $\Omega(1)$-fraction~of $\{0,1\}^n$.
This is why $g$ is far from monotone.
It will become clear after the formal definition of~$\dyes$ that $f$ is monotone;
  this relies on how exactly we handle cases~when $x$ lies in none or multiple $H_i$'s.

Conceptually the construction of $\dyes$ and $\dno$ follows the same
  high-level picture: the terms~$T_i$ and
  clauses $C_{i,j}$ together serve as an addressing function, which we refer to as a \emph{multiplexer}
  in~the proof (see Figure \ref{fig:multiplexer} for a visual depiction).
It maps each string $x$ to one of the $N^2$ independent~and random dictatorship or anti-dictatorship   $h_{i^*,j^*}$, depending on whether the function is~from $\dyes$ or $\dno$. 
Terms $T_i$ in the first level of multiplexing determines $i^*$
  and clauses $C_{i^*,j}$ in the second~level~of multiplexing  determines $j^*$.
The new two-level Talagrand functions are harder than those  of \cite{BB15} since, starting with a string $x\in H_{i,j}$,
  not only flipping $\omega(\sqrt{n}\log n)$ many $1$'s would move it outside of $H_{i,j}$ with high probability
  (because the term $T_i$ is most likely no longer satisfied), the same holds when flipping
  $\omega(\sqrt{n}\log n)$ many $0$'s to $1$ (because the clause $C_{i,j}$ is most likely no longer falsified).

\subsection{An overview of the proof of Theorem \ref{monomain}}\label{sketch2}

Let $q=n^{1/3}/\log^2 n$ and let $B$ be a $q$-query deterministic algorithm, which we view equivalently
  as a binary decision tree of depth~$q$.
Our goal is to prove the following for $\dyes$ and $\dno$:
\begin{equation}\label{useuse3}
 \mathop{\Pr}_{\ff\sim \dyes}\big[\text{$B$ accepts $f$}\big]\le \mathop{\Pr}_{\gg\sim \dno} \big[\text{$B$ accepts $g$}\big]+o(1).
\end{equation}
To prove (\ref{useuse3}), it suffices to show for every leaf $\ell$ of $B$,
\begin{equation}\label{useuse2}
\mathop{\Pr}_{\ff\sim \dyes} \big[\hspace{0.01cm}\text{$\ff$ reaches $\ell$}\hspace{0.03cm}\big]
\le (1+o(1))\cdot \mathop{\Pr}_{\gg\sim \dno} \big[\hspace{0.01cm}\text{$\gg$ reaches $\ell$}\hspace{0.03cm}\big].
\end{equation}
However, this is challenging because both events above are highly complex.
Following the same~idea used in \cite{BB15}, we decompose such events into simpler ones
  by allowing the oracle to return more than just $f(x)$.
Upon each query $x\in \{0,1\}^n$, the oracle returns the so-called \emph{signature}
  of $x$.
When $x$ satisfies a unique term $T_{i^*}$, the signature reveals the index $i^*$.
The same happens to the~second level: when $x$ falsifies a unique clause $C_{i^*,j^*}$,
  the signature also reveals the index $j^*$.
(See the formal definition for what happens when $x$ satisfies, or falsifies, none or multiple terms, or clauses.)

{We consider deterministic $q$-query algorithms $B$ with access to this stronger oracle.}
We view $B$ as a decision tree in which each edge
  is labelled with a possible signature returned by the~oracle. 
  Hence the number of children of each
  internal node~is huge.
We refer to such a tree as a \emph{signature tree}.
Our new goal is then to prove that every leaf $\ell$~of~$B$ satisfies (\ref{useuse2}).
However, this is not true~in general. Instead we divide the leaves into
  \emph{good} ones and \emph{bad} ones, prove (\ref{useuse2}) for each good leaf  
  and show that $\ff\hspace{-0.02cm}\sim\hspace{-0.02cm}\dyes$ reaches a bad leaf with probability~$o(1)$.

The definition of bad leaves and the proof of $\ff\sim\dyes$ reaching
  one with $o(1)$ probability {poses the main technical challenge.} 
First, we characterize four types of edges where a \emph{bad} event
  occurs 
    and refer to them as bad~edges; a leaf $\ell$ then is bad if the root-to-$\ell$
  path has a bad edge. {These bad edges help us rule out certain attacks a possible algorithm may try. The first two events formalize the notion we highlighted earlier that given a string $y$ queried before, 
  flipping $\omega(\sqrt{n}\log n)$ many $1$'s of $y$ to $0$'s, or $0$'s to $1$'s, results in a new string $x$ that most likely lies in a different sub-hypercube. The second two events formalize the notion that if queries do not flip many $1$'s to $0$'s, or $0$'s to $1$'s, then observing a violating pair is unlikely. }
  
In a bit more detail, the first two  events are that (we use $A_{i,1}$ and $A_{i,j,0}$ to denote 
  the common $1$-entries of strings queried so far that satisfy the same term $T_i$ and common $0$-entries 
  of strings~so far that falsify the same clause $C_{i,j}$, respectively) after a new query~$x$, $|A_{i,1}|$ or $|A_{i,j,0}|$ drop by more than
  \hspace{-0.05cm}$\sqrt{n}\log n$.
\hspace{-0.03cm}Such events occur~when $x$ satisfies the same $T_i$ but has many $0$-entries in~$A_{i,1}$,~or~$x$ falsifies the same clause $C_{i,j}$ but has many $1$-entries~in $A_{i,j,0}$.
Intuitively such events are unlikely to happen because before $x$ is queried,
  $T_i$ (or $C_{i,j}$) is ``almost''\hspace{0.03cm}\footnote{The distribution is not
  exactly uniform because we also need to consider strings that are known to not satisfy $T_i$
  or not falsify $C_{i,j}$ as revealed in their signatures, though we will see in the proof that
  their influence is very minor.}
   uniform over $A_{i,1}$ (or $A_{i,j,0}$).
Therefore it is unlikely for the $\sqrt{n}\log n$ many $0$-entries~of $x$ in $A_{i,1}$ ($1$-entries of $x$ in $A_{i,j,0}$)
  to entirely avoid~$T_i$ ($C_{i,j}$).  {We follow this intuition to show to that $\ff \sim \dyes$ takes one such bad edge with probability at most $o(1)$, which allows us to prune such edges.} 
 \vspace{0.24cm}


\noindent\textbf{Organization.}
We introduce some notation and review the characterization of distance to monotonicity and
  unateness in Section \ref{sec:pre}. We also prove two basic tree pruning lemmas that will be used
    several times in the paper.
We prove Theorems \ref{monomain}, \ref{unatemain} and \ref{nonadaptive} in Sections
  \ref{sec:mono}, \ref{sec:unate} and \ref{sec:nonadaptive}, respectively.
 %

\section{Preliminaries}\label{sec:pre}

\def\ff{\boldsymbol{f}}

In this section we introduce some notation and tools we will be using.
\subsection{Notation}

We use bold font letters such as $\boldsymbol{T}$ and $\boldsymbol{C}$ for random variables. We write $[n]$ to denote $\{1,\ldots,n\}$.
Given a string $x\in \{0,1\}^n$, we use $|x|$ to denote its Hamming weight, i.e.,
  the number of $1$'s in $x$.
Given a string $x\in \{0,1\}^n$ and $S\subseteq [n]$, we use $x^{(S)}$ to denote the string obtained
  from $x$ by flipping each entry $x_i$ with $i\in S$.
When $S=\{i\}$ is a singleton, we write $x^{(i)}$ instead of $x^{(\{i\})}$~for~convenience.

We use $N$ to denote $2^{\sqrt{n}}$ throughout the paper.
We use $e_i$, for each $i\in [N]$, to denote the string in $\{0,1\}^N$ with its
  $k$th entry being $0$ if $k\ne i$ and $1$ if $k=i$;
we use $e_{i,i'}$, $i<i'\in [N]$, to denote the string in $\{0,1,*\}^N$ with its
  $k$th entry being $0$ if $k<i'$ and $k\ne i$, $1$ if $k=i$ or $i'$, and $*$ if $k>i'$.~We let $\overline{e}_i$ ($\overline{e}_{i,i'}$)
   denote the string obtained from $e_i$ ($e_{i,i'}$) by flipping its
  $0$-entries to $1$ and $1$-entries to $0$.

\subsection{Distance to monotonicity and unateness}

We review some characterizations of distance to monotonicity and unateness. 

\begin{lemma}[Lemma 4 in \cite{FLNRRS}]\label{pfpf}
Let $f \colon \{0, 1\}^n \to \{0, 1\}$ be a Boolean function. Then
\[ \dist\big(f, \textsc{Mono}\big) =  {|M|}\big/{2^n}, \]
where $M$ is the maximal set of disjoint violating pairs of $f$.
\end{lemma}

\begin{lemma}\label{ufuf}
Given $f \colon \{0, 1\}^n \to \{0, 1\}$, let $(E_i^+,E_i^-:i\in [n])$
  be a tuple of sets such that~(1)~each set $E_i^+$ consists of monotone bi-chromatic edges
  $(x,x^{(i)})$ along direction $i$ with $x_i=0$, $f(x)=0$ and $f(x^{(i)})=1$; (2) each set $E_i^-$
  consists of anti-monotone bi-chromatic edges $(x,x^{(i)})$ along direction~$i$ 
  with $x_i=0$, $f(x)=1$ and $f(x^{(i)})=0$;
  (3) all edges in these $2n$ sets are disjoint. 
   Then
\[ \dist\big(f, \textsc{Unate}\big)\ge \frac{1}{2^n} \sum_{i=1}^n \min \big\{ |E_i^+|, |E_{i}^-| \big\}.\]
\end{lemma}
\begin{proof}
By definition, the distance of $f$ to unateness is given by 
$$\dist\big(f,\textsc{Unate}\big)=\min_{r\in \{0,1\}^n} \dist\big(f_r,\textsc{Mono}\big),$$
  where $f_r(x)=f(x\oplus r)$.
On the other hand, since all edges in the $2n$ sets $E_i^+$ and $E_i^-$ are disjoint, it follows from 
  Lemma \ref{pfpf} that
$$
\dist\big(f_r,\textsc{Mono}\big)\ge \frac{1}{2^n} \left(\sum_{i:r_i=0} \big|E_i^-\big|+\sum_{i:r_i=1}\big|E_i^+\big|\right)\ge \frac{1}{2^n}\sum_{i=1}^n \min \big\{ |E_i^+|, |E_{i}^-| \big\}.
$$
This finishes the proof of the lemma.
\end{proof}

\subsection{Tree pruning lemmas}
\label{sec:pruning}

\def\OO{\boldsymbol{O}} \def\frakP{\frak{P}}

We consider a rather general setup where a $q$-query deterministic algorithm $A$
  has oracle access to an object $\OO$ drawn from a distribution $\calD$:
Upon each query $w$, the oracle with an object $O$ returns $\eta(w,O)$, an element from a finite set $\frakP$.
Such an algorithm can be equivalently viewed as a tree~of depth $q$, where
  each internal node $u$ is labelled a query $w$ to make
  and has $|\frakP|$ edges $(u,v)$ leaving $u$, each labelled a distinct element from $\frakP$.
(In general the degree of $u$ can be much larger than two; this is the case for
  all our applications later since we will introduce new oracles that upon a query string
  $x\in \{0,1\}^n$ returns more information than just $f(x)$.)
For this section we do not~care~about labels of leaves of $A$.
Given $A$,
  we present two~basic pruning techniques
   that will help our analysis of algorithms in our lower bound proofs later.

Both lemmas share the following setup.
Given $A$ and a set $E$ of edges of $A$
  we use $L_E$ to denote the set of leaves $\ell$ that has at least one edge in $E$
  along the path from the root to $\ell$.
Each lemma below states that if $E$ satisfies certain properties with respect
  to $\calD$ that we are interested in, then
\begin{equation}\label{pru}
\mathop{\Pr}_{\OO\sim \calD} \big[\text{$\OO$ reaches a leaf in $L_E$}\big]=o(1).
\end{equation}
This will later allow us to focus on root-to-leaf paths that do not take any edge in $E$.

For each node $u$ of tree $A$, we use $\Pr [u]$ to
  denote the probability of $\OO\sim \calD$ reaching $u$.
When $u$ is an internal node with $\Pr[u]>0$ we use $q(u)$ to denote the following conditional probability:
$$
q(u)=\mathop{\Pr}_{\OO\sim \calD}\Big[\hspace{0.02cm}\text{$\OO$ follows an edge in $E$ at $u$}\hspace{0.08cm}\Big|\hspace{0.08cm}
  \text{$\OO$ reaches $u$}\hspace{0.04cm}\Big]=
  \frac{\sum_{(u,v)\in E} \Pr[v]}{\Pr[u]}.
$$





We start with the first pruning lemma; it is
  trivially implied by the second pruning lemma,
  but we keep it because of its conceptual simplicity. 

\begin{lemma}\label{simplepruning}
Given $E$, if $q(u)=o(1/q)$ for every internal node $u$ with $\Pr[u]>0$, then
  (\ref{pru}) holds.
\end{lemma}

\begin{proof}
We can partition the set $L_E$ of leaves into $L_E=\bigcup_{i\in[q]} L_i$, where
   $L_i$ contains leaves with~its first edge from $E$
  being the $i$th edge along its root-to-leaf path.
We also write $E_i$ as the set of edges in $E$ at the $i$th level (i.e., they appear as the
  $i$th edge along root-to-leaf paths).
Then for each $i$,
$$
\mathop{\Pr}_{\OO\sim \calD} \big[\text{$\OO$ reaches $L_i$}\hspace{0.03cm}\big]
\le \sum_{(u,v)\in E_i} \Pr [v] =\sum_{u} \sum_{(u,v)\in E_i} \Pr [v]
=\sum_u \Pr [u]\cdot o(1/q).
$$
Note that the sum is over certain nodes $u$ at the same depth $(i-1)$.
Therefore, $\sum_u \Pr[u]\le 1$ and the proof is completed by taking a union bound
  over $L_i$, $i\in [q]$.
\end{proof}






Next, for each leaf $\ell$ with $\Pr[\ell]>0$ and the root-to-$\ell$ path being
$u_1u_2\cdots u_{k+1}=\ell$, we let $q^*(\ell)$ denote $\sum_{i\in [k]} q(u_i)$.
The second pruning lemma states that (\ref{pru}) holds if $q^*(\ell)=o(1)$ for all such $\ell$.

\begin{lemma}\label{complicatedpruning}
If every leaf $\ell$ of $A$ with $\Pr[\ell]>0$ satisfies $q^*(\ell)=o(1)$, then 
  (\ref{pru}) holds.
\end{lemma}
\begin{proof}
The first part of the proof goes exactly the same as in the proof of the first lemma.

Let $A'$ be the set of internal nodes $u$ with $\Pr[u]>0$.
After a union bound over $L_i$, $i\in [q]$, 
$$
\mathop{\Pr}_{\OO\sim \calD}\big[\text{$\OO$ reaches $L_E$\hspace{0.03cm}}\big]\le
  \sum_{u\in A'} \Pr [u]\cdot q(u).
$$
Let $L_u$ be the leaves in the subtree rooted at $u\in A'$.
We can rewrite $\Pr [u]$ as $\sum_{\ell\in L_u}\Pr [\ell]$.
Thus,
$$
\mathop{\Pr}_{\OO\sim \calD}\big[\text{$\OO$ reaches $L_E$\hspace{0.03cm}}\big]\le
  \sum_{u\in A'} \sum_{\ell\in L_u} \Pr [\ell]\cdot q(u)
=\sum_\ell \Pr [\ell]\cdot q^*(\ell),
$$
where the last sum is over leaves $\ell$ with $\Pr[\ell]>0$; the last 
  equation follows by switching the order of the two sums. 
The lemma follows from $q^*(\ell)=o(1)$ and $\sum_\ell \Pr [\ell]=1$.
\end{proof}


\section{Monotonicity Lower Bound}
\label{sec:mono}

\def\TT{\boldsymbol{T}} \def\CC{\boldsymbol{C}} \def\SS{\boldsymbol{S}} \def\frakP{\frak P}
\def\oe{\overline{e}} \def\hh{\boldsymbol{h}} \def\HH{\boldsymbol{H}}
\def\Ey{\mathcal{E}_{\text{yes}}} \def\En{\mathcal{E}_{\text{no}}}
\def\bGamma{\boldsymbol{\Gamma}} \def\bh{\boldsymbol{h}} \def\bk{{\boldsymbol{k}}}
\def\XX{\boldsymbol{X}}

\subsection{Distributions}\label{sec:dist:mono}

For a fixed $n>0$, we describe a pair of distributions $\Dy$ and $\Dn$ supported on Boolean functions $f: \{0, 1\}^n \to \{0, 1\}$. We then show that every $\ff \sim \Dy$ is monotone, and $\ff\sim \Dn$ is
  $\Omega(1)$-far from monotone with probability $\Omega(1)$.
Recall that $\smash{N=2^{\sqrt{n}}}$. 

A function $\ff\sim \Dy$ is drawn using the following procedure:
\begin{flushleft}\begin{enumerate}
\item Sample a pair $(\TT,\CC)\sim\calE$ (which we describe next). The
  pair $(\TT,\CC)$ is then used to define\\ a \emph{multiplexer} map $\bGamma=\bGamma_{\TT,\CC}:\{0,1\}^n\rightarrow
    (N\times N)\cup \{0^*,1^*\}$.\footnote{We use $0^*$ and $1^*$ to denote two special symbols
    (instead of the Kleene closure of $0$ and $1$).}
\item Sample  $\HH=(\bh_{i,j}:i,j\in [N])$ from a distribution $\Ey$,
  where each $\bh_{i,j}:\{0,1\}^n\rightarrow \{0,1\}$\\ is  a random dictatorship Boolean function, i.e.,
  $\bh_{i,j}(x)=x_k$ with $k$ sampled independently for each $\bh_{i,j}$ and uniformly at random from $[n]$.
\item Finally, $\ff=\ff_{\TT,\CC,\HH}:\{0,1\}^n\rightarrow \{0,1\}$ is defined as follows:
$\ff(x)=1$ if $|x|> (n/2)+\sqrt{n}$; $\ff(x)=0$ if $|x|<(n/2)-\sqrt{n}$; if $(n/2)-\sqrt{n}
  \le |x|\le (n/2)+\sqrt{n}$, we have
$$
\ff(x)=\begin{cases}0 & \text{if
  $\bGamma(x)=0^*$}\\
  1 & \text{if $\bGamma(x)=1^*$}\\
  \bh_{\bGamma(x)}(x) & \text{otherwise (i.e., $\bGamma(x)\in N\times N$)}\end{cases}
$$
\end{enumerate}\end{flushleft}
On the other hand a function $\ff=\ff_{\TT,\CC,\HH}\sim \Dn$ is drawn using the same procedure,
   with~the~only difference being
  that $\HH=(\bh_{i,j}:i,j\in [N])$ is drawn from $\En$ (instead of $\Ey$): each $\bh_{i,j}(x)=\overline{x_k}$~is a
  random anti-dictatorship function
  with $k$ drawn independently and uniformly from $[n]$.

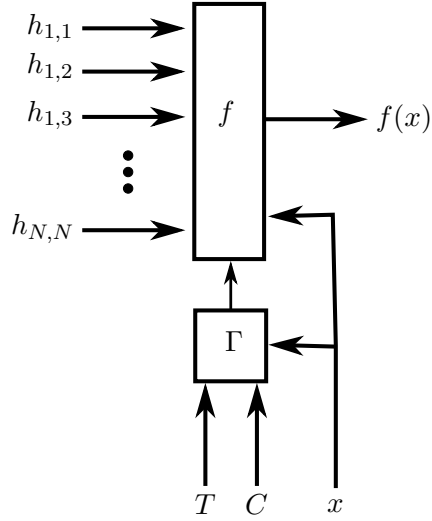
\begin{figure}
\label{fig:multiplex}
\centering
\begin{picture}(140,170)
    \put(0,0){\includegraphics[width=0.25\linewidth]{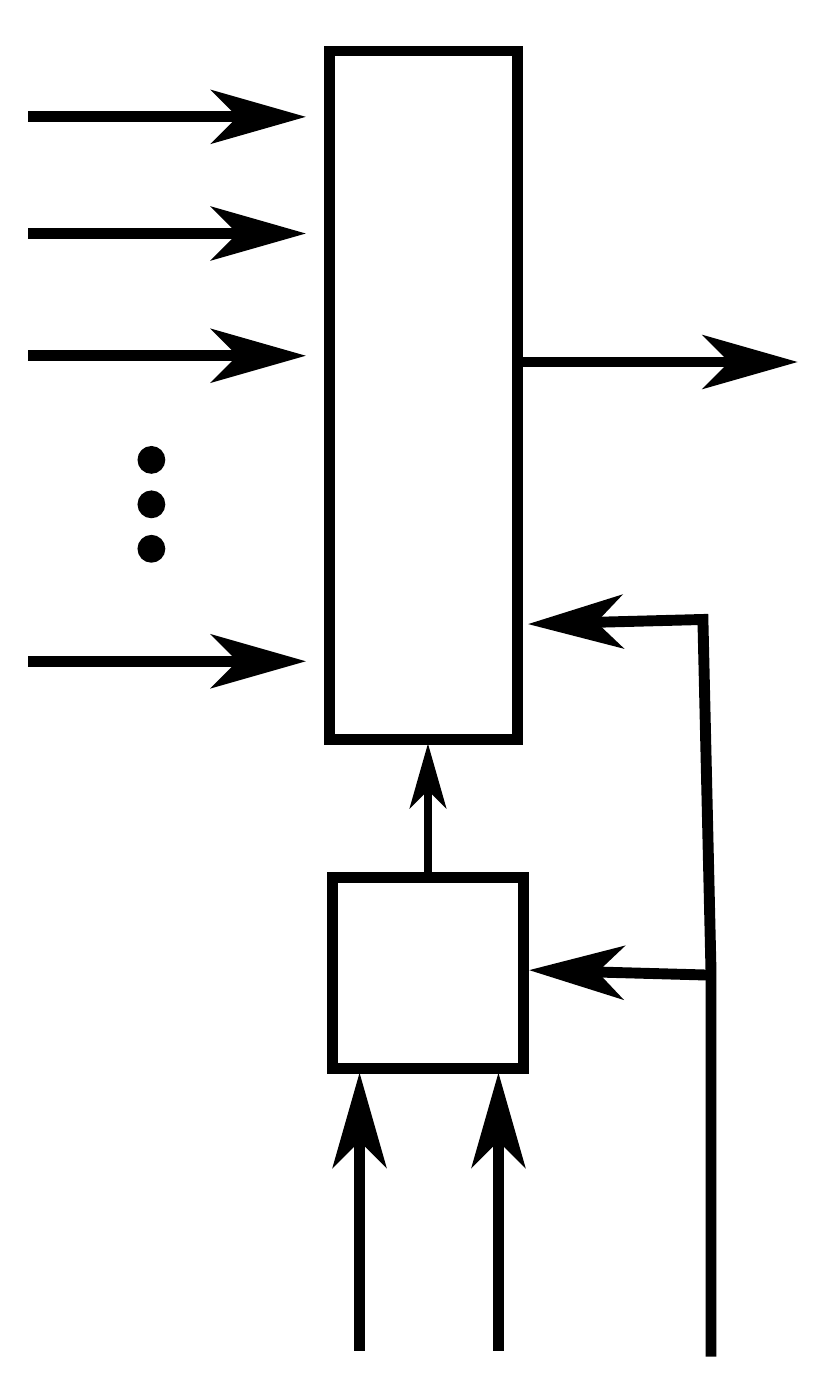}}
    \put(55, 142){$f$}
    \put(58, 55){$\Gamma$}
    \put(-17, 175){$h_{1, 1}$}
    \put(-17, 160){$h_{1, 2}$}
    \put(-17, 142){$h_{1, 3}$}
    \put(-23, 99){$h_{N, N}$}
    \put(46.5, -7){$T$}
    \put(65.5, -7){$C$}
    \put(96.5, -6){$x$}
    \put(115, 140){$f(x)$}
  \end{picture}\vspace{0.32cm}
  \caption{An illustration of the function $f=f_{T, C, H}$ and its dependency on $T$, $C$ and $H$.\vspace{-0.16cm}}\label{fig:multiplexer}
\end{figure}

\begin{remark}
Given the same truncation done in both $\Dy$ and $\Dn$, it suffices
  to show a lower~bound against algorithms that query strings in the \emph{middle layers}  only:
  $(n/2)- \sqrt{n}\le |x|\le (n/2)+ \sqrt{n}$.
\end{remark}

Next we describe the distribution $\calE$ in details.
$\calE$ is uniform over all pairs $(T,C)$ of the following form:
  $T = (T_i:i \in [N])$ with $T_i:[\sqrt{n}]\rightarrow [n]$ and $C=(C_{i,j}:{i,j\in [N]})$
  with $C_{i,j}:[\sqrt{n}]\rightarrow [n]$.
We call $T_i$'s the \emph{terms} and $C_{i,j}$'s the \emph{clauses}.
Equivalently, to draw a pair $(\TT,\CC)\sim \calE$:
\begin{flushleft}\begin{itemize}
\item 
For each $i\in [N]$, we sample a random term $\TT_i$ by
  sampling $\TT_i(k)$ independently and uniformly from $[n]$ for each $k \in [\sqrt{n}]$,
  with $\TT_i(k)$  viewed as the $k$th variable of $\TT_i$.

\item For each $i,j\in [N]$, we sample a random clause $\CC_{i,j}$
  by sampling $\CC_{i,j}(k)$ independently and uniformly from $[n]$ for each $k\in [\sqrt{n}]$,
  with $\CC_{i,j}(k)$ viewed as the $k$th variable of $\CC_{i,j}$.

\end{itemize}\end{flushleft}
Given a pair $(T,C)$,
we interpret $T_i$ as a (DNF) term and abuse the notation to write
\[ T_{i}(x) = \bigwedge_{k \in [\sqrt{n}]} x_{T_{i}(k)} \]
as a Boolean function over $n$ variables.
We say $x$ \emph{satisfies} $T_i$ when $T_i(x) = 1$.
We interpret~each~$C_{i, j}$ as a (CNF) clause and abuse the notation to write
\[ C_{i,j}(x) = \bigvee_{k \in [\sqrt{n}]} x_{C_{i,j}(k)} \]
as a Boolean function over $n$ variables. Similarly we say $x$ \emph{falsifies} $C_{i,j}$ when $C_{i,j}(x) = 0$.

Each pair $(T,C)$ in the support of $\calE$ defines a multiplexer map $\Gamma=\Gamma_{T,C}: \{0,1\}^n\rightarrow (N\times N)\cup \{0^*,1^*\}$.
Informally speaking, $\Gamma$ consists of two levels: the first level uses the terms
  $T_i$ in $T$ to pick the first index $i'\in [N]$; the second level uses the clauses
  $C_{i',j}$ in $C$ to pick the second~index~$j'\in [N]$.
Sometimes $\Gamma$ may choose to directly determine the value of the function by setting $\Gamma(x)\in \{0^*,1^*\}$.

Formally, $(T,C)$ defines $\Gamma$ as follows.
Given an $x\in \{0,1\}^n$ we have $\Gamma(x)=0^*$ if $T_i(x)=0$~for~all $i\in [N]$
  and $\Gamma(x)=1^*$ if $T_i(x)=1$ for at least two different $i$'s in $[N]$.
Otherwise there is a~unique $i'$ with $T_{i'}(x)=1$, and the multiplexer
  enters the second level.
Next, we have $\Gamma(x)=1^*$ if $C_{i',j}(x)=1$ for all $j\in [N]$
  and $\Gamma(x)=0^*$ if $C_{i',j}(x)=0$ for at least two different $j$'s in $[N]$.
Otherwise there~is~a unique $j'\in [N]$ with $C_{i',j'}(x)=0$ and in this case the multiplexer
  outputs $\Gamma(x)=(i',j')$.

This finishes the definition of $\Dy$ and $\Dn$.
Figure~\ref{fig:function} above gives a graphical representation of such functions.
We now prove the properties of $\Dy$ and $\Dn$ promised at the beginning.

\begin{figure}
\centering
\begin{picture}(220,170)
    \put(0,0){\includegraphics[width=0.6\linewidth]{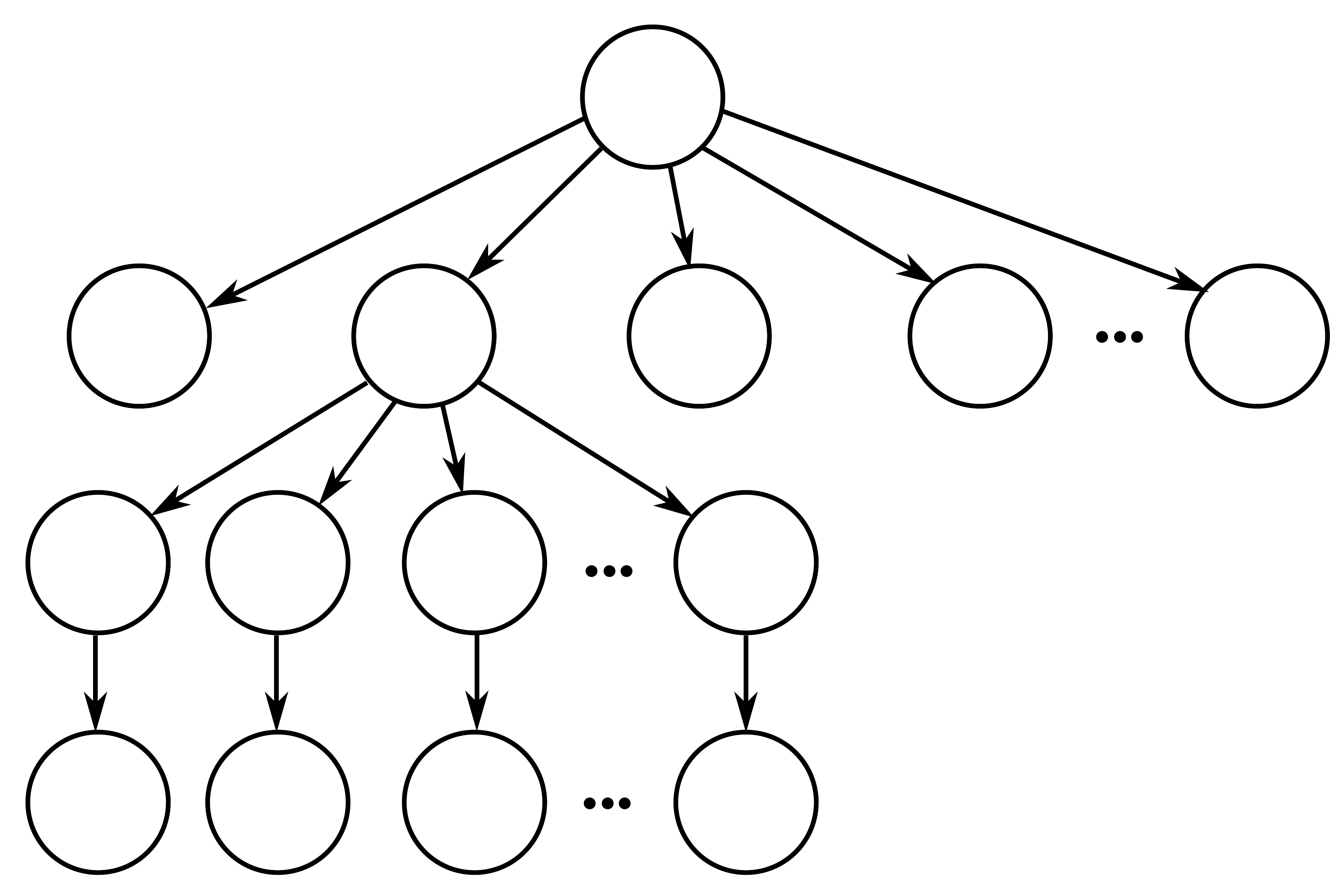}}
    \put(133, 165){$f$}
    \put(24, 115){$T_1$}
    \put(84, 115){$T_2$}
    \put(142, 115){$T_3$}
    \put(201, 115){$T_4$}
    \put(258, 115){$T_{N}$}
    \put(11, 67){$C_{2, 1}$}
    \put(50, 67){$C_{2, 2}$}
    \put(90, 67){$C_{2, 3}$}
    \put(146, 67){$C_{2, N}$}
    \put(11, 17){$h_{2, 1}$}
    \put(50, 17){$h_{2, 2}$}
    \put(90, 17){$h_{2, 3}$}
    \put(146, 17){$h_{2, N}$}
  \end{picture}\vspace{0.18cm}
  \caption{Picture of a function $f$ in the support of $\Dy$ and $\Dn$. We think of evaluating $f(x)$ as following the arrows down the tree. The first level represents multiplexing $x \in \{0, 1\}^n$ with~respect to the terms in $T$. If $x$ satisfies no terms, or multiple terms, then $f$ outputs $0$, or $1$, respectively. If $x$ satisfies $T_i$ for a \emph{unique} term $T_i$ ($T_2$ in the picture), then we follow the arrow to $T_i$ and proceed~to the second level. If $x$ falsifies no clause, or multiple clauses, then $f$ outputs $1$, or $0$, respectively. If $x$ falsifies a unique clause $C_{i, j}$, then we follow the arrow to $C_{i, j}$ and output $h_{i, j}(x)$.\vspace{-0.3cm}}\label{fig:function}
\end{figure}
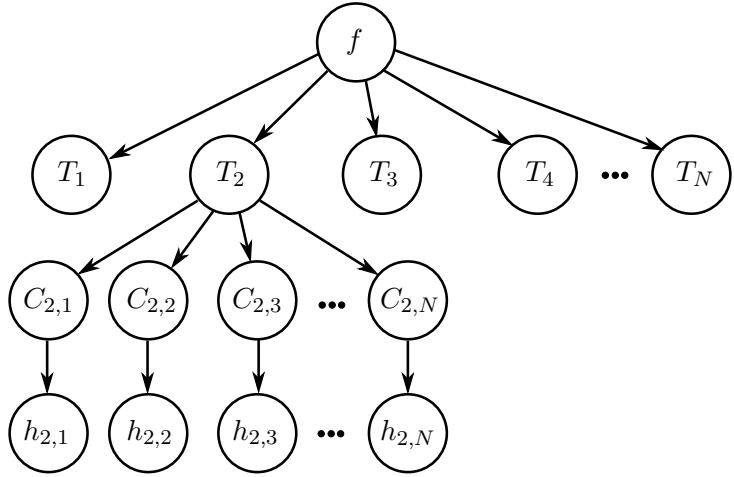

\begin{lemma}\label{monotone:lem}
Every function $f$ in the support of $\Dy$ is monotone.
\end{lemma}

\begin{proof}
Consider  $f = f_{T,C,H}$ with
  $(T,C)$ from the support of $\calE$ and $H$ from the support of $\Ey$.
Let $x\in\{0,1\}^n$ be a string with $f(x)=1$ and $x_i=0$ for some $i$.
Let $x'=x^{(i)}$.
We show that $f(x')=1$.

First note that every term in $T$ satisfied by $x$ remains satisfied by $x'$;
  every clause satisfied by~$x$ remains satisfied by $x'$.
As a result if $\Gamma(x)=1^*$ then $\Gamma(x')=1^*$ as well.
Assume that $\Gamma(x)=(i,j)$. Then $h_{i,j}(x)=f(x)=1$.
For this case we have either $\Gamma(x')=1^*$ and $f(x')=1$, or $f(x')=h_{i,j}(x')$
  and $h_{i,j}(x')=h_{i,j}(x)=1$ because $h_{i,j}$ here is a dictatorship function.
\end{proof}

\begin{lemma}\label{nonmonotone:lem}
A function $\ff \sim \Dn$ is $\Omega(1)$-far-from monotone with probability $\Omega(1)$.
\end{lemma}

\begin{proof}
Fix a pair $(T,C)$ from the support of $\calE$ and an $H$ from the support of $\En$.
Let $f = f_{T,C,H}$. 

Consider the set $X \subset \{0, 1\}^n$ consisting of strings $x$ in the middle layers (i.e., $|x|\in
  (n/2)\pm \sqrt{n}$) with $f(x)=1$, $\Gamma(x)=(i,j)$ for some~$i,j$ $\in [N]$ (instead of $0^*$ or $1^*$), and $h_{i,j}$
  being an anti-dictator function on the $k$th variable for some $k\in [n]$ (so $x_k=0$).~For each $x\in X$, we write $\eta(x)$ to denote the anti-dictator variable $k$ in $h_{i,j}$ and use
  $x^*$ to denote $x^{(\eta(x))}$. (Ideally, we~would like to conclude that
  $(x,x^*)$ is a violating edge of $f$ as $\smash{h_{i,j}(x^*)=0}$. However, flipping one bit potentially may also
  change the value of the multiplexer map $\Gamma$.
So we need to further refine the set $X$.)

Next we
define the following two events with respect to a string $x \in X$ (with $\Gamma(x)=(i,j)$):
\begin{itemize}
\item $E_1(x)$: This event occurs when $\eta(x) \neq C_{i,j}(\ell)$ for any $\ell \in [\sqrt{n}]$
  (and thus, $C_{i,j}(x^*)=0$);\vspace{-0.12cm} 
\item $E_2(x)$: This event occurs when $T_{i'}(x^*) = 0$ for all $i'\ne i \in [N]$.
\end{itemize}
We use $X'$ to denote the set of strings $x\in X$ such that both $E_1(x)$ and $E_2(x)$ hold.
The following claim shows that $(x,x^*)$ for every $x\in X'$ is a violating edge of $f$.

\begin{claim}
For each $x \in X'$, $(x,x^*)$ is a violating edge of $f$.
\end{claim}

\begin{proof}
It suffices to show that $f(x^*)=0$. As $x$ satisfies a unique term $T_i$
  ($T_i$ cannot have $\eta(x)$ as a  variable because $x_{\eta(x)}=0$),
  it follows from $E_2(x)$ that $x^*$ uniquely satisfies the same $T_i$.
It follows from $E_1(x)$ that $x^*$ uniquely falsifies the same clause $C_{i,j}$.
As a result, $f(x^*)=h_{i,j}(x^*)=0$.
\end{proof}

Furthermore, the violating edges $(x,x^*)$ induced by strings $x\in X'$ are indeed disjoint.
(This~is because, given $x^*$, one can uniquely reconstruct $x$ by locating
  $h_{i,j}$ using $\Gamma(x^*)$ and flipping the $k$th bit of $x^*$ if $h_{i,j}$ is
  an anti-dictator function over the $k$th variable.)
Therefore, it suffices to show~that $\XX'$ (as a random set)
  has size $\Omega(2^n)$ with probability $\Omega(1)$, over choices $(\TT,\CC)\sim \calE$ and $\HH\sim\En$.
The lemma then follows from the characterization of \cite{FLNRRS} as stated in Lemma \ref{pfpf}.

Finally we work on the size of $\XX'$. Fix a string $x\in \{0,1\}^n$ in the middle layers.
The next~claim shows that, when $(\TT,\CC)\sim \calE$ and $\HH\sim \En$,
  $\XX'$  contain $x$ with $\Omega(1)$ probability.

\begin{claim}\label{cl:hehe}
For each $x \in \{0, 1\}^n$ with $(n/2) - \sqrt{n} \leq |x|\leq (n/2) + \sqrt{n}$, we have
\[ \mathop{\Pr}_{(\TT,\CC)\sim\calE,\hspace{0.05cm}\HH\sim\En}\big[\hspace{0.03cm}x \in \XX' 
\hspace{0.02cm}\big] = \Omega(1). \]
\end{claim}
\begin{proof}
Fix an $x \in \{0, 1\}^n$ in the middle layers.
We calculate the probability of $x\in \XX'$.

We partition the event of $x\in \XX'$ into $\Theta(nN^2)$ subevents indexed by $i,j\in [N]$ and $k\in [n]$ with
  $x_k=0$. Each subevent corresponds to
  1) Condition on $\TT$: both $x$ and $\smash{x^{(k)}}$ satisfy uniquely the $i$th term; 2)
  Condition on $\CC$: both $x$ and $x^{(k)}$ falsify uniquely the $j$th term;  3)
  Condition on $\HH$: $h_{i,j}$~is the anti-dictatorship function over the $k$th variable.
The probability of 3) is clearly {$1/n$}. 

The probability of 1) is at least
$$
\left(1-\left(\frac{n/2+\sqrt{n}+1}{n}\right)^{\sqrt{n}}\right)^{N-1}\times \left(\frac{n/2-\sqrt{n}}{n}\right)^{\sqrt{n}}=\Omega\left(\frac{1}{N}\right).
$$
The probability of 2) is at least
$$
\left(1-\left(\frac{n/2+\sqrt{n}}{n}\right)^{\sqrt{n}}\right)^{N-1}\times \left(\frac{n/2-\sqrt{n}+1}{n}\right)^{\sqrt{n}}=\Omega\left(\frac{1}{N}\right).
$$
As a result, the probability of $x\in \XX'$ is
  $\Omega(nN^2)\times \Omega(1/N)\times \Omega(1/N)\times \Omega(1/n)=\Omega(1)$.
\end{proof}

From Claim \ref{cl:hehe} and the fact that there are $\Omega(2^n)$ strings in the middle layer,
  the expected size of $\XX'$ is $\Omega(2^n)$.
Via Markov, $|\XX'|=\Omega(2^{n})$ with probability $\Omega(1)$. This finishes the proof.
\end{proof}

Given Lemma \ref{monotone:lem} and \ref{nonmonotone:lem}, Theorem~\ref{monomain}  follows
  directly from the following lemma which we show in the rest of the section. For the rest of the proof we fix the number of queries $q = {n^{1/3}}/{\log^2 n}$.

\begin{lemma}
\label{maintechnical}
Let $B$ be any $q$-query, deterministic algorithm with oracle access to $f$. Then
$$
\mathop{\Pr}_{\ff\sim \Dy} \big[\hspace{0.01cm}\text{$B$ accepts $\ff$}\hspace{0.03cm}\big]\le
\mathop{\Pr}_{\ff\sim \Dn}\big[\hspace{0.01cm}\text{$B$ accepts $\ff$}\hspace{0.03cm}\big]+o(1).
$$
\end{lemma}

Since $f$ is truncated in both distributions, we may assume WLOG that 
  $B$ queries strings in the middle layers only (i.e., strings $x$ with $|x|$ between $(n/2)-\sqrt{n}$ and
    $(n/2)+\sqrt{n}$).

\subsection{Signatures and the new oracle}

Let $(T,C)$ be a pair from the support of $\calE$ and $H$ be a tuple from the support of $\Ey$ or $\En$.
Towards Lemma \ref{maintechnical}, we are interested in deterministic algorithms
  that have
  oracle access to $f=f_{T,C,H}$ and attempt to
  distinguish $\Dy$ from $\Dn$ (i.e., accept if $H$ is from $\Ey$ and reject if it is from $\En$).

For convenience of our lower bound proof, we assume below that the oracle
  returns more than just $f(x)$ for each query $x\in \{0,1\}^n$; instead of simply returning $f(x)$, the oracle
  returns a $4$-tuple $(\sigma,\tau,a,b)$ called the \emph{full signature}
  of $x\in \{0,1\}^n$ with respect to $(T,C,H)$ (see Definition \ref{def:signature} below).
It will become clear later that $f(x)$ can always be derived correctly from
  the full signature of $x$ and thus,
  query lower bounds against the new oracle carry over to the standard oracle.
Once the new oracle is introduced, we may actually ignore the function
  $f$ and view any algorithm as one that has oracle access to the hidden triple $(T,C,H)$
  and attempts to tell whether $H$ is from $\Ey$ or $\En$.

We first give the syntactic definition of \emph{full signatures}.

\begin{definition}
We use $\frakP$ to denote the set of all $4$-tuples $(\sigma,\tau,a,b)$ with
  $\sigma \in \{0,1,*\}^N$ and~$\tau \in \{0,1,*\}^N\cup \{\perp\}$ and $a,b\in \{0,1,\perp\}$
  satisfying the following properties:
\begin{flushleft}\begin{enumerate}
\item $\sigma$ is either 1) the all-$0$ string $0^N$;
  2) $e_i$ for some $i\in [N]$;
  or 3) $e_{i,i'}$ for some $i<i'\in [N]$.

\item $\tau=\sperp$ if $\sigma$ is of case 1) or 3).
Otherwise \emph{(}when $\sigma=e_i$ for some $i$\emph{)},
  $\tau\in \{0,1,*\}^N$ is either 1) the all-$1$ string $1^N$;
  2) $\oe_{j}$ for some $j\in [N]$;
  or 3) $\oe_{j,j'}$ for some $j<j'\in [N]$.

\item $a=b=\sperp$ unless:
1) If $\sigma=e_i $ and $\tau=\oe_j $ for some  $i,j\in [N]$, then $a\in \{0,1\}$ and $b=\sperp$;
or 2) If $\sigma=e_i $ and $\tau=\oe_{j,j'} $ for some $i\in [N]$ and $j<j'\in [N]$,
  then $a,b\in \{0,1\}$.
\end{enumerate}\end{flushleft}
\end{definition}

We next define semantically the
  full signature of $x\in \{0,1\}^n$ with respect to $(T,C,H)$.

\begin{definition}[Full signature]\label{def:signature}
We say $(\sigma,\tau,a,b)$ is the \emph{full signature} of a string $x\in \{0,1\}^n$~with respect to $(T,C,H)$ if it satisfies the following properties:
\begin{flushleft}\begin{enumerate}
\item First, 
  $\sigma\in \{0,1,*\}^N$ is determined by $T$ 
  according to one of the following three cases: 1) $\sigma$
  is the all-$0$ string $0^N$ if
  $T_i(x)=0$ for all $i\in [N]$;
2) If there is a unique $i\in [N]$ with $T_{i}(x)=1$, then $\sigma=e_i $;
or 3) If there are more than one index $i\in [N]$ with $T_{i}(x)=1$, then $\sigma=e_{i,i'} $ with
  $i<i'\in [N]$ being the smallest two such indices. 
We call $\sigma$ the \emph{term signature} of $x$.

\item Second, 
  $\tau=\sperp$
  if $\sigma$ is of case 1) or 3) above.
Otherwise, assuming that $\sigma=e_i$,
  $\tau\in \{0,1,*\}^N$ is determined by $(C_{i,j} : j\in [N])$, 
  according to one of the following cases:
1) $\tau$ is the all-$1$ string $1^N$ if $C_{i,j}(x)=1$ for all $j\in [N]$;
2) If there is a unique $j\in [N]$ with $C_{i,j}(x)=0$, then $\tau=\oe_j$;
or 3) If there are more than one index $j\in [N]$ with $C_{i,j}(x)=0$, then $\tau=\oe_{j,j'}$ with
  $j<j'\in [N]$ being the smallest two such indices.
We call $\tau$ the \emph{clause signature} of $x$.

\item  Finally, $a=b=\sperp$ unless:
1) If $\sigma=e_i$ and $\tau=\oe_j$ for some $i,j\in [N]$, then $a=h_{i,j}(x)$ and $b=\sperp$;
or 2) If $\sigma=e_i$ and $\tau=\oe_{j,j'}$ for some $i,j<j'\in [N]$,
  then $a=h_{i,j}(x)$ and $b=h_{i,j'}(x)$.
\end{enumerate}\end{flushleft}
\end{definition}

It follows from the definitions that the full signature of $x$
  with respect to $(T,C,H)$ is in $\frakP$.
We also define the \emph{full signature} of a set of strings $Q$ with respect to $(T,C,H)$.

\begin{definition}
The \emph{full signature} \emph{(}map\emph{)} of a set $Q\subseteq \{0,1\}^n$ with respect to a triple $(T,C,H)$ is a map $\phi \colon Q \to \frakP$~such that $\phi(x)$ is the full signature of $x$ with respect to $(T,C,H)$ for each $x\in Q$.
\end{definition}

For simplicity, we will write $\phi(x) = (\sigma_x, \tau_x, a_x, b_x)$ to specify the term and clause signatures of $x$ as well as the values of $a$ and $b$ in the full signature $\phi(x)$ of $x$.
Intuitively we may view $\phi$ as two levels of tables with entries in $\{0, 1, *\}$. The (unique) top-level table ``stacks'' the term signatures $\sigma_x$, where each row corresponds to a string $x \in Q$ and each column corresponds to~a term $T_i$ in $T$. In the second level
  a table appears for a term $T_i$ if the term signature of some string $x\in Q$ is $e_i$.~In this case the second-level table at $T_i$ ``stacks'' the clause signatures $\tau_x$ for each $x \in Q$ with $\sigma_x = e_{i} $ 
where each row corresponds to such an $x$ and
  each column corresponds to a clause $C_{i,j}$ in $C$. (The number of columns is still $N$ since
  we only care about clauses $C_{i,j}$, $j\in [N]$, in the table at $T_i$.)

The lemma below shows that the new oracle is at least as powerful as the standard oracle.

\begin{lemma}\label{simple}
Let $(T,C)$ be from the support of $\calE$ and
  $H$ from the support of $\Ey$ or $\En$.
Given~any string $x\in \{0,1\}^n$, $f_{T,C,H}(x)$ is determined by
  its full signature with respect to $(T,C,H)$.
\end{lemma}
\begin{proof}
First if $x$ does not lie in the middle layers, then $f(x)$ is determined by $|x|$.
Below we assume that $x$ lies in the middle layers.
Let $(\sigma,\tau,a,b)$ be the full signature of $x$. There are five cases:
\begin{enumerate}
\item (No term satisfied) If $\sigma = 0^N$, then $f(x) = 0$.\vspace{-0.12cm}
\item (Multiple terms satisfied) If $\sigma = e_{i, i'}$ for some $i, i' \in [N]$, then $f(x) = 1$.\vspace{-0.12cm}
\item (Unique term satisfied, no clause falsified) If $\sigma = e_{i}$ but $\tau = 1^N$, then $f(x) = 1$.\vspace{-0.12cm}
\item (Unique term satisfied, multiple clauses falsified) If $\sigma = e_{i}$ but $\tau = \oe_{j, j'}$, then $f(x) = 0$.\vspace{-0.12cm}
\item (Unique term satisfied, unique clause satisfied) If $\sigma = e_{i}$ and $\tau = \oe_{j}$, then $f(x) = a$.
\end{enumerate}
This finishes the proof of the lemma.
\end{proof}

Given Lemma \ref{simple},
  it suffices to consider deterministic algorithms with the new oracle access
  to a hidden triple $(T,C,H)$, and Lemma \ref{maintechnical} follows directly from the following lemma:

\begin{lemma}\label{maintechnical2}
Let $B$ be any $q$-query algorithm with the new oracle access to $(T,C,H)$. Then
$$
\mathop{\Pr}_{(\TT,\CC)\sim \calE,\HH\sim\Ey} \Big[\text{$B$ accepts $(\TT,\CC,\HH)$}\hspace{0.03cm}\Big]
 \le
\mathop{\Pr}_{(\TT,\CC)\sim\calE,\HH\sim \En}\Big[\text{$B$ accepts $(\TT,\CC,\HH)$}\hspace{0.03cm}\Big]+o(1).
$$
\end{lemma}

Such a deterministic algorithm $B$ can be equivalently viewed as a decision tree of depth $q$
  (and we will abuse the notation  to also denote this tree by $B$).
Each leaf of the tree $B$ is labeled either ``accept'' or ``reject.''
Each internal node $u$ of $B$ is labeled with a query string $x \in \{0, 1\}^n$,
  and each of its outgoing edges $(u,v)$ is labeled a tuple from $\frakP$.
We refer to such a tree as a \emph{signature tree}.

As the algorithm executes, it traverses a root-to-leaf path down the tree making queries to the oracle corresponding to queries in the nodes on the path. For instance at node $u$, after the algorithm queries $x$ and the oracle returns the full signature of $x$ with respect to the unknown $(T,C,H)$, the algorithm follows the outgoing edge $(u,v)$ with that label. Once  a leaf $\ell$ is reached, $B$ accepts if $\ell$ is labelled ``accept''
  and rejects otherwise.

Note that the number of children of each internal node is $| \frakP|$, which is
  huge. Algorithms with the new oracle may adapt its queries to the full signatures returned by the oracle, while under the standard oracle, the queries may only adapt to the value of the function at previous queries. Thus, while algorithms making $q$ queries in the standard oracle model can be described by a tree of size $2^{q}$, $q$-query algorithms with this new oracle are given by signature trees of size $(2^{\Theta(\sqrt{n})})^q$.

We associate each node $u$ in the tree $B$ with a map $\phi_u:Q_u\rightarrow \frakP$
  where $Q_u$ is the set of queries made along the path from the root to $u$ so far,
  and $\phi_u(x)$ is the label of the edge that the root-to-$u$ path takes after querying $x$.
We will be interested in analyzing the following two quantities:
$$
\mathop{\Pr}_{(\TT,\CC)\sim \calE,\HH\sim \Ey} \Big[\text{$({\TT,\CC,\HH})$
  reaches $u$}\hspace{0.03cm}\Big]\quad\text{and}\quad
\mathop{\Pr}_{(\TT,\CC)\sim \calE,\HH\sim \En} \Big[\text{$({\TT,\CC,\HH})$
  reaches $u$}\hspace{0.03cm}\Big].
$$
In particular, Lemma \ref{maintechnical2} would follow trivially if for every leaf $\ell$ of $B$:
\begin{equation}\label{pepe}
\mathop{\Pr}_{(\TT,\CC)\sim \calE,\HH\sim \Ey} \Big[\text{$({\TT,\CC,\HH})$
  reaches $\ell$}\hspace{0.03cm}\Big]\le (1+o(1))\cdot
  \mathop{\Pr}_{(\TT,\CC)\sim \calE,\HH\sim \En} \Big[\text{$({\TT,\CC,\HH})$
  reaches $\ell$}\hspace{0.03cm}\Big].
\end{equation}
However, (\ref{pepe}) above does not hold in general. Our plan for the rest of the proof is to prune
  an~$o(1)$-fraction
  of leaves (measured in terms of their total probability under the yes-case)
  and show (\ref{pepe})~for the rest.
To better understand these probabilities, we need to first introduce some useful notation.

\subsection{Notation for full signature maps}
Given a map $\phi:Q\rightarrow \frakP$ for some $Q\subseteq \{0,1\}^n$,
  we write $\phi(x)=(\sigma_x,\tau_x,a_x,b_x)$ for each $x\in Q$
  and use $\sigma_{x,i},\tau_{x,j}$ to denote the $i$th entry and $j$th entry of $\sigma_x$ and $\tau_x$, respectively.
Note that $\tau_{x,j}$ is not defined if $\tau_x=\sperp$.
(Below we will only be interested in $\tau_{x,j}$ if $\sigma_x=e_i$ for some $i\in [N]$.)

We introduce the following notation for $\phi$. 
We say $\phi$ \emph{induces a tuple} $( I; J; P; R; A;\rho)$, where
\begin{flushleft}\begin{itemize}
\item The set $ I \subseteq [N]$
is given by
$ I = \{ i \in [N] : \exists\hspace{0.04cm} x \in Q\ \text{with}\ \sigma_{x,i}=1  \}.$
(So in terms of the first-level table, $I$ consists of columns that contain at least one $1$-entry.)
\item $J=(J_i\subseteq [N]: i\in I)$ is a tuple of sets indexed by $i \in I$. For each $i \in I$, we have
\[ J_i = \big\{ j \in [N]:\exists\hspace{0.04cm} x \in Q\ \text{with}\
\sigma_x = e_i\ \text{and}\ \tau_{x, j} = 0 \big\}. \]
(In terms of the second-level table at $T_i$, $J_i$ consists of columns
  that contain at least one $0$-entry.)
By the definition of $\frakP$, each $x$ with $\sigma_x=e_i$ can contribute at most two $j$'s to $J_i$.
Also $x$ does not contribute any $j$ to $J_i$ if $\sigma_x=e_{i,i'}$ or $e_{i',i}$,
  in which case $\tau_x=\sperp$, or if $\sigma_x=e_i$ but $\tau_x=1^N$.
So in general $J_i$ can be empty for some $i\in I$.
\item $P=(P_i,P_{i,j}:i\in I,j\in J_i)$ is a tuple of two types of subsets of $Q$. For $i\in I$ and
  $j\in J_i$,
\[ P_i = \big\{ x \in Q:\sigma_{x, i} = 1 \big\}\quad\text{and}\quad
P_{i, j} = \big\{ x \in Q: \sigma_x=e_i\ \text{and}\ \tau_{x, j} = 0\big\}. \]
(In terms of the first-level table, $P_i$ consists of rows that are $1$ on the $i$th column;
in terms of the second-level table at $T_i$, $P_{i,j}$ consists of rows that are $0$
  on the $j$th column.) Note that both $P_i$ and $P_{i,j}$ are not empty by the definition of $I$ and $J_i$.
\item $R=(R_i,R_{i,j}:i\in I,j\in J_i)$ is a tuple of two types of subsets of $Q$. For $i\in I$ and
  $j\in J_i$,
\[ R_i= \big\{x\in Q: \sigma_{x,i}=0\big\}\quad\text{and}\quad
R_{i,j} = \big\{ x\in Q: \sigma_x=e_i\ \text{and}\ \tau_{x,j}=1\big\}.\]
(In terms of the first-level table, $R_i$ consists of rows that are $0$ on the $i$th column;
in terms of the second-level table at $T_i$, $R_{i,j}$ consists of rows that are $1$
  on the $j$th column.)
\item $A=( A_{i,0},A_{i,1},A_{i,j,0},A_{i,j,1}:i\in I,j\in J_i)$ is a tuple
  of subsets of $[n]$. For $i \in I$ and $j \in J_i$,
\begin{align*}
A_{i, 1} = \big\{ k \in [n]: \forall\hspace{0.05cm} x \in P_i,\hspace{0.05cm} x_k = 1 \big\} \quad&\text{and}\quad
A_{i, 0} = \big\{ k \in [n]:  \forall\hspace{0.05cm} x \in P_i,\hspace{0.05cm} x_k = 0 \big\}\\[0.6ex]
A_{i,j, 1} = \big\{ k \in [n]: \forall\hspace{0.05cm} x \in P_{i,j},\hspace{0.05cm} x_k = 1\big\} \quad&\text{and}\quad
A_{i, j, 0} = \big\{ k \in [n]: \forall\hspace{0.05cm} x \in P_{i, j},\hspace{0.05cm} x_k = 0 \big\}.
\end{align*}
Note that all the sets are well-defined since $P_i$ and $P_{i,j}$ are not empty.
\item $\rho=(\rho_{i,j}:i\in I,j\in J_i)$ is a tuple of functions $\rho_{i,j}:P_{i,j}\rightarrow
  \{0,1\}$. For each $x\in P_{i,j}$, we have $\rho_{i,j}(x)=a_x$ if $\tau_x=\oe_j$ or $\tau_x=\oe_{j,j'}$ for some $j'>j$; $\rho_{i,j}(x)=b_x$ if $\tau_x=\oe_{j',j}$ for some $j'<j$.
\end{itemize}\end{flushleft}
Intuitively $I$ is the set of indices of terms with some string $x\in Q$ satisfying the term $T_i$
  as reported in $\sigma_x$, and $P_{i}$ is the set of such strings while $R_i$ is the set of strings which do not satisfy $T_i$.
For each $i\in I$, $J_i$ is the set of indices of clauses with some string $x\in P_i$ satisfying $T_i$ \emph{uniquely}
  and~falsifying the clause $C_{i,j}$. $P_{i,j}$ is the set of such strings, and $R_{i, j}$ is the set of strings which satisfy $T_i$ uniquely but also satisfy $C_{i, j}$.
We collect the following facts which are immediate from the definition.
\begin{fact}
\label{fact:1}
Let $(I;J;P;R;A;\rho)$ be the tuple induced by a map
  $\phi \colon Q \to \Sigma$. Then we have
\begin{itemize}
\item $|I|\le \mathop{\sum}_{i\in I} |P_i| \leq 2|Q|$.\vspace{-0.08cm}
\item For each $i\in I$, $|J_i|\le \mathop{\sum}_{j\in J_i} |P_{i,j}| \leq 2|P_i|$.\vspace{-0.08cm}
\item For each $i\in I$ and $j\in J_i$, $|R_i|$ and $|R_{i,j}|$ are at most $|Q|$ \emph{(}as
  they are subsets of $Q$\emph{)}.\vspace{-0.08cm}
\item For each $i\in I$ and $j\in J_i$, $P_{i,j}\subseteq P_i$,
  $A_{i, 0}\subseteq A_{i,j,0}$, and $A_{i, 1}\subseteq A_{i,j,1}$.
\end{itemize}
\end{fact}
Note that $|I|$ and $\sum_{i\in I} |J_i|$ can be strictly larger than $|Q|$, as some $x$ may satisfy more than one (but at most two) term  with $\sigma_{x} = e_{i, i'}$ and
  some $x$ may falsify more than one clause with $\tau_x=\oe_{j,j'}$.

The sets in $A$ are important for the following reasons that we summarize below.
\begin{fact}
\label{fact:a-terms-clauses}
Let $\phi:Q\rightarrow \frakP$ be the full signature map of $Q$ with respect to $(T,C,H)$. Then
\begin{itemize}
\item For each $i \in I$, $T_i(k) \in A_{i, 1}$ for all $k \in [\sqrt{n}]$ and
  $T_i(x)=0$ for each $x\in R_i$.
\vspace{-0.08cm}
\item For each $i\in I$ and $j \in J_i$, $C_{i,j}(k) \in A_{i, j, 0}$ for all $k \in [\sqrt{n}]$
  and $C_{i,j}(x)=1$ for each $x\in R_{i,j}$.
\end{itemize}
\end{fact}


Before moving back to the proof, we introduce the following consistency condition on $P$.

\begin{definition}
\label{def:consistent}
Let $(I;J;P;R;A;\rho)$ be the tuple induced by a map $\phi:Q\rightarrow \frakP$.
We say that~$P_{i,j}$ for some $i\in I$ and $j\in J_i$ is~\mbox{\emph{$1$-consistent}} if
  $\rho_{i,j}(x)=1$ for all $x\in P_{i,j}$,
  and \emph{$0$-consistent} if 
  $\rho_{i,j}(x)=0$ for all $x\in P_{i,j}$; otherwise we say $P_{i,j}$ is \emph{inconsistent}.
\end{definition}

Let $\phi$ be the full signature map of $Q$ with respect to $(T,C,H)$. 
If $P_{i,j}$ is $1$-consistent, the index $k$ of the
  variable $x_k$ in the dictatorship or anti-dictatorship function $h_{i,j}$
  must lie in $A_{i,j,0}$ (when~$h_{i,j}$ is an anti-dictator) or $A_{i,j,1}$
  (when $h_{i,j}$ is a dictator); the situation is similar if $P_{i,j}$ is $0$-consistent~but
  would be more complicated if $P_{i,j}$ is inconsistent.
Below we prune an edge whenever some $P_{i,j}$~in~$P$ becomes inconsistent.
This way we make sure that $P_{i,j}$'s in every leaf left are consistent.

\subsection{Tree pruning}

Consider an edge $(u, v)$ in  $B$. 
Let $\phi_{u} \colon Q \to \frakP$ and $\phi_v \colon Q \cup \{ x\} \to \frakP$
  be the maps associated with $u$ and $v$, with  $x$ being the query made at $u$ and 
  $\phi_v(x)$ being the label of $(u,v)$. Let
$(I; J;P;R;A;\rho)$ \text{and} $(I';J';P';R';A';\rho')$
be the two tuples induced by $\phi_u$ and $\phi_v$, respectively.

We list some easy facts about how $(I;J;P;R;A;\rho)$ is updated to obtain $(I';J';P';R';A';\rho')$.
\begin{fact}\begin{flushleft}
Let $\phi_v(x)=(\sigma_x, \tau_x, a_x, b_x)$ for the string $x$ queried at $u$. Then we have
\begin{itemize}
\item The new string $x$ is placed in $P_i'$ if $\sigma_{x,i}=1$, and is placed
  in $P_{i,j}'$ if $\sigma_x=e_i$ and $\tau_{x,j}=0$. \vspace{-0.09cm}
\item Each new set in $P'$ \emph{(}i.e., $P_i'$ with $i\notin I$ or $P_{i,j}'$ with either $i\notin I$
  or $i\in I$ but $j\notin J_i$\emph{)}, if any, is $\{x\}$ and the corresponding
  $A_{i,1}'$ or $A_{i,j,1}'$ is $\{k:x_k=1\}$ and $A_{i,0}'$ or $A_{i,j,0}'$ is $\{k:x_k=0\}$.\vspace{-0.08cm}
\item Each old set in $P'$ \emph{(}i.e., $P_i'$ with $i\in I$ or $P_{i,j}'$ with $i\in I$ and $j\in J_i$\emph{)}
  either stays the same or has $x$ being added to the set. For the latter case,
  $\{k:x_k=0\}$ is removed from $A_{i,1}$ or $A_{i,j,1}$  and
  $\{k:x_k=1\}$ is removed from $A_{i,0}$ or $A_{i,j,0}$ to obtain the new sets in $A'$.
\end{itemize}\end{flushleft}
\end{fact}

Now we are ready to define a set of so-called \emph{bad} edges
  of $B$, which will be used to prune $B$. In the rest of the proof we use $\alpha$ to denote a large
  enough positive constant.

\begin{definition}\label{typeAdef}
 An edge $(u,v)$ is called a \emph{bad} edge if at least one of the following events~\mbox{occur}~at $(u,v)$ and none of these events occur along the path from the root to $u$~\emph{(}letting $\phi_u$ and $\phi_v$ be the maps associated with $u$ and $v$, $x$ be the new query string at $u$, 
 $(I;J;P;R;A;\rho)$ and $(I';J';P';R';A';\rho')$~be the tuples that $\phi_u$ and $\phi_v$ induce, respectively\emph{)}:
\begin{flushleft} \begin{itemize}
        \item For some $i\in I$, $\big|A_{i, 1} \setminus A_{i, 1}'\big| \geq \alpha\sqrt{n} \log n$.\vspace{-0.08cm}
        \item For some $i\in I$ and $j\in J_i$, $\big|A_{i, j, 0} \setminus A_{i, j, 0}'\big| \geq \alpha\sqrt{n} \log n$.\vspace{-0.08cm}
        \item For some $i\in I$ and $j\in J_i$, $P_{i,j}$ is $0$-consistent but $P_{i,j}'$ is inconsistent
        \emph{(}meaning that\\ $x$ is added to $P_{i,j}$ with $\rho_{i,j}(y)=0$ for all $y\in P_{i,j}$ but
         $\rho'_{i,j}(x)=1$, instead of $0$\emph{)}.\vspace{-0.08cm}
    \item For some $i\in I$ and $j\in J_i$, $P_{i,j}$ is $1$-consistent but $P_{i,j}'$ is inconsistent \emph{(}meaning that\\
    $x$ is added to $P_{i,j}$ with $\rho_{i,j}(y)=1$ for all $y\in P_{i,j}$
    but $\rho_{i,j}'(x)=0$, instead of $1$\emph{)}.
     \end{itemize} \end{flushleft}
Moreover, a leaf $\ell$ is \emph{bad} 
  if one of the edges along the root-to-$\ell$ path is bad; $\ell$ is \emph{good} otherwise.
\end{definition}

The following pruning lemma states that the probability of $(\TT,\CC,\HH)$ reaching
  a bad leaf of $B$ is $o(1)$, when $(\TT,\CC)\sim \calE$ and $\HH\sim \Ey$.
We~delay the proof to Section \ref{proofcut}.


\begin{lemma}[Pruning Lemma]\label{badleaf}
$\mathop{\Pr}_{(\TT,\CC)\sim \calE,\HH\sim \Ey} \big[\text{$({\TT,\CC,\HH})$
  reaches a bad leaf of $B$\hspace{0.04cm}}\big]=o(1)$.
\end{lemma}

The pruning lemma allow us to focus on the good leaves $\ell$ of $B$ only. In particular we know~that along
  the root-to-$\ell$ path the sets $A_{i, 1}$ and $A_{i, j, 0}$ each
  cannot shrink by more than $\alpha\sqrt{n} \log n$ with a single query
  (otherwise the path contains a bad edge and $\ell$ is a bad leaf which we ignore).
Moreover every set $P_{i,j}$ in $P$ at the end must remain consistent (either $0$-consistent or $1$-consistent).

We use these properties to prove the following lemma in Section \ref{goodleaves2} for good leaves of $B$.



\begin{lemma}[Good Leaves are Nice]\label{goodleaves}
For each good leaf $\ell$ of $B$, we have
\[ \mathop{\Pr}_{(\TT, \CC) \sim \calE, \HH \sim \Ey}\Big[ (\TT, \CC, \HH) \text{ reaches } \ell
\hspace{0.03cm}\Big] \leq (1 + o(1))\cdot \mathop{\Pr}_{(\TT, \CC) \sim \calE, \HH \sim \En} \Big[ (\TT, \CC, \HH) \text{ reaches } \ell\hspace{0.03cm}\Big]. \]
\end{lemma}

We can now combine Lemma \ref{badleaf} and Lemma \ref{goodleaves} to prove Lemma \ref{maintechnical2}.

\begin{proof}[Proof of Lemma \ref{maintechnical2}]
Let $L$ be the leaves labeled ``accept,'' and $L^* \subset L$ be the good leaves labeled ``accept.''   Below we ignore
  $(\TT,\CC)\sim \calE$ in the subscript since it appears in every probability.
\begin{align*}
\mathop{\Pr}_{ \HH \sim \Ey}\Big[ \text{$B$ accepts\ }(\TT, \CC, \HH)
\Big]&=\sum_{\ell\in L}\hspace{0.08cm}\mathop{\Pr}_{ \HH \sim \Ey}\Big[ (\TT, \CC, \HH) \text{ reaches }\ell\hspace{0.03cm}\Big]\\[0.5ex]
&\le\sum_{\ell\in L^*}\mathop{\Pr}_{ \HH \sim \Ey}\Big[ (\TT, \CC, \HH)
\text{ reaches }\ell\hspace{0.03cm}\Big]
+o(1)\\[-0.3ex]
&\le (1+o(1))\cdot \sum_{\ell\in L^*}\mathop{\Pr}_{ \HH \sim \En}\Big[ (\TT, \CC, \HH)
\text{ reaches }\ell\hspace{0.03cm}\Big]
+o(1)\\[0.6ex]
&\le \mathop{\Pr}_{ \HH \sim \En}\Big[ \text{$B$ accepts $(\TT, \CC, \HH)$}\Big]
+o(1),
\end{align*}
where the second line used Lemma \ref{badleaf} and the third line used Lemma \ref{goodleaves}.
\end{proof}

\subsection{Proof of Lemma \ref{goodleaves} for good leaves}\label{goodleaves2}

We prove Lemma \ref{goodleaves} in this section.
Let $\ell$ be a good leaf associated with $\phi_\ell$ and $(I;J;P;R;A;\rho)$ be the tuple that $\phi_\ell$ induces.
Note that along the root-to-$\ell$ path, when a set $A_{i,0},A_{i,1},A_{i,j,0},A_{i,j,1}$ is created for the
  first time in $A$, its size is between $(n/2)\pm \sqrt{n}$ (since all queries made by $B$ lie in the middle
  layers).
As a result, it follows from Definition \ref{typeAdef} that for $i\in I$ and $j\in J_i$:
\begin{itemize}
\item[i)] $|A_{i, 1}|  \geq (n/2) - O(|P_i|\cdot \sqrt{n} \log n)$ and $|A_{i, j, 0}| \geq (n/2) -
  O(|P_{i, j}|\cdot \sqrt{n} \log n)$;\vspace{-0.05cm}

\item[ii)] $|A_{i, 0}|,\hspace{0.04cm}|A_{i,1}|,\hspace{0.04cm}|A_{i,j,0}|,\hspace{0.04cm}|A_{i, j, 1}| \leq (n/2) + \sqrt{n}$;\vspace{-0.05cm}
\item[iii)] $P_{i,j}$ is consistent (either $1$-consistent or $0$-consistent).
\end{itemize}

We start with the following claim:
\begin{claim}
\label{cl:bb}
For each $i \in I$ and $j \in J_i$, $|A_{i, j, 1}| \geq (n/2) - O\big(|P_{i, j}|^2\cdot \sqrt{n} \log n\big)$.
\end{claim}

\begin{proof}
For any two strings $x, y \in P_{i, j}$, we have
\[ \big|\{ k \in [n] : x_k = y_k = 0\}\big| \geq |A_{i, j, 0}| \geq ({n}/{2}) - O\big(|P_{i, j}|
  \cdot \sqrt{n} \log n\big).\]
As a result, it follows from $|\{k:y_k=0\}|\le (n/2)+\sqrt{n}$ 
  and $P_{i,j}$ being nonempty that  $$\big|\{k\in [n]:x_k=1,y_k=0\}\big|\le O\big(|P_{i, j}|
  \cdot \sqrt{n} \log n\big).$$ 
Finally we have 
\begin{equation}\label{papa1}
|A_{i,j,1}|\ge \big|\{k:x_k=1\}\big|-\sum_{y\in P_{i,j}\setminus \{x\}}
  \big|\{k:x_k=1,y_k=0\}\big|\ge (n/2)-O\big(|P_{i,j}|^2\cdot \sqrt{n}\log n\big).
\end{equation}
This finishes the proof of the lemma.
\end{proof}

Additionally, notice that $A_{i, 1} \subseteq A_{i, j, 1}$; thus from i) we have 
\begin{equation}\label{papa2}
|A_{i, j, 1}|\ge |A_{i, 1}|
  \ge (n/2)-O\big(|P_i|\cdot \sqrt{n}\log n\big).
\end{equation}
The following claim is an immediate consequence of this fact and Claim~\ref{cl:bb}.

\begin{claim}
\label{cl:diff}
For each $i \in I$ and $j \in J_i$, we have
\[ \big|\hspace{0.03cm}|A_{i, j, 1}| - |A_{i, j, 0}|\hspace{0.03cm} \big| \leq O\left(\sqrt{n} \log n \cdot \min\big\{ |P_{i, j}|^2, |P_i|\big\}\right) \]
\end{claim}
\begin{proof}
We have from i) and ii) that
$$
|A_{i,j,1}|-|A_{i,j,0}|\le (n/2)+\sqrt{n}-\big((n/2)-O\big(|P_{i,j}|\cdot \sqrt{n}\log n\big)\big)
=O\big(|P_{i,j}|\cdot \sqrt{n}\log n\big).
$$
On the other hand, from ii), (\ref{papa1}) and (\ref{papa2}), we have
$$|A_{i,j,0}|-|A_{i,j,1}|\le O\big(\sqrt{n}\log n\cdot \min\big\{|P_{i,j}|^2,|P_i|\big\}\big).
$$
Note that $|P_{i,j}|\le |P_i|$. The lemma then follows.
\end{proof}

We are now ready to prove Lemma \ref{goodleaves}.

\begin{proof}[Proof of Lemma~\ref{goodleaves}]
Let $\ell$ be a good leaf and 
  let $\phi:Q\rightarrow \frakP$ be the map associated with $\ell$.

Let $|\calE|$ denote the support size of $\calE$.
We may rewrite the two probabilities as follows: 
\begin{align*}
\mathop{\Pr}_{(\TT, \CC) \sim \calE, \HH \sim \Ey}\Big[ (\TT, \CC, \HH) \text{ reaches } \ell
\hspace{0.03cm}\Big] &= \dfrac{1}{|\calE|}\hspace{0.05cm} \sum_{(T, C) } \hspace{0.05cm}\mathop{\Pr}_{\HH \sim \Ey}\Big[(T, C, \HH) \text{ reaches } \ell\hspace{0.03cm}\Big] \\
\mathop{\Pr}_{(\TT, \CC) \sim \calE, \HH \sim \En} \Big[ (\TT, \CC, \HH) \text{ reaches } \ell
\hspace{0.03cm}\Big] &= \dfrac{1}{|\calE|}\hspace{0.05cm} \sum_{(T, C) } \hspace{0.05cm}\mathop{\Pr}_{\HH \sim \En}\Big[(T, C, \HH) \text{ reaches } \ell\hspace{0.03cm}\Big],
\end{align*}
where the sum is over the support of $\calE$.
Hence, it suffices to show that for each $(T, C)$ such that
\begin{equation}\label{ueue}
\mathop{\Pr}_{\HH \sim \Ey}\Big[(T, C, \HH) \text{ reaches } \ell\hspace{0.03cm}\Big]>0,
\end{equation}
we have the following inequality:
\begin{equation}\label{prpr} \dfrac{\mathop{\Pr}_{\HH \sim \En}\left[\hspace{0.01cm}(T, C, \HH) \text{ reaches } \ell\hspace{0.05cm} \right]}{\mathop{\Pr}_{\HH \sim \Ey}\left[\hspace{0.01cm}(T, C, \HH) \text{ reaches } \ell \hspace{0.05cm}\right]} \ge 1 - o(1). \end{equation}

Fix a pair $(T,C)$ such that (\ref{ueue}) holds.
Recall that $(T,C,H)$ reaches $\ell$ if and only if the signature of 
  each $x\in Q$ with respect to $(T,C,H)$ matches exactly $\phi(x)=(\sigma_x,\tau_x,a_x,b_x)$.
Given (\ref{ueue}), the~term and clause signatures of $x$ are already known to match $\sigma_x$ and $\tau_x$ (otherwise the LHS of (\ref{ueue}) is $0$).
The rest, i.e., $a_x$ and $b_x$ for each $x\in Q$, depends on $H=(h_{i,j})$ only. 

Since $\ell$ is consistent, there is a $\rho_{i,j}\in \{0,1\}$ for each $P_{i,j}$ such that 
  every $x \in P_{i, j}$ should satisfy $h_{i, j}(x) = \rho_{i, j}$. 
These are indeed the only conditions for~$H$ to match $a_x$ and $b_x$ for each $x\in Q$, and 
as a result, below we give the conditions on $H=(h_{i,j})$ for the triple $(T,C,H)$ to reach $\ell$:\vspace{0.06cm}
\begin{flushleft}\begin{itemize}
\item For $\Ey$, $(T, C, H)$ reaches $\ell$, where $H=(h_{i,j})$ and $h_{i,j}(x)=x_{k_{i,j}}$, if and only if $k_{i,j}\in A_{i, j, \rho_{i, j}}$ for each $i\in I$ and $j\in J_i$ (so that each $x \in P_{i, j}$ has $h_{i,j}(x) = \rho_{i, j}$).
\item For $\En$, $(T, C,H)$ reaches $\ell$, where $H=(h_{i,j})$ and $h_{i,j}(x)=\overline{x_{k_{i,j}}}$,
  if and only if $k_{i,j}\in A_{i, j, 1-\rho_{i, j}}$ for each $i\in I$ and $j\in J_i$ (so that
  each $x\in P_{i,j}$ has $h_{i,j}(x)=\rho_{i,j}$).\vspace{0.06cm}
\end{itemize}\end{flushleft}

With this characterization, we can rewrite the LHS of (\ref{prpr}) as follows:
\[ \dfrac{\mathop{\Pr}_{\HH \sim \En}\left[\hspace{0.01cm}(T, C, \HH) \text{ reaches } \ell \hspace{0.05cm}\right]}{\mathop{\Pr}_{\HH \sim \Ey}\left[\hspace{0.01cm}(T, C, \HH) \text{ reaches } \ell \hspace{0.05cm}\right]} = \prod_{i \in I, j \in J_i} \left( \dfrac{|A_{i, j, 1-\rho_{i, j}}|}{|A_{i, j,  \rho_{i, j}}|} \right) = \prod_{i \in I, j \in J_i} \left( 1 + \dfrac{|A_{i, j, 1-\rho_{i, j}}| - |A_{i, j, \rho_{i, j}}|}{|A_{i, j, \rho_{i, j}}|} \right). \]
Thus, applying Claim~\ref{cl:diff} and noting that $|A_{i, j, \rho_{i,j}}|\le n$
  (whether $\rho_{i,j}=0$ or $1$),
\begin{align*}
\dfrac{\mathop{\Pr}_{\HH \sim \En}\left[\hspace{0.01cm}(T, C, \HH) \text{ reaches } \ell \hspace{0.05cm}\right]}{\mathop{\Pr}_{\HH \sim \Ey}\left[\hspace{0.01cm}(T, C, \HH) \text{ reaches } \ell \hspace{0.05cm}\right]} &\ge \prod_{i \in I, j \in J_i} \left( 1-O\left( \frac{\log n \cdot \min\{|P_{i, j}|^2, |P_i|\}}{\sqrt{n}}\right) \right)  \\[0.5ex]
											&\ge 1- O\left(\frac{\log n}{\sqrt{n}}\right) \mathop{\sum}_{i \in I, j \in J_i} \min\big\{ |P_{i, j}|^2, |P_i|\big\} .
\end{align*}
As $\sum_{j} |P_{i,j}|\le 2|P_i|$,
  $\sum_{j\in J_i} \min\big\{|P_{i,j}|^2,|P_i|\big\}$ is maximized if $|J_i|=2\sqrt{|P_{i}|}$ 
  and $|P_{i,j}|=\sqrt{|P_i|}$. So
$$
\sum_{i\in I,j\in J_i} \min\big\{|P_{i,j}|^2,|P_i|\big\}
\le \sum_{i\in I} 2|P_i|^{3/2}\le O(q^{3/2}),
$$  
since $\sum_i |P_i|\le 2q$. This finishes the proof of the lemma since $q$ is chosen to be
  $n^{1/3}/\log^2 n$.
\end{proof}

\subsection{Proof of the pruning lemma}\label{proofcut}

Let $E$ be the set of bad edges as defined in Definition \ref{typeAdef}
  (recall that if $(u,v)$ is a bad edge, then the root-to-$u$ path cannot have any bad edge).
We split the proof of Lemma \ref{badleaf} into four lemmas, one lemma for each
  type of bad edges.
To this end, we define four sets $E_1,E_2,E_3$ and $E_4$ (we follow the same notation
  of Definition \ref{typeAdef}): An edge $(u,v)\in E$ belongs to \begin{enumerate}
\item $E_1$ if $|A_{i,1}\setminus A_{i,1}' |\ge 
  \alpha\sqrt{n}\log n$ for some $i\in I$;
\item $E_2$ if $ |A_{i,j,0}\setminus A_{i,j,0}' |\ge
  \alpha\sqrt{n}\log n$ for some $i\in I$ and $j\in J_i$;
\item $\smash{E_3}$ if it is not in $E_2$ and for some $i\in I$ and $j\in J_i$,
  $P_{i,j}$ is $0$-consistent but $\smash{P_{i,j}'}$ is\\ inconsistent (when $(u,v)\in E_3$ and the  above occurs, we say $(u,v)$ is $E_3$-bad at $(i,j)$);
\item $E_4$ if it is not in $E_1$ or $E_2$ and  for some $i\in I$ and $j\in J_i$,
  $P_{i,j}$ is $1$-consistent but $P_{i,j}'$ is\\ inconsistent (when $(u,v)\in E_4$ and the
  above occurs, we 
  say $(u,v)$ is $E_4$-bad at $(i,j)$).
\end{enumerate}  
It is clear that $E=E_1\cup E_2\cup E_3\cup E_4$.
\hspace{-0.04cm}(These four sets are not necessarily pairwise disjoint though we 
  did exclude edges of $E_2$ from $E_3$ and edges of $E_1$ and $E_2$ from $E_4$ explicitly.)
Each lemma below states that the probability of 
  $(\TT,\CC)\sim \calE$ and $\HH\sim \Ey$ taking an edge in $E_i$ is $o(1)$. 
Lemma \ref{badleaf} then follows directly from a union bound over the four sets.
  

\begin{lemma}\label{type1}
The probability of $(\TT,\CC)\sim \calE$ and $\HH\sim \Ey$ taking an edge in $E_1$ is $o(1)$.
\end{lemma}
\begin{proof}
Let $u$ be an internal node.
We prove that, when $(\TT,\CC)\sim \calE$
  and $\HH\sim \Ey$, either $(\TT,\CC,\HH)$ reaches node $u$ with probability $0$ or
\begin{equation}\label{goal1}
\mathop{\Pr}_{(\TT,\CC)\sim\calE,\HH\sim\Ey}\Big[\hspace{0.02cm}\text{$(\TT,\CC,\HH)$ takes an $E_1$-edge at $u$}\hspace{0.08cm}\Big|\hspace{0.08cm}
  \text{$(\TT,\CC,\HH)$ reaches $u$}\hspace{0.04cm}\Big]=o\hspace{0.02cm}(1/q).
\end{equation} 
Lemma~\ref{type1} follows  from Lemma \ref{simplepruning}. 
Below we assume that the probability of $(\TT,\CC,\HH)$ reaching node $u$ is positive. 
Let $\phi:Q\rightarrow \frakP$ be the map associated with $u$, and 
  let $x\in \{0,1\}^n$ be the string queried~at~$u$. 
\hspace{-0.04cm}Whenever we discuss a child node $v$ of $u$ below, we use
  $\phi'\hspace{-0.02cm}:\hspace{-0.02cm}Q\cup\{x\}\rightarrow\frakP$ to denote~the map associated with $v$ and
  $(I;J;P;R;A;\rho)$ and~$(I';J';P';R';A';\rho')$ to denote the tuples $\phi$ and $\phi'$ induce.
(Note that $v$ is not a specific node but can be any child of $u$.)

Fix an $i\in I$. We upperbound by $\smash{o(1/q^2)}$ the conditional probability of $(\TT,\CC, \HH)$  
  following~an edge $(u,v)$ with $\smash{|A_{i,1}\setminus A_{i,1}'|\ge \alpha\sqrt{n}\log n}$.
(\ref{goal1})  follows directly  from a union bound over $i\in I$.

With $i$ fixed, observe that any edge $(u,v)$ has either $A_{i,1}'=A_{i,1}$ or
  $A_{i,1}'=A_{i,1}\setminus \Delta_i$ with
$$
\Delta_i = \big\{\ell\in A_{i,1}:x_\ell=0\big\}\subseteq A_{i,1}.
$$
The latter occurs if and only if $P_i'=P_i\cup \{x\}$.
Therefore, we assume WLOG that $|\Delta_i|\ge \alpha\sqrt{n}\log n$ (otherwise the 
  conditional probability is $0$ for $i$), and now it suffices to upperbound by $o(1/q^2)$ 
  the conditional probability of $(\TT,\CC,\HH)$ taking an edge $(u,v)$
  with $P_i'=P_i\cup \{x\}$.

To analyze this conditional probability for $i\in I$,
  we fix a triple $(T_{-i},C,H)$, where~we use $T_{-i}$ to denote a sequence of $\smash{N}-1$
  terms with only the $i$th term missing, 
  such that 
  \[ \mathop{\Pr}_{\TT_i}\big[\hspace{0.01cm}\text{$((T_{-i},\TT_i),C,H)$ reaches $u$}\hspace{0.03cm}\big]>0,\]
where $\TT_i$ is a term drawn uniformly at random.
It suffices to prove for any such $(T_{-i},C,H)$:
\begin{align}
\label{prprpr} \mathop{\Pr}_{\TT_i}&\big[\hspace{0.01cm}\text{$((T_{-i},\TT_i),C,H)$ reaches $u$ and $P_i'=P_i\cup \{x\}$\hspace{0.03cm}}\big]
\\[-0.7ex]&\hspace{4cm}\le o\hspace{0.02cm}(1/q^2)\cdot \mathop{\Pr}_{\TT_i}\big[\hspace{0.01cm}\text{$((T_{-i},\TT_i),C,H)$ reaches $u$}\hspace{0.03cm}\big].\nonumber
\end{align}
Recalling Fact~\ref{fact:a-terms-clauses}, the latter event, $((T_{-i}, \TT_i), C, H)$ reaching $u$, imposes two conditions on $\TT_{i }$:
\begin{enumerate} 
\item For each $y \in P_i$, $\TT_{i }(y) = 1$, and\vspace{-0.18cm}
\item For each $z \in R_{i}$, $\TT_{i }(z)=0$.
\end{enumerate}
Let $U$ denote the set of all such terms $T:\sqrt{n}\rightarrow [n]$.
Then equivalently $T\in U$ if and only if 
\begin{equation*}\text{$U$: 
$T(k)\in A_{i,1}$ for all $k\in [\sqrt{n}]$ and  each $z\in R_i$ has $z_{T(k)}=0$ for some $k\in [\sqrt{n}]$.}\end{equation*}
Regarding the former event in (\ref{prprpr}), i.e. $((T_{-i}, \TT_i), C, H)$ reaching $u$ and $P_i' = P_i \cup \{ x\}$, a necessary condition over $\TT_i$ is 
  the same as above but in addition we require $\TT_i(x)=1$.
(Note that this is not a sufficient condition since for that we also need 
  $\TT_i$ to be one of the first two terms that $x$ satisfies, which depends on $T_{-i}$.)
Let $V$ denote the set of all such terms. Then $T\in V$ if  
$$
\text{$V$: 
$T(k)\in A_{i,1}\setminus \Delta_i$ for all $k\in [\sqrt{n}]$ and  each $z\in R_i$ has $z_{T(k)}=0$ for some $k\in [\sqrt{n}]$.}
$$

In the rest of the proof we prove that $|V|\hspace{0.03cm}/\hspace{0.03cm}|U|=o(1/q^2)$, from which (\ref{prprpr}) follows.
Let $\ell=\log n$. First we write $U'$~to denote the following subset of $U$: $T'\in U$ is in $U'$ if
$$
\big|\{ k \in [\sqrt{n}] : T'(k) \in \Delta_i \}\big| = \ell,
$$ 
and it suffices to show $|V|/|U'| = o(1/q^2)$. 
Next we define the following bipartite graph $G$ between $U'$ and $V$ (inspired by
  similar arguments of \cite{BB15}): 
  $T'\in U'$ and $T\in V$ have~an edge if and only~if 
  $T'(k)=T(k)$ for all $k\in [\sqrt{n}]$ with $T'(k)\notin \Delta_i$.
Each $T'\in U'$ has degree at most $|A_{i,1}\setminus \Delta_i|^\ell$,
  as one can only move each $T'(k)\in \Delta_i$ to $A_{i,1}\setminus \Delta_i$.

To lowerbound the degree of a $T\in V$, note that one only needs at most
  $q$ many variables of $T$ to kill all strings in $R_i$.
Let $H\subset [\sqrt{n}]$ be any set of size at most $q$ such that 
  for each string $z\in R_i$, there exists a $k \in H$ with $z_{T(k)} = 0$.\hspace{0.04cm}\footnote{For example, since $|R_i| \leq q$, 
  one can set $H$ to contain the smallest $k \in [\sqrt{n}]$ such that $z_{T(k)} = 0$, for each $z \in R_i$.}
Then one can choose any $\ell$ distinct indices $k_1, \dots, k_{\ell}$ from~$\overline{H}$, as well as any $\ell$
  (not necessarily distinct) variables $t_1, \dots, t_{\ell}$ from $\Delta_i$, and let $T'$ be 
a term where 
\[ T'(k) = \left\{ \begin{array}{ll} 
						t_i & k = k_i\ \text{for some $i \in [\ell]$} \\[0.5ex]
						T(k) & \text{otherwise.}
						\end{array} \right. 
\]
The resulting $T'$ is in $U'$ and $(T,T')$ is an edge in $G$.
As a result, the degree of $T\in V$ is at least
$$
{\sqrt{n}-q\choose \ell}\cdot |\Delta_i|^{\ell}.
$$
By counting the number of edges in $G$ in two different ways and
  using $|A_{i,1}|\le (n/2)+\sqrt{n}$, $$
\frac{|U'|}{|V|}\ge {\sqrt{n}-q\choose \ell}\cdot \left(\frac{|\Delta_i|}{|A_{i,1}\setminus \Delta_i|}
\right)^{\ell}
\ge \left( \frac{\sqrt{n}}{2\ell} \cdot \frac{\alpha\sqrt{n}\hspace{0.04cm}\ell}{(n/2)+\sqrt{n}}\right)^{\ell}
>\omega(q^2),
$$
by choosing a large enough constant $\alpha>0$. This finishes the proof of the lemma.
\end{proof}

\begin{lemma}\label{type2}
The probability of $(\TT,\CC)\sim\calE$ and $\HH\sim \Ey$ taking an edge in $E_2$ is $o(1)$.
\end{lemma}
\begin{proof}
The proof of this lemma is similar to that of Lemma \ref{type1}.
Let $u$ be any internal node of the tree. We prove that, when $(\TT,\CC)\sim \calE,\hspace{0.03cm}\HH\sim \Ey$,
  either $(\TT,\CC,\HH)$ reaches $u$ with probability $0$ or 
\begin{equation}\label{goal2}
\mathop{\Pr}_{(\TT,\CC)\sim\calE,\HH\sim\Ey}\Big[\hspace{0.02cm}\text{$(\TT,\CC,\HH)$ takes an $E_2$-edge at $u$}\hspace{0.08cm}\Big|\hspace{0.08cm}
  \text{$(\TT,\CC,\HH)$ reaches $u$}\hspace{0.04cm}\Big]=o\hspace{0.02cm}(1/q).
\end{equation}
Assume below WLOG that the probability of $(\TT,\CC,\HH)$ reaching $u$ is positive.

Fix $i\in I$ and $j\in J_i$. \hspace{-0.03cm}We upperbound  the conditional probability of $(\TT,\CC,\HH)$ taking an edge $(u,v)$
  with $|A_{i,j,0}\setminus A_{i,j,0}'|\ge \alpha\sqrt{n}\log n$ by $o(1/q^3)$. (\ref{goal2})
 follows by a union bound.
Similarly let 
\begin{equation}\label{juju}
\Delta_{i,j}=\big\{\ell\in A_{i,j,0}:x_\ell=1\big\}\subseteq A_{i,j,0},
\end{equation}
and assume WLOG that $|\Delta_{i,j}|\ge \alpha\sqrt{n}\log n$ (as otherwise the conditional
  probability is~$0$ for $i,j$).
Then it suffices to upperbound  the conditional probability of $(\TT,\CC,\HH)$
  going along an~edge $(u,v)$ with $P_{i,j}'=P_{i,j}\cup \{x\}$ by $o(1/q^3)$.
The rest of the proof is symmetric to that of Lemma \ref{type1}.
\end{proof}


\begin{lemma}\label{type3}
The probability of $(\TT,\CC)\sim\calE$ and $\HH\sim \Ey$ taking an edge in $E_3$ is $o(1)$.
\end{lemma}
\begin{proof}
We fix any pair $(T,C)$ from the support of $\calE$ and prove that
\begin{equation}\label{goal31}
\mathop{\Pr}_{\HH\sim \Ey} \big[\hspace{0.01cm}\text{$(T,C,\HH)$ takes an $E_3$-edge}
\hspace{0.04cm}\big]=o\hspace{0.02cm}(1).
\end{equation}
The lemma follows by averaging (\ref{goal31}) over all pairs $(T,C)$ in the support of $\calE$.
To prove (\ref{goal31}) we~fix any internal node $u$ such that
  the probability of $\smash{(T,C,\HH)}$ reaching $u$ is positive, and prove that
\begin{equation}\label{goal32}
\mathop{\Pr}_{\HH\sim \Ey}\Big[\hspace{0.02cm}\text{$(T,C,\HH)$ takes an $E_3$-edge leaving $u$}\hspace{0.08cm}\Big|\hspace{0.08cm}
  \text{$(T,C,\HH)$ reaches $u$}\hspace{0.04cm}\Big]=o\hspace{0.02cm}(1/q).
\end{equation}
(\ref{goal31}) follows by Lemma \ref{simplepruning}.
Below we assume  the probability of $(T,C,\HH)$ reaching $u$ is positive.

We assume WLOG  that there is no edge in $E$ along the root-to-$u$ path; 
   otherwise, \hspace{-0.03cm}(\ref{goal32}) is $0$.~We  follow the same notation used in the proof of Lemma \ref{type1}, i.e., $\phi_u:Q\rightarrow \frakP$
  as the~map associated with $u$, $x$ as the query made at $u$, 
  and $(I;J;P;R;A;\rho)$ 
    as the tuple induced by $\phi_u$.
We also write~$F$ to denote the set of pairs $(i,j)$ such that $i\in I$ and $j\in J$.


Observe that since $(T,C)$ is fixed, the term and clause signatures of every string are fixed, and in particular 
  the term and clause signatures (denoted $\sigma_x$ and $\tau_x$) of $x$ are fixed.
We assume WLOG that $\sigma_x=e_{k}$ for some $k\in [N]$
  (otherwise $x$ will never be added to any $P_{i,j}$ when $(T,C,\HH)$ leaves~$u$ 
  and (\ref{goal32}) is $0$ by the definition of $E_3$).
In this case we write $D$ to denote the set of $\{(k,j):\tau_{x,j}=0\}$ 
  with $|D |\le 2$.
As a result, whenever $(T,C,\HH)$ takes an $E_3$-edge leaving from $u$,
  this edge must be $E_3$-bad at one of the pairs $(k,j)\in D$.
Thus, the LHS of (\ref{goal32}) is the same as
\begin{equation}\label{ofofp2}
\sum_{(k,j)\in D} \mathop{\Pr}_{\HH\sim \Ey}\Big[\hspace{0.02cm}\text{$(T,C,\HH)$ takes
  a $(u,v)$ that is $E_3$-bad at $(k,j)$}\hspace{0.08cm}\Big|\hspace{0.08cm}
  \text{$(T,C,\HH)$ reaches $u$}\hspace{0.04cm}\Big].
\end{equation}  
  
To bound the conditional probability for $(k,j)$ above by $o(1/q)$, we  
   assume WLOG that $(k,j)\in$ $F$ (otherwise $x$ would create a new $P_{k,j}$ whenever
   $(T,C,\HH)$ takes an edge
  $(u,v)$ leaving $u$,~and~the latter cannot be $E_3$-bad at $(k,j)$).  
Next we define ($A_{k,j,0}$ below is well defined since $(k,j)\in F$)
$$
\Delta_{k,j}=\big\{\ell\in A_{k,j,0}: x_\ell=1\big\}.
$$
We may assume WLOG that $|\Delta_{k,j}|<\alpha\sqrt{n}\log n$; otherwise $(T,C,\HH)$
  can never take an edge $(u,v)$ in $E_3$ because $E_2$-edges are
  explicitly excluded from $E_3$.
Finally, we assume WLOG~$\rho_{k,j}(y)=0$ for all $y\in P_{k,j}$;
  otherwise the edge $(u,v)$ that $(T,C,\HH)$ takes can never be  $E_3$-bad at $(k,j)$.  

With all these assumptions on $(k,j)$ in place, we prove the following inequality:
\begin{align}\label{goal33}
 \mathop{\Pr}_{\HH\sim\Ey}&\Big[\hspace{0.02cm}\text{$(T,C,\HH)$ takes a $(u,v)$ that is $E_3$-bad at $(k,j)$}\hspace{0.04cm}\Big]\\[-0.8ex] &\hspace{4cm}\le \frac{|\Delta_{k,j}|}{|A_{k,j,0}|}\cdot \mathop{\Pr}_{\HH\sim\Ey}\Big[\hspace{0.02cm}
  \text{$(T,C,\HH)$ reaches $u$}\hspace{0.04cm}\Big]. \nonumber
\end{align} 
Given $|\Delta_{k,j}|= O(\sqrt{n} \log n)$ and $|A_{i,j,0}|\ge (n/2)-O(q\sqrt{n}\log n)=\Omega(n)$
  (since there is no bad edge particularly no $E_2$-edge, from the root to $u$),
  (\ref{goal32}) follows by summing over $D$, with $|D|\le 2$.

We work on (\ref{goal33}) in the rest of the proof.
Fix any tuple $H_{-(k,j)}$ (with its $(k,j)$th entry missing) such that the probability of $(T,C,(H_{-(k,j)},\hh))$ reaching~$u$~is positive, where $\hh$ is a random dictator function with its dictator variable
    drawn from $[n]$ uniformly.
Then (\ref{goal33}) follows from  
\begin{align}\label{goal34}
 \mathop{\Pr}_{\hh }&\Big[\hspace{0.02cm}\text{$(T,C,(H_{-(k,j)},\hh ))$ takes $(u,v)$ that is $E_3$-bad at $(k,j)$}\hspace{0.03cm}\Big] \\[-0.5ex] &\hspace{5cm} \le \frac{|\Delta_{k,j}|}{|A_{k,j,0}|}\cdot 
 \mathop{\Pr}_{\hh }
  \Big[\hspace{0.01cm}\text{$(T,C,(H_{-(k,j)},\hh))$ reaches $u$}\hspace{0.03cm}\Big] .\nonumber
\end{align} 
The event on the RHS, i.e., that $\smash{(T,C,(H_{-(k,j)},\hh))}$ reaches $u$, imposes the following condition
  on~$r$ the dictator variable of $\hh$: $r\in A_{k, j, 0}$,
  since $\rho_{k,j}(y)=0$ for all $y\in P_{k,j}$.
Hence the probability on the RHS of (\ref{goal34}) is $|A_{i,j,0}|\hspace{0.03cm}/\hspace{0.03cm}n$.
On the other hand, the event on the LHS of (\ref{goal34}), that $\smash{(T,C,(H_{-(i,j)},\hh ))}$ follows
  a $(u,v)$ that is $E_3$-bad at $(k,j)$, imposes the following necessary condition for $r$:
  $r\in \Delta_{k,j}$.\hspace{0.06cm}\footnote{Note that this is \emph{not} a sufficient condition, because the other pair $(k,j')\in D$ may have $|\Delta_{k,j'}|\ge \alpha\sqrt{n}\log n$.}
  As a result, the probability on the LHS of (\ref{goal34}) is at most  
$|\Delta_{k,j}| /{n}$. 
(\ref{goal34}) then follows.
\def\hh{\boldsymbol{h}}
\end{proof}


\begin{lemma}\label{typeB}
The probability of $(\TT,\CC)\sim \calE$ and $\HH\sim\Ey$ taking an edge in $E_4$ is $o(1)$.
\end{lemma}
\begin{proof}
We fix a pair $(T,C)$ from the support of $\calE$ and prove that
\begin{equation}\label{goal41}
\mathop{\Pr}_{\HH\sim \Ey} \big[\hspace{0.01cm}\text{$(T,C,\HH)$ takes an $E_4$-edge}
\hspace{0.06cm}\big]=o\hspace{0.02cm}(1).
\end{equation}
The lemma follows by averaging (\ref{goal41}) over all $(T,C)$ in the support of $\calE$.
To prove (\ref{goal41}), fix a leaf $\ell$ such that the probability of 
  $(T,C,\HH)$ reaching $\ell$ is positive. 
Let $u_1 \cdots u_{t'}u_{t'+1}=\ell$ be the~root-to-$\ell$ path 
  and let $q(u_s)$ denote the following conditional probability:
\begin{equation*}  
 \mathop{\Pr}_{\HH\sim\Ey}\Big[\hspace{0.03cm} (T, C, \HH) \text{ takes an $E_4$-edge leaving $u_s$}\hspace{0.08cm} \Big|\hspace{0.08cm} (T, C, \HH) \text{ reaches } u_s\hspace{0.03cm}\Big]. 
 \vspace{-0.1cm} \end{equation*}
It then suffices to show for every such leaf $\ell$,
\begin{equation}\label{goalgoal}
\sum_{s\in [t']} q(u_s)=o(1),
\end{equation}
  since (\ref{goal41}) would then follow by Lemma \ref{complicatedpruning}.
To prove (\ref{goalgoal}), we use $t$ to denote the smallest integer such that 
  $(u_{t+1},u_{t+2})$ is an edge in $E_1$ or $E_2$ 
  with $t=t'$ by default if there is no such edge along~the path.
By the choice of $t$,   there is no edge in $E_1$ or $E_2$ along 
  $u_1\cdots u_{t+1}$.
For (\ref{goalgoal}) it suffices to  show 
\begin{equation}\label{goal44}
\sum_{s\in [t]} q(u_s)=o(1).
\end{equation}
To see this we consider two cases. If there is no $E_1,E_2$ edge along the root-to-$\ell$ path,
  then~the~two sums in (\ref{goalgoal}) and (\ref{goal44}) are the same.
If $(u_{t+1},u_{t+2})$ is an edge in $E_1$ or $E_2$, then 
  $q(u_s)=0$ if $s\ge t+2$ (since $(u,v)\notin E$ if there is already an edge in $E$ along the path to $u$).
We claim that $q(u_{t+1})$ must be $0$ as well. This is because, given
  that $(T,C)$ is fixed and that $(T,C,\HH)$ takes $(u_{t+1},u_{t+2})$ 
  with~a positive probability, whenever $(T,C,\HH)$ follows an edge $(u_{t+1},v)$ from $u_{t+1}$,
  $v$ has the same term and clause signatures $(\sigma_x, \tau_x)$ as $u_{t+2}$ and thus, also has the same $P$ and $A$ (as part of the tuple its map induces). 
As a~result $(u_{t+1},v)$ is also in $E_1$ or $E_2$ and cannot be an edge in $E_4$
  (recall that we explicitly excluded $E_1$~and~$E_2$ from $E_4$).
Below we focus on $u_s$ with $s\in [t]$ and upperbound $q(u_s)$.

For each $s\in [t]$ we
  write $\smash{x^{s}}$ to denote 
  the string queried at $u_s$ and 
  let $\smash{(I^{s};J^{s};P^{s};Q^{s};R^{s}; \rho^{s})}$ be the~tuple induced by the map associated with $u_s$.
  \hspace{-0.04cm}We also write $F_s$ to denote the set of pairs~$(i,j)$ with $i\in I^{s},\smash{j\in J^{s}_i}$.
Following the same arguments used to derive (\ref{ofofp2}) in the proof of Lemma \ref{type3},
  let $D_s\subseteq F_{s}$ denote the set of at most two pairs $(i,j)$ such that $x^{s}$ is added 
  to $\smash{P^{s}_{i,j}}$ when~$(T,C,\HH)$ reaches $u_s$. 
Note that if $x^{s}$ just creates a new $P_{i,j}$ (so $(i,j)\notin F_s$), we do not include   it in  
  $D_s$.  
As a result,
whenever $(T,C,\HH)$  takes an $E_4$-edge $(u,v)$, the latter must be $E_4$-bad at 
  one of $(i,j)\in D_s$.

Next for each pair $(i,j)\in D_s$,  
we can follow the analysis of (\ref{goal33}) to show that
\begin{align*} 
 \mathop{\Pr}_{\HH\sim\Ey}&\Big[\hspace{0.02cm}\text{$(T,C,\HH)$ takes a $(u,v)$ that is $E_4$-bad at $(i,j)$}\hspace{0.03cm}\Big] \le \dfrac{\big|\Delta^{s}_{i,j}\big|}{\big|A_{i,j,1}^{s}\big|}\cdot 
 \mathop{\Pr}_{\HH\sim\Ey }
  \Big[\hspace{0.01cm}\text{$(T,C,\HH)$ reaches $u$}\hspace{0.03cm}\Big], 
\end{align*}
where the set $\Delta^{s}_{i,j}$ is defined as
$$\Delta^{s}_{i,j}=\Big\{k\in A_{i,j,1}^{s}: x_{k}^{s}
  =0\Big\}.$$
As there is no $E_1$ or $E_2$ edge along the path to $u_s$,
  we have by (\ref{papa2}) that $A_{i,j,1}^{s}$ has size $\Omega(n)$. Thus,
\begin{equation}\label{ifif}
q(u_s)\le O\hspace{0.02cm}(1/n)\cdot \sum_{(i,j)\in D_s} {\left|\Delta_{i,j}^{s}\right|}\quad\text{\ and\ }\quad
\sum_{s\in [t]} 
q(u_s)\le O\hspace{0.02cm}(1/n)\cdot \sum_{s\in [t]}\sum_{(i,j)\in D_s} {\left|\Delta_{i,j}^{s}\right|}. 
\end{equation}

Let $(I^*;J^*;P^*;R^*;A^*;\rho^*)$ be the tuple induced by the map associated with $u_{t+1}$ 
  and let $F^*$ be the set of $(i,j)$ with $i\in I^*$ and $j\in J_i^*$.
We upperbound the second sum in (\ref{ifif}) above by focusing on any fixed pair $(i,j)\in F^*$ and observing that
$$
\sum_{s:(i,j)\in D_s} \left|\Delta_{i,j}^{s}\right|+\big|A^*_{i,j,1}\big|\le (n/2)+\sqrt{n}.
$$
This is because $\smash{\Delta_{i,j}^{s}}$ and $A_{i,j,1}^*$ are pairwise disjoint and their union 
  is indeed exactly the number of $1$-entries of the query string along the path that first creates $P_{i,j}$.
The latter is at most $(n/2)+\sqrt{n}$ because we assumed that strings queried in the tree lie  
  in the middle layers. On the other hand,
$$
\big|A_{i,j,1}^*\big|\ge (n/2)-O\big(\sqrt{n}\log n\cdot \min\big\{|P_{i,j}^*|^2,|P_i^*|\big\}\big).
$$
This follows directly from (\ref{papa1}) and (\ref{papa2}) and our choice of $t$ at the beginning of the
  proof so that 
  there is no $E_1$ or $E_2$ edge from $u_1$ to $u_{t+1}$.
We finish the proof by plugging the two inequalities into (\ref{ifif}) and follow the same arguments used at the end of 
  the proof of the lemma for good leaves.
\end{proof}

\section{Unateness Lower Bound}\label{sec:unate}

\def\dyes{\mathcal{D}_{\text{yes}}}
\def\Dyes{\dyes}
\def\Dno{\dno}
\def\boldr{\rr}
\def\calEn{\En}
\def\calEy{\Ey}

We start with some notation for strings.
Given $A\subseteq [n]$ and $x\in \{0,1\}^n$, we use $x_A$ to denote the
  string in $\{0,1\}^A$ that agrees with $x$ over $A$.
Given $y\in \{0,1\}^A$ and $\smash{z\in \{0,1\}^{\overline{A}}}$,
  we use $x=y\circ z$~(as their concatenation) to denote the string $x\in \{0,1\}^n$ that agrees with
  $y$ over $A$ and $z$ over $\overline{A}$.
Given $x\in \{0,1\}^n$ and $y\in \{0,1\}^A$ with $A\subseteq [n]$, we use $x\oplus y$ to denote 
  the $n$-bit string $x'$ with $x_i'=x_i$ for all $i\notin A$ and $x_i'=x_i\oplus y_i$ for all $i\in A$,
  i.e., $x'$ is obtained from $x$ by an XOR with $y$ over $A$.

\subsection{Distributions}\label{sec:unatedist}

For a fixed $n > 0$ we describe a pair of distributions, $\Dyes$ and $\Dno$, supported on Boolean functions $f \colon \{0, 1\}^n \to \{0, 1\}$ that will be used to obtain a two-sided and adaptive lower bound for unateness testing. After defining the distributions, we show in this subsection that any $\boldf \sim \Dyes$ is unate, and $\boldf \sim \Dno$ is $\Omega(1)$-far from being unate with probability $\Omega(1)$. 
Let $N$ be the following parameter:
$$N = \left(1 + \frac{1}{\sqrt{n}}\right)^{n/4}\approx e^{\sqrt{n}/4}.$$ 

A function $\boldf \sim \Dyes$ is drawn using the following procedure:
\begin{flushleft}\begin{enumerate}
\item Sample a subset $\bM \subset [n]$ uniformly at random from all subsets of size $n/{2}$. 
\item Sample $\bT \sim \calE(\bM)$ (which we describe next). $\bT$ is a sequence of terms $(\bT_i : i \in [N])$. $\bT$ is then 
  used to define a multiplexer map $\bGamma = \bGamma_{\bT} \colon \{0, 1\}^n \to [N] \cup \{ 0^*, 1^* \}$. 
\item Sample $\bH \sim \Ey(\bM)$ where $\bH = ( \bh_i : i \in [N])$. For each $i \in [N]$, $\bh_i \colon \{0, 1\}^n \to \{0, 1\}$ is a dictatorship function $\bh_i(x) = x_k$ with $k$ sampled independently and uniformly  from $\overline{\bM}$.
We will refer to $\bh_i$ as the dictatorship function 
  and $x_k$ (or simply its index $k$) as the \emph{special} variable associated with the $i$th term $\bT_i$.
  
\item Sample two strings $\br\in \{0,1\}^\bM$ and $\bs\in \{0,1\}^\bM$ 
  uniformly at random. Finally, the function
  $\boldf = \boldf_{\bM,\bT, \bH, \boldr,\bs} \colon \{0, 1\}^n \to \{0, 1\}$ is defined as follows:
\begin{align*}
\boldf_{\bM,\bT,\bH,\boldr,\bs}(x) =f_{\bM, \bT,\bH}\big(x\oplus (\boldr\circ\bs)\big), 
\end{align*}
where $\boldf_{\bM, \bT,\bH}$ is defined as follows (with the truncation done first):
\begin{align*}
\boldf_{\bM, \bT, \bH}(x) &= \left\{ \begin{array}{ll} 0 & \text{if $|x_\bM| < (n/4) - \sqrt{n}$}\\[0.7ex]
1& \text{if $|x_{\bM}| > (n/4) + \sqrt{n}$}\\[0.7ex]
0 & \text{if } \bGamma(x) = 0^* \\[0.7ex]
						1 & \text{if }\bGamma(x) = 1^* \\[0.7ex]
						h_{\bGamma(x)}(x) & \text{otherwise (i.e., when $\bGamma(x) \in [N]$)} \end{array} \right. 
\end{align*}
This finishes the definition of our yes-distribution $\Dyes$.
\end{enumerate}\end{flushleft}

A function $\boldf = \boldf_{\bM, \bT,\bH, \boldr,\bs} \sim \Dno$ is drawn using a similar procedure, with the only difference being that $\bH=(\bh_i : i \in [N])$ is sampled from $\En(\bM)$ instead of $\Ey(\bM)$:
  each $\bh_i$ is a dictatorship function 
  $\bh_i(x)=x_k$ with probability $1/2$ and an
  anti-dictatorship $\bh_i(x)=\overline{x_k}$ with probability $1/2$, 
where $k$ is chosen independently and uniformly at random from $\overline{\bM}$. 
We will also refer to $\bh_i$ as the dictatorship or anti-dictatorship function and 
  $x_k$ as the special variable associated with $\bT_i$.

\begin{remark} 
Note that the truncation in $\boldf_{\bM, \bT,\bH,\boldr,\bs}$ is done after sampling $\boldr$. 
As a result, we may \emph{not} assume all queries are made in the middle layers, like we did in Section~\ref{sec:mono}. 
\end{remark}

Fixing an $M\subset [n]$ of size $n/2$,
  we now describe $\bT\sim \calE(M)$ to finish the description of the~two distributions.
Each term $\bT_i$ in $\bT$, $i\in [N]$, is drawn independently and is a 
  random \emph{subset} of $M$ with each $j\in M$ included with probability $1/\sqrt{n}$ independently.
We also abuse the notation and 
  interpret each term $\bT_i$ as a Boolean function that is the conjunction of its variables:
\[ \bT_i(x) = \bigwedge_{j\in \bT_i}  x_{j}. \]
Note that, for some technical reason that will become clear later in the proof of 
  Lemma \ref{lem:prunebad},~the definition of terms here is slightly different from
  that used in the monotonicity lower bound, though both are the conjunction of 
  roughly $\sqrt{n}/2$ ($\sqrt{n}$ in monotonicity) variables.
Given $\bT$, the multiplexer map $\bGamma_{\bT} \colon \{0, 1\}^n \to [N] \cup \{0^*, 1^* \}$  {indicates the index of the term $\bT_i$ that is satisfied by $x$, if there is a unique one; it returns $0^*$ if no term is satisfied, or $1^*$ if more than one term are satisfied: }
\[ \bGamma_{\bT}(x) = \left\{ \begin{array}{ll} 0^* & \forall\hspace{0.06cm} i \in [N], 
\hspace{0.08cm}\bT_i(x) = 0 \\[0.8ex]
								      1^* & \exists\hspace{0.06cm} i \neq j \in [N],\hspace{0.06cm} \bT_i(x) = \bT_j(x) = 1 \\[0.8ex]
								      i & \bT_i(x) = 1\ \text{for a unique $i\in [N]$} \end{array} \right. \]

We give some intuition for the reason why the two distributions are hard to distinguish and
  can be used to obtain a much better lower bound for unateness testing, despite of being
  much simpler than the two-level construction used in the previous section.
Note that $\Dyes$ and $\Dno$ are exactly the same except that (1) in $\Dyes$, $\bh_i$'s are random 
  dictatorship or anti-dictatorship functions (if one takes $\bs$ into consideration)
  but are \emph{consistent} in the sense that all $\bh_i$'s with the same special variable $x_k$
  are either all dictatorship or anti-dictatorship functions; (2)
in contrast, whether $\bh_i$~is a dictatorship or anti-dictatorship is independent for each $i\in [N]$
  in $\Dno$.
Informally, the only~way for an algorithm to be sure that $f$ is from $\Dno$ (instead of $\Dyes$)
  is to find two terms  with the same special variable $x_k$ but one 
  with a dictatorship and the other with an anti-dictatorship function over $x_k$.
As a result, one can interpret our $\tilde{\Omega}(n^{2/3})$ lower bound (at a high level) as the product of two
  quantities: the number of queries one needs to \emph{breach} a term $\bT_i$ 
  (see Section \ref{sec:unate-oracle} for details) and find 
  its special variable, and the number of terms one needs to breach in order to find two 
  with the same special variable.
This is different from monotonicity testing since we are done once a term is breached there,
  and enables us to obtain a much better lower bound for unateness testing.
  
Next we prove that $\boldf\sim \Dyes$ is unate and $\boldf\sim \Dno$ is far from unate with high probability.  

\begin{lemma}
\label{lem:unate}
Every $f$ in the support of $\Dyes$ is unate.
\end{lemma}

\begin{proof}
Given the definition of $f = f_{M, T, H, r,s}$ using $f_{M,T,H}$,
  it suffices to show  that $f_{M, T, H}$ is monotone. The rest of the proof is similar to that of Lemma~\ref{monotone:lem}. 
\end{proof}

\def\ff{\boldf}
\def\TT{\bT}
\def\HH{\bH}
\def\XX{\bX}

\begin{lemma}
\label{lem:nonunate}
A function $\ff \sim \Dno$ is $\Omega(1)$-far from unate with probability $\Omega(1)$. 
\end{lemma}
\begin{proof}
Consider a fixed subset $M \subset [n]$ of size $n/2$. It suffices to prove that, when $\TT\sim \calE(M)$ and $\HH\sim \En(M)$, 
  the function $\ff=\ff_{M, \TT,\HH}$ is $\Omega(1)$-far from unate.
This is due to the fact that flipping variables~of a function as we do using $\br$ and $\bs$ 
  does not change its distance to unateness.


Fix $T$ in the support of $\calE(M)$ and $H$ in the support of $\En(M)$.
We let $X \subset \{0, 1\}^n$ denote~the set of $x\in \{0,1\}^n$ in the middle layers
  (i.e. $|x_M|$ is within $n/4\pm \sqrt{n}$) such that $\Gamma_T(x)=i$~for~some $i\in [N]$ (rather than $0^*$ or $1^*$). For each $x \in X$ with $\Gamma_T(x)=i$, 
 {we also let $\rho(x)=k$ be the~special variable associated with $T_i$ 
 (i.e., $h_i(x) = x_k$ or $h_i(x) = \overline{x_k}$).}
As $\rho(x) \in \overline{M}$ and $\Gamma_T(x)$ depends only on variables in $M$, we have that  
$$ 
\Gamma_T\big(x^{(\rho(x))}\big)=\Gamma_T(x),
$$ 
i.e., after flipping the $\rho(x)$th bit of $x$, the new string
  still satisfies uniquely the same term as $x$.
  
Let $\smash{x^*=x^{(\rho(x))}}$ for each string $x\in X$ (then $(x^*)^*=x$).
The claim below shows that $\smash{(x,x^*)}$ is a bi-chromatic edge along the $\rho(x)$th direction. As a result,
  one can decompose $|X|$ into $|X|/2$ many \emph{disjoint} bi-chromatic edges $(x,x^*)$.

\begin{claim}
For all $x \in X$, $(x, x^*)$ is a bi-chromatic edge of $f_{M,T, H}$. 
\end{claim}
\begin{proof}
Let  $k=\rho(x)\in \overline{M}$. Then 
 $f_{M,T,H}(x)$ and $f_{M,T,H}(x^*)$ are either $x_k$ and $ {x_k^*}$ 
  or $\overline{x_k}$ and $\smash{\overline{x_k^*}}$.
The claim follows directly from $x^*=x^{(k)}$ and thus, $x_k^*=\overline{x_k}$.
\end{proof}

For each $k \in \overline{M}$, we partition strings $x\in X$ with $\rho(x)=k$ and $f(x)=0$ into
\[ X_{k}^{+} = \big\{ x \in X : \rho(x) = k,\hspace{0.04cm} x_k=0,\hspace{0.04cm}f(x) = 0 \big\} 
\quad\text{and}\quad X_{k}^{-} = \big\{ x \in X : \rho(x) = k,\hspace{0.04cm}x_k=1,\hspace{0.04cm} f(x) = 0\big\}. \]
Note that for each $x\in X_k^+$, $(x,x^*)$ is a monotone bi-chromatic edge; for each $x\in X_k^-$,
  $(x,x^*)$~is an anti-monotone bi-chromatic edge.
Since all these $|X|/2$ edges are disjoint,
  by Lemma \ref{ufuf} we have:
\[ \dist\big(f_{M,T,H}, \textsc{Unate} \big) \geq \dfrac{1}{2^{n}} \cdot \sum_{k\in \overline{M}} \min\big\{ |X_k^+|, |X_k^{-}|  \big\}. \]
Therefore, it suffices to show that with probability $\Omega(1)$ over $\TT \sim \calE(M)$ and $\HH \sim \En(M)$,  both~$\XX_k^+$ and $\XX_k^{-}$
  (as random variables derived from $\TT$ and $\HH$) have size $\Omega(2^n/n)$ for every $k \in \overline{M}$. 

To simplify the proof we introduce a new distribution~$\calE'(M)$ that is 
  the same as $\calE(M)$ but conditioned on that 
  every $T_i$ in $T$ contains at least $n^{1/3}$ elements.
Our goal is to show that 
\begin{equation}\label{pfof}
\mathop{\Pr}_{\TT\sim\calE'(M),\hspace{0.05cm}\HH\sim \En(M)}\Big[\hspace{0.03cm}
  \text{$\forall\hspace{0.05cm}k\in \overline{M}$,
  both $\XX_k^+$ and $\XX_k^-$ have size $\Omega(2^n/n)$\hspace{0.03cm}}\Big]=\Omega(1).
\end{equation}
This implies the desired  claim over $\bT\sim \calE(M)$ as the probability of $\bT\sim \calE(M)$ lying
  in the support of $\calE'(M)$ is at~least $1- \exp\left(-\Omega(\sqrt{n})\right)$. To see this is the case, the probability of $\bT_i$
    having less~than $n^{1/3}$ many elements can be bounded from above by
\begin{align*}
 \Prx \left[|\bT_i| \leq n^{1/3} \right] &=\sum_{j\le n^{1/3}}
   \binom{n/2}{j}\cdot \left(1-\frac{1}{\sqrt{n}}\right)^{n/2-j}\cdot \left(\frac{1}{\sqrt{n}}\right)^{j}\\ 
 &\le (n^{1/3}+1)\cdot \binom{n/2}{n^{1/3}} \cdot \left( 1 - \frac{1}{\sqrt{n}} \right)^{n/2 - n^{1/3}} < e^{- 0.49 \sqrt{n}}.
   \end{align*}
Taking a union bound over all $N\approx e^{\sqrt{n}/4}$ terms, we conclude that 
  $\bT \sim \calE(M)$ lies in the support of $\calE'(M)$ 
  with probability at least $1 - \exp(-0.24 \sqrt{n} )$.

In Claim \ref{claim111}, we prove a lower bound for the expectation of $|\bX|$: %

\begin{claim}\label{claim111}
We have \emph{(}below we use $\bH$ as an abbreviation for $\bH\sim \calEn(M)$\emph{)}
\begin{equation}
\label{eq:mu1} 
 \Ex_{\bT\sim\calE(M) ,\hspace{0.05cm} \bH}\Big[ |\bX|\Big] = \Omega\hspace{0.03cm}(2^n)
 \quad\text{and}\quad
 \Ex_{\bT\sim\calE'(M),\hspace{0.05cm} \bH}\Big[ |\bX|\Big] = \Omega\hspace{0.03cm}(2^n).
\end{equation}
\end{claim}
\begin{proof}
By linearity of expectation, we have
\[ \Ex_{\bT \sim \calE(M) ,\hspace{0.05cm} \bH}\Big[|\bX| \Big] = \sum_{\text{middle $x$}} \mathop{\Pr}_{\bT \sim\calE(M),\hspace{0.05cm}\bH }\big[\hspace{0.01cm}x \in \bX \hspace{0.03cm}\big]. \]

Fix a string $x \in \{0, 1\}^n$ in the middle layers
  (i.e., $|x_M|$ lies in $n/4\pm \sqrt{n}$). We decompose the probability on the RHS for $x$ into $N$
    disjoint subevents.
The $i$th subevent corresponds to $\bT_i$ being the unique term which $x$ satisfies. The probability of 
  the $i$th subevent is at least 
\[ \left( 1 - \dfrac{1}{\sqrt{n}} \right)^{\frac{n}{4} + \sqrt{n}} \times \left(1 - \left(1 - \frac{1}{\sqrt{n}} \right)^{\frac{n}{4} - \sqrt{n}} \right)^{N - 1} = \Omega\left(\frac{1}{N}\right). \] 
As a result, the probability of $x \in \XX$ is $N \cdot \Omega(1/N ) = \Omega(1)$.
The first part of (\ref{eq:mu1}) follows from the fact that there are $\Omega(2^n)$ many strings $x$ 
  in the middle layers. 
  
The second part of (\ref{eq:mu1}) 
  follows from the first part and the fact that $|\bX|\le 2^n$ and 
  $\TT\sim \calE(M)$ does not lie in 
  the support of $\calE'(M)$ with probability $o(1)$ as shown above.
\end{proof}


Let $\mu^*=\Omega(2^n)$ be the expectation of $|\bX|$ over $\bT\sim\calE'(M)$ and $\bH\sim\calEn(M)$,
  and let $p$ be the probability of $|\bX|\ge \mu^*/2$. Then we have
$$
\mu^*\le p\cdot 2^n+(1-p)\cdot (\mu^*/2)\le p\cdot 2^n+\mu^*/2
$$
and thus, $p=\Omega(1)$.
As a result, it suffices to consider 
  a $T$ in the support of $\calE'(M)$ that satisfies $|X| \geq \mu^*/2$
  and show that, over $\bH\sim\calEn(M)$, all $|\bX_k^+|$ and $|\bX_k^-|$ 
  are $\Omega(2^n/n)$ with probability $\Omega(1)$.
To this end, we focus on $\bX_k^+$ and then use symmetry and a union bound on
  all the $n$ sets.
  
Given $T$ and its $X$ (with $|X|\ge \mu^*/2$), we note that 
  half of $x\in X$ have $x_k=0$ (since whether $x\in X$ only depends on $x_M$)
  and for each $x\in X$ with $x_k=0$, the probability of $x\in \bX_k^+$ (over $\bH$)
  is $1/(2n)$.
Hence, the expectation of $|\bX_k^+|$ is $|X|/4n\ge \mu^*/8n=\Omega(2^n/n)$.
Let $\mu=|X|/4n$.
To obtain a concentration bound on $|\bX_k^+|$, we apply Hoeffding's inequality 
  over $\bH\sim \calEn(M)$ in the next claim.

\begin{claim}
For each $k \in \overline{M}$, we have 
\[ \mathop{\Pr}_{\bH\sim\calEn(M)}\Big[\hspace{0.02cm}\mu -  |\bX_k^{+}| \geq {\mu}/{2} \hspace{0.02cm}\Big] \leq \exp\left( -\Omega\hspace{0.02cm}\big(2^{n^{1/3}}/n^2\big)\right). \]
\end{claim}

\begin{proof}
Consider the size of $X_k^+$  as a function over
  $\smash{h_1,\ldots,h_N}$ for a particular fixed $T$ in the support of $\calE'(M)$ with $|X| \geq \Omega(2^n)$. We have that $X_k^+$ is a sum of independent random variables taking values between $0$ and $\smash{2^{n - n^{1/3}}}$, and the expectation of~$|\bX_k^+|$~is $\mu$ because the choices in $\bH$ partitions half of $X$ into $2n$ disjoint parts. 
  Therefore, we can now apply Hoeffding's inequality:
{ \[ \Prx_{\bH \sim \calEn(M)} \left[\mu - |\bX_k^+| \geq \frac{\mu}{2} \right] \leq \exp\left( -\dfrac{\Omega(2^{2n} / n^2)}{2^{2n - n^{1/3}}} \right) \]}
As each term has length at least $n^{1/3}$, each $T_i$ can add at most $b_i < (1/2) \cdot 2^{n - n^{1/3}}$ to $|\bX_k^+|$, then $$\sum_{i \in [N]} b_i^2 \leq 2^{n-n^{1/3}} \sum_{i\in [N]} b_i \leq 2^{2n-n^{1/3}}.$$
This finishes the proof of the claim. 
\end{proof}

The same argument works for $|\bX_k^-|$.
(\ref{pfof}) then follows from a union bound on $k \in \overline{M}$ and both sets $\bX_k^+$ and $\bX_k^-$. This finishes the proof of Lemma~\ref{lem:nonunate}.
\end{proof}

Given Lemmas~\ref{lem:unate} and \ref{lem:nonunate}, our lower bound for testing unateness
  (Theorem \ref{unatemain}) follows directly from the lemma below. We fix $q = {n^{2/3}}/{\log^3 n}$
  as the number of queries in the rest of the proof. The remainder of this section will prove the following lemma.
\begin{lemma}
\label{lem:lb}
Let $B$ be any $q$-query deterministic algorithm with oracle access to $f$. Then 
\[ \mathop{\Pr}_{\ff \sim \dno}\big[B \text{ rejects } \ff\hspace{0.03cm}\big] \leq \mathop{\Pr}_{\ff \sim \dyes}
\big[B \text{ rejects } \ff\hspace{0.03cm}\big] + o(1). \] 
\end{lemma}

\subsection{Balanced decision trees}\label{ueiu}

Let $B$ be a $q$-query deterministic algorithm, i.e., a binary decision tree of depth at most $q$ in which
  each internal node is labeled a query string $x\in \{0,1\}^n$ and each leaf is labelled
  ``accept'' or ``reject.''
Each internal~node $u$ has one $0$-child and one $1$-child.
For each internal node $u$, we use $Q_u$ to denote the set of strings queried so far
  (not including the query $x$ to be made at $u$).

Next we give the definition of a $q$-query tree $B$ being
  \emph{balanced} with respect to a subset $M \subset [n]$ of size $n/2$ and a string $r \in \{0,1\}^M$   
  (as the $\bM$ and $\br$ in the procedure that generates $\Dyes$ and $\Dno$).
After the definition we show that, when 
  both $\bM$ and $\br$ are drawn uniformly at random (as in the procedure),
  $B$ is balanced with respect~to $\bM$ and $\br$ with probability at least $1-o(1)$.

\begin{definition}[Balance]\label{def:balance}
We say $B$ is \emph{balanced} with respect to a subset $M \subset [n]$ of size~$n/{2}$ and $r \in \{0,1\}^{M}$ if for every internal node $u$~of $B$
  \emph{(}letting $x$ be the query at $u$\emph{)} and every $Q\subseteq Q_u$, with
\begin{equation}\label{iiififif}
A=\big\{k\in [n]: \forall\hspace{0.05cm} y, y' \in Q,\hspace{0.05cm} y_k = y_k' \big\}\quad\text{and}\quad 
A' = \big\{k \in [n] : \forall\hspace{0.05cm} y, y' \in Q \cup \{x \},\hspace{0.05cm} y_k = y_k' \big\},
\end{equation}
the set $\Delta= A \setminus A'$ having size at least $n^{2/3} \log n$ implies that   
\begin{equation}\label{iiifififififif}\Delta_{1}=\big\{k\in \Delta \cap M: x_k \oplus r_k = 0\ \text{and}\ \forall\hspace{0.05cm} y\in Q,\hspace{0.05cm} y_k \oplus r_k= 1\big\}
\end{equation}
has size at least $n^{2/3}\log n/8$.
\end{definition}
\begin{lemma}\label{urur2}
Let $B$ be a $q$-query decision tree. Then $B$ is balanced with respect to a subset $\bM \subset [n]$ of size $n/{2}$ and 
  an $\rr\in  \{0,1\}^{\bM}$, both drawn uniformly at random, with probability at least $1-o(1)$
\end{lemma}
\begin{proof}
Fix an internal node $u$ and a $Q\subseteq Q_u$ such that $|\Delta|\ge n^{2/3}\log n$.
Then the probability {over the draw of $\bM$ and $\br$ of $\mathbf{\Delta}_{1}$} 
being smaller than $n^{2/3}\log n/8$ is at most 
  $\exp(-\Omega(n^{2/3}\log n))$ using~the Chernoff bound.~The lemma follows by a union bound as there are at most $O(2^q)$ choices for $u$
  and $2^q$ choices for $Q$.
\end{proof}

Lemma \ref{lem:lb} follows from the following lemma. 

  
\begin{lemma}\label{truetrue}
Let $B$ be a $q$-query tree that is balanced with respect to $M$ and $r$. Then we have
\begin{equation}\label{ororor}
\mathop{\Pr}_{\substack{\TT,\HH\sim\En(M),\bs} }\big[B \text{ rejects } f_{M,\TT,\HH,r,\bs}\hspace{0.03cm}\big] \leq \mathop{\Pr}_{\TT,\HH\sim\Ey(M),\bs} 
\big[B \text{ rejects } f_{M,\TT,\HH,r,\bs}\hspace{0.03cm}\big] + o(1). 
\end{equation}
where $\TT\sim \calE(M)$ and $\bs\sim \{0,1\}^{\overline{M}}$.
\end{lemma} 
\begin{proof}[Proof of Lemma \ref{lem:lb} assuming Lemma \ref{truetrue}]
To simplify the notation, in the  sequence of equations below we ignore in the subscripts 
  names of distributions from which certain random variables are drawn when it is clear from the context.
Using Lemma \ref{urur2} and Lemma \ref{truetrue}, we have
\begin{align*} 
&\hspace{-1.3cm}\mathop{\Pr}_{\bM,\bT,\HH\sim \En(M),\rr,\bs}\big[B \text{ rejects } f_{\bM, \TT,\HH,\rr,\bs}\hspace{0.03cm}\big]\\[1ex]
&\le \frac{1}{2^{n/2}\cdot  \binom{n}{n/2}}\cdot \sum_{M,r}\hspace{0.1cm}
\mathop{\Pr}_{\TT,\HH\sim \En(M),\bs}\big[B \text{ rejects } f_{M,\TT,\HH,r,\bs}\hspace{0.03cm}\big]\\&
\le \frac{1}{2^{n/2}\cdot \binom{n}{n/2}}\cdot \sum_{M, r:\hspace{0.05cm} \text{balanced $B$}} \hspace{0.1cm}\mathop{\Pr}_{\TT,\HH\sim \En(M),\bs}\big[B \text{ rejects } f_{M,\TT,\HH,r,\bs}\hspace{0.03cm}\big]+o(1)\\
&\le \frac{1}{2^{n/2}\cdot \binom{n}{n/2}}\cdot \sum_{M, r:\hspace{0.05cm} \text{balanced $B$}} \hspace{0.1cm}\mathop{\Pr}_{\TT,\HH\sim \Ey(M),\bs}\big[B \text{ rejects } f_{M,\TT,\HH,r,\bs}\hspace{0.03cm}\big]+o(1)\\[0.6ex]
&\le\mathop{\Pr}_{\bM,\TT,\HH\sim \Ey(M),\rr,\bs}\big[B \text{ rejects } f_{\bM,\TT,\HH,\rr,\bs}\hspace{0.03cm}\big]+o(1).
\end{align*}
This finishes the proof of Lemma \ref{lem:lb}.
\end{proof}

To prove Lemma \ref{truetrue}, we may consider an adversary that has $M$ of size $n/{2}$ and $r \in \{0,1\}^M$
  in hand and can come up with any $q$-query decision tree $B$ as long as $B$ is balanced with respect
  to $M$ and $r$. Our goal is to show that any such tree $B$ satisfies (\ref{ororor}).
This inspires us to introduce the definition of \emph{balanced} decision trees.

\begin{definition}[Balanced Decision Trees]\label{b-trees}
A $q$-query tree $B$ is said to be \emph{balanced} if it is balanced with respect to 
  $M^*=[n/2]$ and $r^*=0^{[n/2]}\in \{0,1\}^M$.
Equivalently, for every internal node $u$ of~$B$ and every $Q\subseteq Q_u$ \emph{(}letting $A$ and $A'$ denote the sets as defined in \emph{(\ref{iiififif})}\emph{)}, if
$ 
\Delta= A \setminus A'$  
has size~at~least $n^{2/3}\log n$, then the set $ 
\Delta_1$ 
  as defined in \emph{(\ref{iiifififififif})} using $M^*$ and $r^*$
  has size at least $n^{2/3}\log n/8$.
\end{definition}

With Definition \ref{b-trees} in hand, we use the following lemma to prove Lemma \ref{truetrue}.
  
\begin{lemma}\label{finalmain}
Let $B$ be a balanced $q$-query decision tree. Then we have
\begin{equation}\label{ororor2}
\mathop{\Pr}_{\TT,\HH\sim \En(M^*),\bs}\big[B \text{ rejects } f_{M^*,\TT,\HH,r^*,\bs}\hspace{0.03cm}\big] \leq \mathop{\Pr}_{\TT,\HH\sim \Ey(M^*),\bs}
\big[B \text{ rejects } f_{M^*,\TT,\HH,r^*,\bs}\hspace{0.03cm}\big] + o(1), 
\end{equation}
where $\TT\sim \calE(M^*)$ and $\bs\sim \{0,1\}^{\overline{M^*}}$.
\end{lemma}
\begin{proof}[Proof of Lemma \ref{truetrue} assuming Lemma \ref{finalmain}]
Let $B$ be a $q$-query tree that is balanced with respect to $M$ and $r\in \{0,1\}^M$,
  which are not necessarily the same as $M^*$ and $r^*$.
Then we use $B,M$ and $r$ to define a new $q$-query tree $B'$ that is balanced (i.e., with respect to $M^*$ and 
  $r^*$): $B'$ is obtained by replacing every query $x$ made in $B$ by $x'$,
  where $x'$ is obtained by first doing an XOR of $x$ with $r$ over coordinates in $M$
  and then reordering the coordinates of the new $x$ using a bijection between $M$ and $M^*$. Note that
  $B'$ is balanced and satisfies that the LHS of (\ref{ororor}) for $B'$ is the same as
  the LHS of (\ref{ororor2}).
The same holds the RHS as well. Lemma \ref{truetrue} then follows from Lemma \ref{finalmain}.
\end{proof}

For simplicity in notation, we fix $M$ and $r$ to be $[n/2]$ and $0^{[n/2]}$ in the 
  rest of the section. 
We also write $\calE$ for $\calE(M)$, $\Ey$ for $\Ey(M)$, and $\En$ for $\En(M)$.
Given $T$ in the support of $\calE$, $H$ from the support of $\Ey$ or $\En$, and $s\in \smash{\{0, 1\}^{\overline{M}}}$, we 
write $$f_{T,H,s}\eqdef f_{M,T,H,r,s}$$
for convenience. Then the goal (\ref{ororor2}) of Lemma \ref{finalmain} becomes
$$
\mathop{\Pr}_{\TT,\HH\sim \En,\bs}\big[B \text{ rejects } f_{ \TT,\HH,\bs}\hspace{0.03cm}\big] \leq \mathop{\Pr}_{\TT,\HH\sim \Ey ,\bs}
\big[B \text{ rejects } f_{ \TT,\HH, \bs}\hspace{0.03cm}\big] + o(1), 
$$
where $\bT\sim \calE$ and $\bs\sim \{0,1\}^{\overline{M}}$ in both probabilities.


\begin{remark}\label{haharemark} Since $B$ works on $f_{\TT,\HH, \bs}$ and $r$ is all-$0$,
  the multiplexer $\bGamma_\TT$ is first truncated 
  according to~$|x_M|$, the number of $1$'s in the first $n/2$ coordinates. 
As a consequence, we may  
  assume without loss generality from now on that $B$ only
  queries strings $x$ that have $|x_{M}|$ lying between $n/4 \pm\sqrt{n}$.
We will refer to them as strings in the middle layers in the rest of the section.
\end{remark}

\subsection{Balanced signature trees}\label{sec:unate-oracle}

At a high level we proceed in a similar fashion as in the monotonicity lower bound. 
We first define a new and stronger oracle model that returns more than
  just $f(x)\in \{0,1\}$ for each query~$x\in \{0,1\}^n$.
Upon each query $x\in \{0,1\}^n$,
  the oracle returns the so-called \emph{signature} of $x\in \{0,1\}^n$ with respect
  to $(T,H,s)$ when hidden function is $f_{T,H,s}$ (and it will become clear that $f_{T,H,s}(x)$ is 
  determined by the signature of $x$);
  in addition,  the oracle also reveals the special variable $k$ of a term $T_i$ 
  when the latter is \emph{breached} (see Definition \ref{defbreach}). 
Note that the revelation of special variables is unique in the unateness lower bound.
On the other hand, the definition of signatures in this section 
  is much simpler due to the single-level construction
  of the multiplexer map.
  
After the introduction of the stronger oracle model, 
  ideally we would like to prove that every $q$-query deterministic algorithm
  $C$ with access to the new oracle can only have at most $o(1)$ advantage
  in rejecting the function $f_{\TT,\HH, \bs}$ when $\TT\sim \calE$, $\HH\sim \En$ and $\smash{\bs \sim \{0,1\}^{\overline{M}}}$  
  as compared to
  $\TT$, $\HH\sim \Ey$ and $\bs$.
It turns out that we are only able to prove this when $C$ is 
    represented by a so-called~\emph{balanced signature tree}, a definition
  closely inspired by that
  of balanced decision trees in Definition \ref{b-trees}.
This suffices for us to prove
  Lemma \ref{finalmain} since only balanced decision trees are considered there.

Recall the definition of $e_i$ and $e_{i,i'}$ from Section~\ref{sec:mono}.
We first define signatures syntactically~and then semantically.
The two definitions below are simpler than their counterparts in Section~\ref{sec:mono}
  (as we only have one level of multiplexing in $\Gamma_T$).
By Remark \ref{haharemark}, we can assume without loss of generality that 
  every string queried lies in the middle layers.  

\begin{definition}
We use $\frakP$ to denote the set of all triples $(\sigma, a, b)$, where $\sigma \in \{0, 1, *\}^N$ and $a, b$ $\in \{0, 1, \perp\}$ satisfy the following properties:
\begin{flushleft}\begin{enumerate}
\item $\sigma$ is either 1) the all 0-string $0^N$, 2) $e_i$ for some $i \in [N]$, or 3) $e_{i, i'}$ for some $i<i' \in [N]$.\vspace{-0.08cm} 
\item If $\sigma$ is of case 1), then $a = b = \sperp$. If $\sigma$ is of case 2), then $a \in \{0, 1\}$ and $b = \sperp$. Lastly, if $\sigma$ is of case 3), then we have $a, b \in \{0, 1\}$. 
\end{enumerate}\end{flushleft}
\end{definition}

\begin{definition}\label{def:unate-sig}
We say $(\sigma, a, b)\in \frakP$ is the \emph{signature} of 
  a string $x \in \{0, 1\}^n$ in the middle layers with respect to $(T, H,s)$ if it satisfies the following properties:
\begin{flushleft}\begin{enumerate}
\item $\sigma \in \{0, 1, *\}^N$ is set according to the following three cases: 
1) $\sigma = 0^N$ if 
  $T_i(x) = 0$ for all $i \in [N]$; 2) $\sigma = e_i$ if $T_i(x) = 1$ is the unique term that is
  satisfied by $x$; 3) $\sigma = e_{i, i'}$ if $i<i'$ and $T_i(x) = T_{i'}(x) = 1$ are the first two terms that are satisfied by $x$.\vspace{-0.08cm} 
\item If $\sigma$ is in case 1), then $a = b = \sperp$. If $\sigma$ is in case 2) with $\sigma = e_i$, then $a = h_i(x\oplus s)$\hspace{0.1cm}\footnote{Recall that $x\oplus s$
  is the $n$-bit string obtained from $x$ after an XOR with $s$ over coordinates in $\overline{M}$.} and $b = \sperp$. If $\sigma$ is in case 3) with $\sigma = e_{i, i'}$, then $a = h_i(x\oplus s)$ and $b = h_{i'}(x\oplus s)$. 
\end{enumerate}\end{flushleft}

The signature of a set $Q \subset \{0, 1\}^n$ of strings in the middle layers with respect to $(T,H,s)$ is the map $\phi \colon Q \to \frakP$ such that $\phi(x)$ is the signature of $x$ with respect to $(T, H,s)$. 
\end{definition}

Next we show that $f_{T,H,s}(x)$ is uniquely determined by the signature of $x$.
Thus, the new oracle is at least as powerful as the standard one.
The proof is similar to that of Lemma~\ref{simple}.

\begin{lemma}
\label{lem:simul}
Let $T$ be from the support of $\calE$, $H$ be from the support of $\Ey$ or $\En$ and $s \in \{0, 1\}^{\overline{M}}$. Given an $x \in \{0, 1\}^n$ in the middle layers, $f_{T, H, s}(x)$ is uniquely determined by the signature $(\sigma,a,b)$ of $x$ with respect to $(T,H, s)$.
\end{lemma}
\begin{proof}
Let $f=f_{T,H,s}$. We consider the following three cases:
\begin{enumerate}
\item (No term is satisifed) If $\sigma=0^N$, then $f(x) = 0$.\vspace{-0.1cm}
\item (Unique term satisfied) If 
If $\sigma=e_i$ for some $i\in [N]$, then $f(x)=h_i(x\oplus s)=a$.\vspace{-0.1cm}

\item (Multiple terms satisfied) 
If $\sigma=e_{i,i'}$ for some $i<i'\in [N]$, then $f(x)=1$.
\end{enumerate}
This finishes the proof of the lemma.
\end{proof}

We have defined the signature of $x$ with respect to $(T,H,s)$,
  which is the first thing that the new oracle returns upon a query $x$.
Let $Q\subset \{0,1\}^n$ be a set of strings in the middle layers (and consider $Q$ as 
  the set of queries made so far by an algorithm).
Next we define terms \emph{breached} by $Q$ with respect to a triple $(T,H,s)$.
Upon a query $x$, the new oracle checks if there is any term(s) newly breached 
  after $x$ is queried; if so, the oracle also reveals its special variable in $\overline{M}$.

For this purpose, let
  $\phi:Q\rightarrow \frakP$ be the signature of $Q$ with respect to $(T,H,s)$,
  where $\phi(x)=(\sigma_x,a_x,b_x)$.
We say $\phi$ induces a $5$-tuple $(I;P;R;A;\rho)$ if it satisfies the following properties:
\begin{enumerate}
\item The set $I\subseteq [N]$ is given by $$I=\big\{i\in [N]: \exists\hspace{0.03cm}x\in Q\ \text{with}\ 
  \sigma_{x,i}=1\big\}.$$
\item $P=(P_i:i\in I)$ and $R=(R_i:i\in I)$ are two tuples of subsets of $Q$. For each $i\in I$,
$$
P_i=\big\{x\in Q: \sigma_{x,i}=1\big\}\quad\text{and}\quad 
R_i=\big\{x\in Q: \sigma_{x,i}=0\big\}.
$$
\item $A=(A_i,A_{i,0},A_{i,1}:i\in I)$ is a tuple of subsets of $[n]$. For each $i\in I$,
  $A_i=A_{i,0}\cup A_{i,1}$ and 
$$
A_{i,1}=\big\{k\in [n]: \forall\hspace{0.05cm}x\in P_i, x_k=1\big\}\quad\text{and}\quad
A_{i,0}=\big\{k\in [n]: \forall\hspace{0.05cm}x\in P_i, x_k=0\big\}.
$$
\item $\rho=(\rho_i:i\in I)$ is a tuple of functions $\rho_i:P_i\rightarrow \{0,1\}$ with
  $\rho_i(x)=a_x$ if either $\sigma_x=e_i$\newline or  $\sigma_x=e_{i,i'}$ for some $i'>i$,
  and $\rho_i(x)=b_x$ if $\sigma_x=e_{i',i}$ for some $i'<i$, for each $x\in P_i$,\newline i.e.,
  $\rho_i(x)$ gives us the value of $h_i(x\oplus s)$ for each $x\in P_i$.
\end{enumerate}

The following fact is reminiscent of Fact~\ref{fact:a-terms-clauses}.

\begin{fact}
Let $\phi:Q\rightarrow \frakP$ be the signature of $Q$ with respect to $(T,H,s)$. Then 
  for each $i\in I$, we have $T_i\subseteq A_{i,1}\cap M$,
  $T_i(x)=0$ for all $x\in R_i$, and $h_i(x\oplus s)=\rho_i(x)$ for each $x\in P_i$.
\end{fact}


We introduce the similar concept of consistency as in Definition~\ref{def:consistent}.

\begin{definition}\label{def:incons}ß
Let $(I;P;R;A;\rho)$ be the tuple induced by $\phi:Q\rightarrow \frakP$.
For each $i\in I$, we say $P_i$ is \emph{$1$-consistent} if $\rho_i(x)=1$ for all $x \in P_i$,
  and \emph{$0$-consistent} if  $\rho_i(x) = 0$ for all $x\in P_i$.
We say $P_i$ is \emph{consistent} if it is either $1$-consistent or $0$-consistent;
we~say $P_i$ is \emph{inconsistent} otherwise.
\end{definition}

We are now ready to define terms breached by $Q$ with respect to $(T,H,s)$.

\begin{definition}[Breached Terms]\label{defbreach}
Let $Q\subset \{0,1\}^n$ be a set of strings in the middle layers.
Let~$T$ be from the support of $\calE$, $H$ be from the support of $\calEy$ or $\calEn$,
  and $\smash{s\in \{0,1\}^{\overline{M}}}$.
Let $(I;P;R; A; \rho)$ be the tuple induced by the signature of $Q$ with respect to $(T,H,s)$.
We say the $i$th term is breached by $Q$ with respect to $(T,H,s)$, for some $i\in I$, if 
  at least one of the following two events occurs:
(1) $P_i$ is inconsistent or (2) 
  $|A_i \cap \overline{M}| \leq n/10$.
We say the $i$th term is \emph{safe} if it is not breached.
\end{definition}

We can now finish the formal definition of our new oracle model.
Upon each query $x$, the oracle first returns the signature of $x$ with respect to the 
  hidden triple $(T,H,s)$. It then examines if there is any newly breached term(s)
  (by Definition \ref{defbreach} there can be at most two such terms since
  $x$ can be added to at most two $P_i$'s) and return the special variable $k\in \overline{M}$
  of the newly breached term(s).
As a result, if $Q$ is the set of queries made so far, 
  the information returned by the new oracle can be summarized as a 
  $6$-tuple $(I;P;R;A;\rho;\delta)$, where
\begin{flushleft}\begin{enumerate}
\item $(I;P;R;A;\rho)$ is the tuple induced by the signature of $Q$ with 
  respect to $(T,H,s)$;
\item Let $I_B\subseteq I$ be the set of indices of terms breached by $Q$, and 
  let $I_S=I\setminus I_B$ denote the safe terms.
Then $\delta:I_B\rightarrow \overline{M}$ satisfies that $k=\delta(i)$ is the 
  special variable of the $i$th term in $h_i$.
\end{enumerate}\end{flushleft}
  
We view a $q$-query deterministic algorithm $C$ with access to the new oracle as 
  a \emph{signature~tree}, in which each leaf is labeled  ``accept'' or ``reject'' 
  and each internal node $u$ is labeled a query string $x\in \{0,1\}^n$ in the middle layers.
Each internal node $u$ has $|\frakP|\cdot O(n^2)$ children 
  with each~of its edges $(u,v)$~labeled by (1) a triple $(\sigma,a,b)\in \frakP$ as the signature
  of $x$ with respect to the hidden $(T,H,s)$, and (2)
  the special variable of any newly breached (at most two) term(s). 
Each node $u$ is associated with a 
  set $Q_u$ as the set of queries made so far (not including $x$), its signature $\phi:Q_u\rightarrow\frakP$,
  and a tuple $(I;P;R;A;\rho;\delta)$ as the summary of all information received from the oracle so far.
(Note that one can fully reconstruct the signature $\phi$ from $(I;P;R;A;\rho)$ so it is redundant to keep $\phi$.
We keep it because sometimes it is (notation-wise) easier to work with $\phi$ directly.)

Finally we define \emph{balanced} signature trees. 

\begin{definition}[Balanced Signature Trees]
We say that a signature tree $C$ is \emph{balanced} if for~any internal node $u$ of $C$
  \emph{(}letting $x$ be the query to make and
   $(I;P;R;A;\rho;\delta)$ be the summary so far\emph{)}~and any $i\in I$,
  $\Delta=\{j\in A_i: x_j\ \text{disagrees with}\ y_j\ \text{of $y\in P_i$}\}$ having size at least
  $ n^{2/3}\log n$ implies that $\Delta_{1}=\{k\in \Delta \cap M: x_k=0\ \text{and}\ 
  \forall\hspace{0.05cm} y \in  {P_i},\hspace{0.05cm} y_k = 1 \}$ has size at least $ n^{2/3}\log n/8$.
\end{definition}
Note that the definition above  is weaker compared to 
  Definition \ref{b-trees} of balanced decision trees, in the sense that the condition on
  $\Delta_1$ in the latter
  applies to any subset of queries $Q\subseteq Q_u$ (instead of only $P_i$'s). 
Lemma \ref{finalmain} follows from the lemma below on balanced signature trees.

\begin{lemma}\label{finalfinalmain}
Let $C$ be a $q$-query balanced signature tree. Then we have
\begin{equation}\label{ororor3}
\mathop{\Pr}_{\TT,\HH\sim \En,\bs}\big[\hspace{0.01cm}C \text{ rejects } (\TT,\HH, \bs)\hspace{0.03cm}\big] \leq \mathop{\Pr}_{ \TT,\HH\sim \Ey,\bs}
\big[\hspace{0.01cm}C \text{ rejects } (\TT,\HH, \bs)\hspace{0.03cm}\big] + o(1). 
\end{equation}
\end{lemma}
\begin{proof}[Proof of Lemma \ref{finalmain} assuming Lemma \ref{finalfinalmain}]
Let $B$ be a $q$-query balanced decision tree.~We use $B$ to obtain~a~$q$-query algorithm $C$ with access to the new oracle 
   by simulating
  $B$ as follows: Each time a string~$x$~is queried, $C$ uses the signature of $x$ returned by the oracle to
  extract $f(x)$ (using Lemma \ref{lem:simul}) and then continue the simulation of $B$.
One can verify that the corresponding signature tree of $C$ is balanced and
  the probabilities of $C$ rejecting $(\TT,\HH, \bs)$ in both cases are the same as $B$.
\end{proof}

Before moving on to the proof of Lemma \ref{finalfinalmain},
  let us remark on how the new oracle may help an algorithm distinguish between functions in $\Dy$ and $\Dn$. Suppose that a deterministic algorithm~$C$ is at some internal node $u$ with a tuple $(I; P; R; A; \rho;\delta)$. 
For each breached $i\in I_B$,
  the algorithm knows that $h_i$ is either a dictator or anti-dictator 
 with special variable $x_k$ with $k=\delta(i)$. 
By inspecting the $y_k$ of a $y \in P_i$ and $\rho_i(y)$, the algorithm can also deduce whether 
  $h_i(x\oplus s)$ is $x_k$ or $\overline{x_k}$.
The former suggests that $x_k$ is monotone and the latter suggests that $x_k$ is anti-monotone.

However, unlike monotonicity testing, observing $h_i(x\oplus s)=\overline{x_k}$ 
  has no indication on whether~$f$ is drawn from $\Dyes$ or $\Dno$:
  indeed $h_i(x\oplus s)$ is equally possible to be $x_k$ or $\overline{x_k}$ in both distributions because 
  of the random bit $s_k$.
But if the algorithm observes a so-called \emph{collision}, i.e.  
  $i,i'\in I_B$ such that $h_i(x\oplus s)=x_k$ and 
  $h_i(x\oplus s)=\overline{x_k}$, then one can safely assert that the hidden function belongs 
  to $\Dn$.
This gives us the crucial insight (as sketched earlier in Section \ref{sec:unatedist}) that leads to a higher unateness testing lower bound than  monotonicity testing: for testing monotonicity, deducing that a variable goes in an anti-monotone direction suffices for a violation; for testing unateness, however, 
one needs to find a collision in order to observe a violation.
While the proof of Lemma \ref{finalfinalmain} is quite technical,
  it follows the intuition that with $q$ queries, it is hard for a balanced signature tree
  to find a collision in breached terms $I_B$, and when no collision is found, it
  is hard to tell where the hidden function is drawn from.




\subsection{Tree pruning}

To prove Lemma \ref{finalfinalmain} on a given balanced $q$-query signature tree $C$,
  we start by identifying a set of \emph{bad edges}~of $C$ and using them to prune the tree. 
\begin{definition}\label{def:bad-edges}
An edge $(u, v)$ in $C$ is  a \emph{bad} edge if at least one of the following events~\mbox{occurs} at $(u,v)$ and none of these events occurs along the root-to-$u$ path \emph{(}letting
  $x$ be the string queried at $u$,
  and $\smash{(I_B \cup I_S;P;R;A;\rho;\delta)}$ and $(I_B' \cup I_S';P';R';A';\rho';\delta')$
  be the summaries at $u$ and $v$, respectively\emph{)}:
\begin{enumerate}
\item For some $i\in I_S$, $|A_{i}\setminus A_i'\hspace{0.03cm}| \geq n^{2/3} \log n$;\vspace{-0.08cm}
\item $|I_B'| >  {n^{1/3}}\big/{\log n}$; or\vspace{-0.08cm}
\item There exist two distinct indices $i, j \in I_{B}'$ with $\delta'(i) = \delta'(j)$.
\end{enumerate}
\end{definition}

We say a leaf $\ell$ of $C$ is a \emph{good} leaf if there is no bad edge along the root-to-$\ell$ path; otherwise,
  $\ell$ is \emph{bad}. The following lemma allows us to focus on good leaves. We defer the proof to Section~\ref{sec:prunebadsec}.

\begin{lemma}[Pruning Lemma]
\label{lem:prunebad}
Let $C$ be a balanced $q$-query signature tree. Then
\[ \mathop{\Pr}_{\TT,\HH \sim \En,\bs}\big[\hspace{0.01cm}\text{$(\TT, \HH, \bs)$  reaches a bad leaf}
\hspace{0.1cm}\big] = o(1). \] 
\end{lemma}


We prove the following lemma for good leaves in Section \ref{lem:goodleaves}:
\begin{lemma}[Good Leaves are Nice]
\label{lem:goodleaves}
For any good leaf $\ell$ of $C$, we have
\[ \mathop{\Pr}_{ \TT,\HH\sim\En,\bs}\big[\hspace{0.01cm}(\TT, \HH, \bs) \text{ reaches $\ell$}
\hspace{0.04cm}\big] \leq (1 + o(1))\cdot \mathop{\Pr}_{\TT,\HH\sim\Ey,\bs}
\big[\hspace{0.01cm}(\TT, \HH, \bs)\text{ reaches $\ell$}\hspace{0.04cm}\big]. \]
\end{lemma}

Assuming Lemma~\ref{lem:prunebad} and Lemma~\ref{lem:goodleaves}, we can prove Lemma~\ref{finalfinalmain}:
\begin{proof}[Proof of Lemma~\ref{finalfinalmain} assuming Lemma~\ref{lem:prunebad} and Lemma~\ref{lem:goodleaves}]
Let $L$ be the set of leaves of $C$ that are labeled ``reject'' and let $L^* \subseteq L$ be the good ones in $L$. Then we have 
\begin{align*}
\mathop{\Pr}_{\TT,\HH\sim\En,\bs}\big[\hspace{0.01cm}C \text { reject } (\TT, \HH, \bs)\hspace{0.03cm}\big] &= \sum_{\ell \in L}\hspace{0.1cm} \mathop{\Pr}_{\TT,\HH\sim\En,\bs}\big[\hspace{0.01cm}(\TT, \HH, \bs) \text{ reaches }\ell\hspace{0.03cm}\big] \\[0.3ex]
													   &\leq \sum_{\ell \in L^*}\hspace{0.01cm} \mathop{\Pr}_{\TT,\HH\sim\En,\bs}\big[\hspace{0.01cm}(\TT, \HH, \bs) \text{ reaches }\ell\hspace{0.03cm}\big] + o(1) \\[0.3ex]
													   &\leq (1 + o(1))\cdot  \sum_{\ell \in L^*} \mathop{\Pr}_{\TT,\HH\sim\Ey,\bs}\big[\hspace{0.01cm}(\TT, \HH, \bs) \text{ reaches }\ell\hspace{0.03cm}\big] + o(1) \\[0.5ex]
													   &\leq (1 + o(1))\cdot \mathop{\Pr}_{\TT,\HH\sim\Ey,\bs}\big[\hspace{0.01cm}C \text { rejects } (\TT, \HH, \bs)\hspace{0.03cm}\big] + o(1)\\[0.6ex]
													   &\le \mathop{\Pr}_{\TT,\HH\sim\Ey,\bs}\big[\hspace{0.01cm}C \text { rejects } (\TT, \HH, \bs)\hspace{0.03cm}\big]+o(1),
\end{align*}
where we used Lemma~\ref{lem:prunebad} in the second line and 
  Lemma~\ref{lem:goodleaves} in the third line. 
\end{proof}


\subsection{Proof of Lemma~\ref{lem:goodleaves} for good leaves}\label{sec:goodleaves}

The proof of Lemma \ref{lem:goodleaves} is similar in spirit to Lemma \ref{goodleaves}
  for monotonicity.  

Fix a good leaf $\ell$ in $C$. 
We let $Q$ be the set of queries made along the root-to-$\ell$ path,
  $\phi:Q\rightarrow \frakP$ be the signature of $Q$ with $\phi(x)=(\sigma_x,a_x,b_x)$ for each $x\in Q$, and
  let $(I_B \cup I_S;P;R;A;\rho;\delta)$ be the summary associated with $\ell$.
Since $\ell$ is a good leaf, there are no bad edges along the root-to-$\ell$ path. 
Combining this with the definition of breached/safe terms, we have the following list of properties:
\begin{flushleft}\begin{enumerate}
\item For each $i \in I_S$, $|A_i \cap \overline{M}| \geq n/10$;\vspace{-0.08cm}
\item Every $i \in I_S$ is either $1$-consistent or $0$-consistent;\vspace{-0.08cm}
\item $|I_B| \leq n^{1/3} \big/{\log n}$; and\vspace{-0.08cm}
\item For any two distinct indices $i, j \in I_B$, we have $\delta(i) \neq \delta(j)$.
\end{enumerate}
\end{flushleft}  
Let $D=\{\delta(i):i\in I_B\} \subset \overline{M}$ be the special variables 
  of breach terms.
We have $|D| = |I_B|$.  
  
Next we fix a tuple $T$ from the support of $\calE$ such that
  the probability of $(T,\HH, \bs)$ reaching $\ell$ is positive,
  when $\HH\sim \En$ and $\smash{\bs \sim \{0, 1\}^{\overline{M}}}$.
It then suffices to show that 
\begin{equation}\label{eqeqprove}
\Prx_{\HH\sim \Ey, \bs}\big[\hspace{0.01cm}\text{$(T,\HH, \bs)$ reaches $\ell$}\hspace{0.04cm}\big] \geq (1-o(1)) \Prx_{\HH\sim \En, \bs}\big[\hspace{0.01cm}\text{$(T,\HH, \bs)$ reaches $\ell$}\hspace{0.04cm}\big].
\end{equation}
The properties below follow directly from the assumption that the probability of $(T,\HH,\bs)$ reaching $\ell$
  is positive when $\HH\sim \En$ and $\bs\sim\smash{\{0,1\}^{\overline{M}}}$:
\begin{enumerate}
\item For every $x\in Q$ and $i\in [N]$ such that $\sigma_{x,i}\in \{0,1\}$, 
  we have $T_i(x)=\sigma_{x,i}$; and\vspace{-0.08cm}
\item For each $i\in I_B$, letting $k=\delta(i)$, there exists a bit $b$ such that
  $\rho_i(x)=x_k\oplus b$ for all $x\in P_i$.
\end{enumerate}
For each $i\in I_B\cup I_R$ we pick a string $y_i$ from $P_i$ arbitrarily  as a representative and let $\alpha_i=\rho_i(y_i)$.

We first derive an explicit expression for the probability over $\En$ in (\ref{eqeqprove}).
To this end, we note that, given properties listed above, $(T,H,s)$ (with $H$ from the 
  support of $\En$) reaches $\ell$ iff
\begin{flushleft}\begin{enumerate}
\item For each $i\in I_S$, let $k$ be the special variable of $h_i$.
Then we have 
  $k\in A_i\cap \overline{M}$, and $h_i$ is a dictatorship function if
  $y_{i,k}\oplus s_k=\alpha_i$ or an anti-dictatorship if $y_{i,k}\oplus s_k\ne\alpha_i$;
\item For each $i\in I_B$, the special variable of $h_i$ is the same as $k=\delta(i)$ and similarly,
   $h_i$ is a dictatorship function if
  $y_{i,k}\oplus s_k=\alpha_i$ or an anti-dictatorship if $y_{i,k}\oplus s_k\ne\alpha_i$.
\end{enumerate}\end{flushleft}
Thus, once $s$ is fixed, there is exactly one choice of $h_i$ for each $i\in I_B$ 
  and $|A_i\cap \overline{M}|$ choices of $h_i$ for each $i\in I_S$.
Since there are $(n/2)\cdot 2$ choices overall for each $h_i$, 
   the probability over $\En$ in (\ref{eqeqprove}) is
\[  \left( \frac{1}{n} \right)^{|I_B|}\cdot \prod_{i \in I_S} \left( \dfrac{|A_i \cap \overline{M}|}{n} \right).\]

Next we work on the more involved
  probability over $\Ey$ in (\ref{eqeqprove}).
Given properties listed above $(T,H,s)$ (with $H$ from the 
  support of $\Ey$ so every $h_i$ is a dictatorship function) reaches $\ell$ iff
\begin{flushleft}\begin{enumerate}
\item For each $i\in I_S$, let $k$ be the special variable of the dictatorship function $h_i$.
Then we have 
  $k\in A_i\cap \overline{M}$ and $s_k$ satisfies that $y_{i,k}\oplus s_k=\alpha_i$;
\item For each $i\in I_B$, the special variable of $h_i$ is the same as $k=\delta(i)$ and 
   $y_{i,k}\oplus s_k=\alpha_i$.
\end{enumerate}\end{flushleft} 
Note that once $s$ is fixed, these are independent conditions over $h_i$'s
  (among the overall $n/2$~choices for each $h_i$).
As a result, we can rewrite the probability for $\Ey$ as
\begin{equation}\label{eq:prob}
\Ex_{\bs \sim \{0, 1\}^{\overline{M}}} \hspace{0.08cm}\left[\hspace{0.08cm} \prod_{i \in I} \bZ_i 
\hspace{0.06cm}\right], 
\end{equation}
where $\bZ_i$'s are (correlated) random variables that depend on $\bs$.
For each $i\in I_B$, 
$\bZ_i = 2/n$ if $$\alpha_i = y_{i,\delta(i)}\oplus \bs_{\delta(i)}$$ 
  and $\bZ_i=0$ otherwise. 
For each $i\in I_S$, we have
$$
\bZ_i=\frac{|\{k\in A_i\cap \overline{M}:y_{i,k}\oplus \bs_k=\alpha_i \}|}{n/2}
$$
For some technical reason, for each $i\in I_S$, let $\bB_i$ be the following random set that depends on $\bs$: 
\[ \bB_{i} = \big\{ k \in (A_i \cap \overline{M}) \setminus D :  y_{i,k} \oplus \bs_k = \alpha_i \big\}.\]
Using $|D|=|I_B|$, we may now simplify (\ref{eq:prob}) by:
\begin{align*}
\Ex_{\bs \sim \{0, 1\}^{\overline{M}}} \hspace{0.08cm}\left[ \hspace{0.08cm}\prod_{i \in I} \bZ_i \hspace{0.06cm}\right] &= \frac{1}{2^{|I_B|}} \cdot \left( \frac{2}{n}\right)^{|I_B|} \Ex_{\bs \sim \{0, 1\}^{\overline{M}\setminus D}} \left[\hspace{0.08cm} \prod_{i \in I_S} \bZ_i\hspace{0.06cm} \right] 
				&\geq  \left( \frac{1}{n}\right)^{|I_B|} \Ex_{\bs \sim \{0, 1\}^{\overline{M}\setminus D}} \left[ \hspace{0.08cm}\prod_{i \in I_S} \left( \dfrac{|\bB_{i}|}{n/2} \right)\hspace{0.05cm} \right].
\end{align*}


Therefore, it remains to show that
\begin{equation}\label{eq:final}\Ex_{\bs \sim \{0, 1\}^{\overline{M}\setminus D} }\left[\hspace{0.08cm}  \prod_{i \in I_S} \left(\dfrac{2|\bB_{i }|}{|A_i \cap \overline{M}| }\right)
\hspace{0.05cm}\right] \geq 1 - o(1). 
\end{equation}
Next we further simplify (\ref{eq:final}) by introducing new, simpler random variables.
We may re-write
\[ |\bB_{i}| = \sum_{k \in (A_i \cap \overline{M}) \setminus D} \bX_{i,k},\quad \text{where\quad\ } \bX_{i,k} = \left\{ \begin{array}{ll} 1 & \text{if}\  y_{i,k} \oplus \bs_k = \alpha_i \\[0.6ex]
								0 & \text{otherwise} \end{array} \right. \]
For each $i \in I_S$ and $k \in (A_{i} \cap \overline{M}) \setminus D$, let $\bY_{i,k}$ and $\bY_i$ be the following random variables:
\[ \bY_{i,k} = \dfrac{1 - 2\bX_{i,k} + 2\tau_i}{|A_i \cap \overline{M}| }\quad \text{and} \quad \bY_i = \sum_{k \in A_i \cap M \setminus D} \bY_{i,k},\quad\text{where\quad\ }\tau_i = \frac{| A_i\cap  {\overline{M}} \cap D|}{2 |(A_i \cap \overline{M}) \setminus D|}.\] 
(Note that $|(A_i \cap \overline{M}) \setminus D|$ is $\Omega(n)$ so $\tau_i$'s are well-defined.)
A simple derivation shows that 
\begin{align}\label{derivation}
 \prod_{i \in I_S} \left( \dfrac{2|\bB_{i }|}{|A_i \cap \overline{M}| }\right) &= \prod_{i \in I_S} \left( 1 - \sum_{k \in (A_{i} \cap M) \setminus D} \bY_{i,k}\right) 
 	 = \prod_{i \in I_S} \big( 1 - \bY_i \big).
\end{align}
Using the fact that each fraction on the LHS is between $0$ and $2$, we have that 
  $\bY_i$ always satisfies $|\bY_i|\le 1$.
The difficulty in lowerbounding (\ref{derivation}) is that $\bY_i$'s are not independent.
But with a fixed~$i$, $\bY_{i,k}$'s are indeed independent 
  with respect to the randomness in $\bs$ and each $\bY_{i,k}$ is either 
$$  %
	\frac{1}{|A_i \cap \overline{M}|} + O\left(\frac{1}{n^{5/3} \log n}\right)\quad\text{or}\quad
	-\frac{1}{|A_i \cap \overline{M}|} + O\left(\frac{1}{n^{5/3} \log n}\right) $$
with equal probabilities,
where we used the fact that $|A_i \cap \overline{M}| = \Omega(n)$ and $|D| \leq n^{1/3}/ \log n$.

For each $i \in I_S$, 
  let $\bW_i$ be the random variable defined as
\[ \bW_i = \left\{ \begin{array}{ll} \bY_i & \text{if}\ |\bY_i| \leq \log^2 n/{\sqrt{n}} \\[0.6ex]
						   2|I_S| & \text{otherwise} \end{array} \right. \]
We prove the following claim that helps us avoid the correlation between $\bY_i$'s.						   
		
\begin{claim}\label{cl:final-claim}
The following inequality always holds:
\begin{equation*} \prod_{i \in I_S} \big(1 - \bY_i\big) \geq \big(1 - o(1)\big) \cdot \left( 1 - \sum_{i \in I_S} \bW_i \right).\end{equation*}
\end{claim}
\begin{proof}
The inequality holds trivially if $|\bY_j| \geq {\log^2 n}/{\sqrt{n}}$ for some $j \in I_S$. This is because $|\bY_i|\le 1$ and thus, the LHS is nonnegative.
On the other hand $\bW_j=2|I_S|$ implies that the RHS is negative even when every other
  $\bW_i$ is $-1$. 
So we may assume that $|\bY_i| \leq  {\log^2 n}/{\sqrt{n}}$ for every $i$. The proof in this case follows   directly from Claim~\ref{cl:real-nums} in the appendix.
\end{proof}
		
Given Claim \ref{cl:final-claim}, it suffices to upperbound the expectation of each $\bW_i$ over $\bs
  \sim\{0,1\}^{\overline{M}\setminus D}$:	
\begin{equation}\label{expw}  
  \Ex_{\bs \sim \{0, 1\}^{\overline{M}\setminus D}}\big[\bW_i\big]
  \le \Ex_{\bs \sim \{0, 1\}^{\overline{M}\setminus D}}\big[\bY_i\big]
  +\big(2|I_S|+1\big)\cdot \Pr_{\bs}\big[\bY_i\ge \log^2 n/\sqrt{n}\big]=O\left(\frac{1}{n^{2/3} \log n}\right) 
\end{equation}
where we used $|I_S| \leq n^{2/3}$ and that the probability of 
  $\bY_i\ge \log^2n/\sqrt{n}$ is superpolynomially small, by a Chernoff bound.
Our goal, (\ref{eq:final}), then follows directly from (\ref{expw}) and Claim \ref{cl:final-claim}. 

\subsection{Proof of the pruning lemma}
\label{sec:prunebadsec}

Let $E$ be the set of bad edges in $C$.
We start by partitioning $E$ into three (disjoint) subsets~$E_1,E_2$ and $E_3$ according the the event that 
  occurs at $(u,v)\in E$.
Let $(u,v)\in E$ and~let~$(I_B \cup I_S; P; R; A;$ $\rho;\delta)$ and $(I_B' \cup I_{S}'; P'; R'; A'; \rho';\delta')$ 
  be the summaries associated with $u$ and $v$, respectively.
Then 
\begin{enumerate}
\item $(u,v)\in E_1$ if for some $i \in I_{S}$, we have $|A_i \setminus A_{i}'| \geq n^{2/3} \log n$;\vspace{-0.06cm} 
\item $(u,v)\in E_2$ if $(u,v)\notin E_1$ and $|I_B'|\ge {n^{1/3}}/{\log n}$; \vspace{-0.05cm} or
\item $(u,v)\in E_3$ if $(u,v)\notin E_1\cup E_2$ and for two distance indices $i, j \in I_B'$, we have $\delta(i) = \delta(j)$. 
\end{enumerate}
Note that $E_1,E_2$ and $E_3$ are disjoint.
Moreover, by the definition of bad edges none of these events occurs
  at any edge along the root-to-$u$ path.

Our plan below is to show that the probability of $(\TT,\HH,\bs)$, as $\TT\sim \calE,\HH\sim\calEn$ and
  $\smash{\bs\sim\{0,1\}^{\overline{M}}}$, passing through an edge in $E_i$ is $o(1)$ for each $i$.
The pruning lemma follows from a union bound.

For edge sets $E_1$ and $E_3$, we show that for any internal node $u$ of $C$,
  the probability of $(\TT, \HH, \bs)$ taking an edge $(u, v)$ that belongs to $E_1$ or $E_3$ is at most $o(1/q)$,  conditioning on $(\TT, \HH, \bs)$ reaching $u$ when
  $\TT\sim \calE,\HH\sim\calEn$ and
  $\smash{\bs\sim\{0,1\}^{\overline{M}}}$. This allows us to apply Lemma~\ref{simplepruning}.
We handle $E_2$ using a different argument by showing that, roughly speaking,
  $I_B$ goes up with very low probability after each round of query and thus, 
  the probability of $|I_B|$ reaching $n^{1/3}/\log n$ is $o(1)$. 


\paragraph{Edge Set $E_1$.\hspace{0.05cm}} 
Fix an internal node $u$ of $C$. We show that the probability of $(\TT,\HH,\bs)$
  leaving~$u$ with an $E_1$-edge, conditioning on it reaching $u$, is $o(1/q )$.
It then follows from Lemma \ref{simplepruning} that~the probability of $(\TT,\HH,\bs)$
  passing through an $E_1$-edge is $o(1)$.
    
Let $x$ be the query made at $u$, and let $(I_B\cup I_S; P;R;A;\rho;\delta)$ be the summary 
  associated with $u$.
Fix an index $i\in I_S$. We upperbound by $o(1/q^2)$ the conditional probability of $(\TT,\HH, \bs)$ 
  taking an $E_1$-edge with $|A_i\setminus A_i'|\ge n^{2/3}\log n$.
The claim follows by a union bound on $i\in I_S$ (as $|I|=O(q)$).


Note that either $A_i'=A_i$ or $A_i'=A_i\setminus \Delta$, where 
$$\Delta=\big\{k\in A_i:x_k\ \text{disagrees\ 
  with $y_k$ of $y\in P_i$}\big\}.$$
Thus, a necessary condition for $|A_i\setminus A_i'|\ge n^{2/3}\log n$ to happen is 
  $|\Delta|\ge n^{2/3} \log n$ and $\bT_i(x)=1$.

Since $C$ is balanced, $|\Delta| \geq n^{2/3} \log n$ implies that 
$$\Delta_1=\big\{k\in A_i\cap M: x_k=0\ \text{and}\ y_k=1, y\in P_i\big\}$$ has size at least $n^{2/3}\log n/8$.
On the other hand, fix any triple $(T_{-i},H, s)$, where $T_{-i}$ is a tuple of $N-1$ terms with  $T_i$ missing, $H$ is from the support of $\En$ and $\smash{s \in \{0, 1\}^{\overline{M}}}$ 
  such that 
\begin{equation}\label{odod}
\mathop{\Pr}_{\TT_i} \big[\hspace{0.01cm}\text{$((T_{-i},\TT_i),H, s)$ reaches $u$}\hspace{0.03cm}\big]>0,
\end{equation}
where $\TT_i$ is  drawn by including each index in $M$ with probability $1/\sqrt{n}$.
It suffices to show that
\begin{equation}\label{ofofp}
 {\mathop{\Pr}_{\TT_i}\big[\hspace{0.01cm}((T_{-i},\TT_i),H,s) \text{ reaches $u$ and $\bT_i(x)=1$}\hspace{0.05cm}\big]}
\le o(1/q^2)\cdot {\mathop{\Pr}_{\TT_i}\big[\hspace{0.01cm}((T_{-i},\TT_i),H,s) \text{ reaches } u\hspace{0.03cm}\big]}. 
\end{equation}
For this purpose, note that given (\ref{odod}), 
  the event on the RHS of (\ref{ofofp}) occurs at $T_i$ if and only if
  $T_i$ is a subset of
$
A_{i,1}^*=A_{i,1}\cap M
$ and $T_i(y)=0$ for every $y\in R_i$;
we use $U$ to denote the set of all~such terms $T_i$ ($U$ cannot be empty by (\ref{odod})).
On the other hand, the event on the LHS of (\ref{ofofp})~occurs~if and only if 
  $T_i$ further avoids picking variables from $\Delta_1$, i.e.  
  $T_i\subseteq A_{i,1}^*\setminus \Delta_1$. 
We use $V$ to~denote~the set of all such $T_i$'s.
To prove (\ref{ofofp}), note that we can take any $T_i$ in $V$,
  add an arbitrary subset of $\Delta_1$, and the result must be a set in $U$.
As a result we have (note that the bound is very loose here)
\[ \dfrac{\Pr [\hspace{0.02cm}\bT_i\in V\hspace{0.02cm}]}{\Pr[\hspace{0.02cm}\bT_i\in  U\hspace{0.02cm}]} \leq \left(1 - \frac{1}{\sqrt{n}} \right)^{|\Delta_1|} = o(1/q^2). \]
This finishes the proof for $E_1$.
Next we work on the edge set $E_3$.

\paragraph{Edge set $E_3$.}
Fix an internal node $u$ of $C$. We show that the probability of $(\TT,\HH,\bs)$
  leaving~$u$ with an $E_3$-edge, conditioning on it reaching $u$, is $o(1/q )$.
By definition, we can assume that 
  there~is no bad edge along the root-to-$u$ path and thus, 
  $|I_B|\le n^{1/3}/\log n$ and $I_B$ has no collision, i.e. there are no distinct $i,j\in I_B$ such that $\delta(i)=\delta(j)$.
For $(\TT,\HH,\bs)$ to leave $u$ with an $E_3$-edge, 
  it must be the case that some (at most two) terms are breached after the query $x$
  and a collision occurs (either between a newly breached term and a term in $I_B$, or
  between the two newly breached terms). 

Fix a pair $(T,s)$, where $T$ is from the support of $\calE$ and 
  $\smash{s\in \{0,1\}^{\overline{M}}}$, such that 
  $(T,\HH,s)$ reaches $u$ with a non-zero probability when $\HH\sim \calEn$.
It suffices to show that
\begin{equation}\label{ofof2}
 {\mathop{\Pr}_{\HH}\big[\hspace{0.01cm}(T,\HH,s) \text{ reaches $u$ and a collision occurs}\hspace{0.05cm}\big]}
\le o(1/q)\cdot {\mathop{\Pr}_{\HH}\big[\hspace{0.01cm}(T,\HH,s) \text{ reaches } u\hspace{0.03cm}\big]}. 
\end{equation}
Note that set of (at most two) $i\in I_S$ such that 
  $x$ is added to $P_i$ after it is queried is determined by $T$ (if $x$ starts a new $P_i$, then
  this $i$ is safe for sure).
If there exists no such $i$, then the probability on the LHS of (\ref{ofof2}) is $0$ since no term
  is newly breached and we are done.
Below we prove (\ref{ofof2}) for 
  the case when $i\in I_S$ is the only index such that $x$ is added to $P_i$.
The case when there are two such $i$'s can be handled similarly.
  
The proof of (\ref{ofof2}) easily follows from the following simple but useful claim:

\begin{claim}\label{ffffclaim}
Let $T$ and $s$ be such that $(T,\HH,s)$ reaches $u$ with non-zero probability
  when $\HH\sim \calEn$.
Then conditioning on reaching $u$, $\bh_i$ has its special variable uniformly 
  distributed in $A_i\cap \overline{M}$.
\end{claim} 
\begin{proof}
As $i\in I_S$, $P_i$ is consistent.
For $(T,H,s)$ to reach $u$, the only condition on $h_i$ and its special variable $k$ is that
  (1) if $y_k\oplus s_k=\rho_i(y)$ for some $y\in P_i$, then $h_i$ is a dictatorship function $x_k$;
  (2) if $y_k\oplus s_k\ne \rho_i(y)$ for some $y\in P_i$, then $h_i$ is an anti-dictatorship function
  $\overline{x_k}$.
Given $T$ and $s$, there are $|A_i\cap \overline{M}|$ choices for $h_i$ among
  the $2\cdot (n/2)$ choices and they are all equally likely.
\end{proof}
  
Our goal, (\ref{ofof2}), follows easily from $|A_i\cap \overline{M}|=\Omega(n)$ since $i\in I_S$,
   Claim \ref{ffffclaim}, $|I_B|\le n^{1/3}/\log n$, our choice of $q=n^{2/3}/\log^3 n$, and 
   the fact that, for the event on the LHS to happen, 
   the special variable of $\bh_i$ must fall inside $I_B$.


\paragraph{Edge set $E_2$.} 
Let $(u,v)$ be a bad edge in $E_2$ with $|I_B'|\ge n^{1/3}/\log n$.
We decompose $I_B'$ into $K$ and $L$:
  $i\in I_B'$ is in $K$ if \emph{at the edge $(u',v^*)$ along the root-to-$v$ path
  where $i$ becomes newly~breached}, we have $|A_i^*\cap \overline{M}|\le n/10$, where $A_i^*$
  is the set at $v^*$,
  and $i\in I_B'$ is in $L$ otherwise (i.e. $|A_i^*\cap \overline{M}|>n/10$ but 
  $P_i^*$ at $v^*$ becomes inconsistent after the query at $u'$).
The claim below shows that $K$ is small:

\begin{claim}\label{reallyfinal}
For every $E_2$-bad edge $(u,v)$, we have $|K| \leq O(n^{1/3}/ \log^2 n)$. 
\end{claim}
\begin{proof}
Fix an $i\in K$ and let $(u',v^*)$ be the edge along the root-to-$v$ path 
  where $i$ becomes breached.
Note that when $A_i$ is first created along the path, $A_i=\overline{M}$ and
  $|A_i \cap \overline{M}| =n/2$ (since at that time $P_i$ consists of a single string). 
As we walk down the root-to-$u^*$ path,
  every time a string is added to $P_i$, the size of $A_i$ can only drop 
  by $n^{2/3}\log n$ (otherwise, this edge is an $E_1$-edge, contradicting with the 
  assumption that $(u,v)\in E_2$ since $E_1$ edges have a higher priority) and thus,
  $|A_i \cap \overline{M}|$ 
  can drop by at most 
  $n^{2/3} \log n$. 
As a result, we have that $|P_i^*|$ at $v^*$ is at least 
$$
1+\frac{n/2-n/10}{n^{2/3}\log n}=\Omega\left(\frac{n^{1/3}}{\log n}\right).
$$
Using the fact that each of the $q$ queries can 
  be added to at most two $P_i$'s, we have
$$
|K|\le \frac{2q}{\Omega(n^{1/3}/\log n)}=O\left(\frac{n^{1/3}}{\log^2 n}\right).
$$ 
This finishes the proof of the claim.
\end{proof}

It follows directly from Claim \ref{reallyfinal} that 
  every bad $(u,v)\in E_2$ has $|L|\ge n^{1/3}/(2\log n)$.
This~inspires us to consider the following random process of walking 
  down  the tree $C$ from its root, with respect to $(\TT,\HH,\bs)$ over $\TT\sim \calE$, $\HH\sim \calEn$, and $\bs\sim\smash{\{0,1\}^{\overline{M}}}$.
As we walk down an edge $(u,v)$ of $C$, letting $(I_B\cup I_S;P;R;A;\rho;\delta)$
  and $(I_B'\cup I_S';P';R';A';\rho';\delta')$ be the summaries associated with $u$ and $v$, if $|A_i\setminus A_i'|\ge n^{2/3}\log n$
  for some $i\in I_S$, then we fail and terminate the random process;
if not we add the newly breached term(s) $i$ with 
  and $|A_i'\cap \overline{M}|> n/10$ (so $P_i'$ becomes inconsistent), if any,
  to $\bL$.
We succeed if $|\bL|\ge n^{1/3}/(2\log n)$, and it suffices for us to show that 
  we succeed with probability $o(1)$ over $\TT,\HH$ and $\bs$.


\def\bm{\mathbf{m}}

For the analysis, let $u$ be an internal node of $C$, and fix 
  any pair $(T,s)$ such that $(T,\bH,s)$ can reach $u$ with a non-zero probability.
As discussed earlier, the set of (at most two) $P_i$, $i\in I_S$, that the query string $x$ joins is 
  determined only by $T$.
If one of them has $|A_i\setminus A_i'|\ge n^{2/3}\log n$ then the process would always fail;
  otherwise, we have that $\bL$ can grow by at most two and this occurs with
  probability (over the randomness of $\bH$ but conditioning on $(T,\bH,s)$ reaching $u$) 
  at most 
$$
p=O\left(\frac{n^{2/3}\log n}{n}\right)=O\left(\frac{\log n}{n^{1/3}}\right) 
$$
because $|A_i\cap \overline{M}|=\Omega(n)$ ($i\in I_S$), 
  the special variable of $\bh_i$ is uniform over $A_i\cap \overline{M}$ by Claim \ref{ffffclaim},
  and for $i$ to be added to $\bL$, the special variable
  of $\bh_i$ must lie in $A_i\setminus A_i'$ (of size at most $n^{2/3}\log n$). 

In summary, after each query the random process either fails, or if it does not fail,
  $\bL$ can grow by at most two with probability at most $p$.
Therefore, the probability that we succeed is at most 
\[ \Prx_{\bm \sim \Bin(q, p)}\left[ 2\bm \geq \frac{n^{1/3}}{2\log n} \right] = o(1), \]
since $q= {n^{2/3}}/{\log^3 n}$ and $p = O( {\log n}/{n^{1/3}})$.

This finishes the proof that $(\TT,\HH,\bs)$ passes through an edge in $E_2$ with probability $o(1)$. 

\section{Non-Adaptive One-Sided Unateness Lower Bound}\label{sec:nonadaptive}
In this section we prove Theorem \ref{nonadaptive}: an ${\Omega}(n/\log^2 n)$ lower bound on the query complexity of testing unateness for \emph{one-sided} and \emph{non-adaptive} algorithms. This lower bound
 matches the upper bound of \cite{CS16} up to a poly-logarithmic factor. Our arguments are an adaptation of Theorem~19 of \cite{FLNRRS} to the setting of unateness, with one additional observation that allows us to obtain a higher lower bound. Previously \cite{BMPR16} proved a lower bound of $\Omega\left(\sqrt{n}\right)$ for one-sided, non-adaptive algorithms. For the rest of the section, we fix $q = {n}/{\log^2 n}$. 
 
 \def\ii{\boldsymbol{i}}
 
For a fixed $n > 0$, we describe a distribution $\Dn$ supported on Boolean functions $f$ over $n+2$ variables. We then show that every $\ff \sim \Dn$ is $\Omega(1)$-far from unate. An $\ff \sim \Dn$ is drawn by first drawing an index $\ii \sim [n]$ uniformly at random, and then letting $\ff = f_{\ii}$, where for each $x \in \{0, 1\}^n$:
\begin{align*}
f_{\ii}(0, 0, x) &= 0, \\
f_{\ii}(0, 1, x) &= \overline{x_{\ii}}, \\
f_{\ii}(1, 0, x) &= x_{\ii}, \\
f_{\ii}(1, 1, x) &= 1. 
\end{align*}
In order to simplify the notation, given $a,b\in \{0,1\}$ and $i \in [n]$, 
we write $f_{i,ab} \colon \{0, 1\}^n \to \{0, 1\}$  
to denote the function $f_{i,ab}(x)=f_i(a,b,x)$ that 
  agrees with $f_i$ when $a$ and $b$ are the first two inputs.

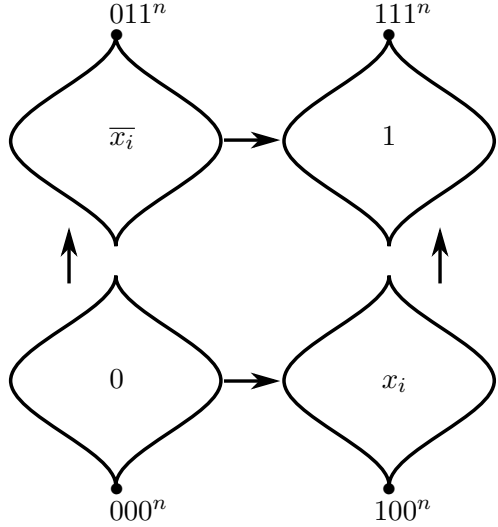
\begin{figure}\vspace{1cm}
\centering
\begin{picture}(200,160)
    \put(0,0){\includegraphics[width=0.4\linewidth]{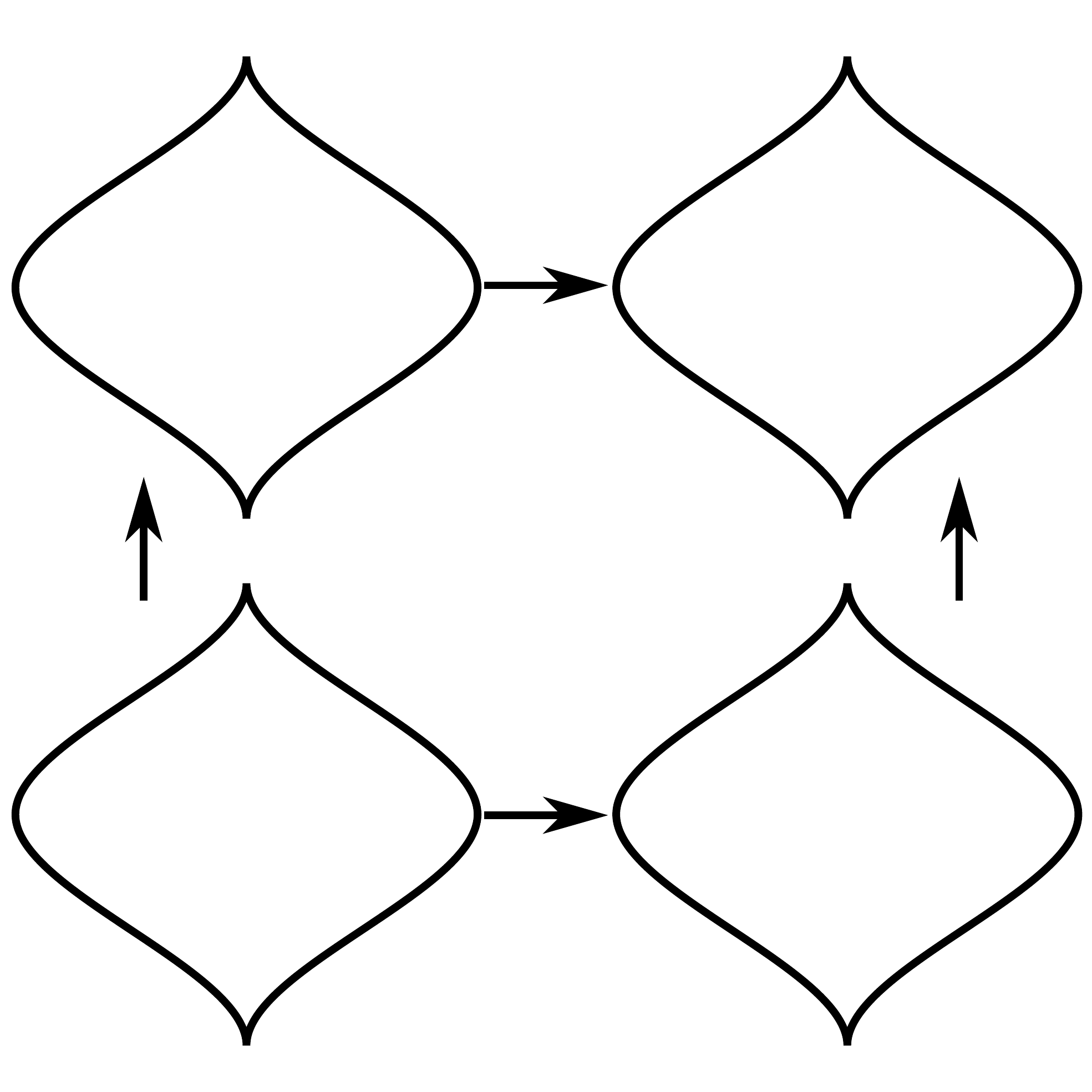}}
    \put(40, -5){$000^{n}$}
    \put(140, 183){$111^{n}$}
    \put(40, 183){$011^{n}$}
    \put(140, -5){$100^{n}$}
    \put(40, 4){$\bullet$}
    \put(40, 176){$\bullet$}
    \put(143, 176){$\bullet$}
    \put(143, 4){$\bullet$}
    \put(40, 45){$0$}
    \put(40, 137){$\overline{x_i}$}
    \put(143, 137){$1$}
    \put(143, 45){$x_i$}
  \end{picture}\vspace{0.18cm}
  \caption{An illustration of $f_i \colon \{0, 1\}^{n+2} \to \{0, 1\}$. The first two coordinates  index the sub-cubes.}\label{fig:unate-function}
\end{figure}

Figure~\ref{fig:unate-function} gives a simple visual representation of $f_i$.
We show that $f_i$ is the $\Omega(1)$-far from unate.

\begin{lemma}
\label{lem:unate-dist-far}
For all $i \in [n]$, $f_i$ is $\Omega(1)$-far from unate.
\end{lemma}

\begin{proof}
This is immediate from Lemma \ref{ufuf}, because there are $\Omega(2^{n})$ monotone 
  bi-chromatic edges in direction $i$, as well as $\Omega(2^{n})$ anti-monotone bi-chromatic edges in direction $i$.
\end{proof}

We consider \emph{non-adaptive, one-sided}, deterministic $q$-query algorithm $B$ with oracle access to a Boolean function. Note that a non-adaptive, deterministic algorithm $B$ is simply a set of $q$ query strings $x_1, \dots, x_q$, as well as a decision procedure which outputs ``accept'' or ``reject'' given  $f(x_k)$ for each $k \in [q]$. Furthermore, since $B$ is one-sided, $B$ outputs ``reject'' only if it observes a \emph{violation} to unateness (which we formally define next).
\begin{definition}
\label{def:vio-unate}
A \emph{violation} to unateness for a function $f \colon \{0, 1\}^n \to \{0, 1\}$ is a function $v \colon \{0, 1\}^n \to \left(\{0, 1\}^n\right)^2$, such that for each $r \in \{0, 1\}^n$\emph{:} $v(r) = (x, y)$ where $x, y \in \{0, 1\}^n$ and
\[ x \oplus r \prec y \oplus r \quad \text{and} \quad f(x) = 1, f(y) = 0. \]
\end{definition}

Intuitively, a violation to unateness consists of a violation to monotonicity, for every possibly orientation $r \in \{0, 1\}^n$. We refer to $f^r \colon \{0, 1\}^{n} \to \{0, 1\}$ as the function $f^r(x) = f(x \oplus r)$, for any $r \in \{0, 1\}^n$. So a violation to unateness for $f$ consists of a violation to monotonicity for each $f^r$. 

Thus, the algorithm $B$ with oracle access to $f \colon \{0, 1\}^n \to \{0, 1\}$ works in the following way:
\begin{enumerate}
\item Query the oracle with queries $Q = \{ x_1, \dots, x_q \} \subset \{0, 1\}^n$.
\item If there exists a violation to unateness of $f$, $v \colon \{0, 1\}^n \to \left(\{0, 1\}^n \right)^2$ where the image of  $v$, $\{ v(r) \colon r \in \{0, 1\}^n \}$, is a subset of $Q\times Q$, then output ``reject"; otherwise, output ``accept".
\end{enumerate}
Note that if $B$ does not find a violation, then there exists some unate function $f' \colon \{0, 1\}^n \to \{0, 1\}$ which is consistent with $Q$ (i.e., $f'(x_k) = f(x_k)$ for all $k \in [q]$). In order to say that $B$ does not find a violation, it suffices to exhibit some $r \in \{0, 1\}^n$ such that $B$ does not find a violation to monotonicity of $f^r$. Given Lemma~\ref{lem:unate-dist-far}, Theorem~\ref{nonadaptive} follows from the following lemma:

\begin{lemma}
\label{thm:one-sided}
For any $q$-query non-adaptive algorithm $B$, there exists some $r \in \{0, 1\}^{n+2}$ such that with probability $1 - o(1)$ over $\ii \sim [n]$, $B$ does not observe any violations to monotonicity of $f_{\ii}^{r}$.
\end{lemma}

\begin{proof}[Proof of Theorem~\ref{nonadaptive} assuming Lemma~\ref{thm:one-sided}]
Lemma~\ref{thm:one-sided} implies that with probability $1 - o(1)$ over the draw of $\ff \sim \Dn$, $B$ does not observe any violation to unateness, since there is some $r \in \{0, 1\}^{n+2}$ where $B$ does not observe any violation for monotonicity of $\ff^r$. Thus, any $q$-query algorithm $B$ does not output ``reject" on inputs drawn from $\Dn$ with probability at least $\frac{2}{3}$.
\end{proof}

We now proceed to prove Lemma~\ref{thm:one-sided}. For two strings $y,z\in\{0,1\}^{n}$, we denote the Hamming distance between $y$ and $z$ as $d(y,z) = | \{ k \in [n] \colon y_k \neq z_k \}|$.

\begin{lemma}
\label{lem:bad-dir}
For any $q$ strings $x_1, \ldots, x_q\in\{0,1\}^n$, there exists an $r \in \{0, 1\}^{n}$ such that for any $j, k \in [q]$, if
$x_j\oplus r \prec x_k\oplus r$,
then $d(x_j, x_k) \leq 2\log n$.
\end{lemma}

\begin{proof}
Consider a random $n$-bit  $\rr \sim \{0, 1\}^{n}$. Suppose $x_j$ and $x_k$ have
$d(x_j, x_k) > 2\log n.$
Then:
\begin{align*}
\mathop{\Pr}_{\rr \sim \{0, 1\}^{n}}\big[ x_j\oplus\rr \prec x_k\oplus \rr \big] &< 2^{-2\log n} = n^{-2},
\end{align*}
since if $x_{j}$ and $x_k$ differ at $i$, $\rr_i$ can only take one of two possible values to make them comparable.
Thus we can union bound over all possible pairs of queries with distance at least $2\log n$  to obtain
\[ \mathop{\Pr}_{\rr \sim \{0, 1\}^{n}}\big[ \exists\hspace{0.05cm} j, k \in [q], d(x_j, x_k) > 2\log n \text{ and } x_j\oplus \rr \prec x_k \oplus \rr \big] < n^2/n^2 =1. \]
Therefore, there exists an $r$ such that for all $j, k \in [q]$, $x_j\oplus r \prec x_k \oplus r$ implies $d(x_j, x_k) > 2\log n$.
\end{proof}

\begin{proof}[Proof of Lemma~\ref{thm:one-sided}]
Consider a non-adaptive, deterministic algorithm $B$ making $q$ queries $x_1', \ldots, x_q'\in\{0,1\}^{n+2}$, and let $x_1,\ldots,x_q$ be the last $n$ bits of these strings. We will focus on $x_1, \dots, x_q$ and refer to the sub-functions the strings query. For example $x_k$ will query the sub-function $f_{ab}$ corresponding to $a = x_{k,1}'$ and $b = x_{k, 2}'$. We may partition the set of queries $Q = \{ x_1, \dots, x_q\}$, according to the sub-function queried:
\begin{align*}
Q_{00} &= \{ x_k \in Q \colon x_{k,1}' = x_{k, 2}' = 0\} \\
Q_{01} &= \{ x_k \in Q \colon x_{k,1}' = 0, x_{k, 2}' = 1\} \\
Q_{10} &= \{ x_k \in Q \colon x_{k,1}' = 1, x_{k, 2}' = 0\} \\
Q_{11} &= \{ x_k \in Q \colon x_{k,1}' = x_{k, 2}' = 1\}.
\end{align*}

Let $r \in \{0, 1\}^{n}$ be the string such that all comparable pairs among $x_1\oplus r,\ldots,x_q\oplus r$ have distance at most $2 \log n$, which is guaranteed to exist by Lemma~\ref{lem:bad-dir}. We will show that when $r' = (0, 0, r) \in \{0, 1\}^{n+2}$, with probability $1 - o(1)$ over the draw of $\ii \sim [n]$, $B$ does not observe any violation to monotonicity of $f_{\ii}^{r'}$. 

Consider any $i \in [n]$ and one possible violation to monotonicity, given by the pair $(x_k, x_j)$ where
\[ x_k' \oplus r' \prec x_j' \oplus r' \quad \text{and} \quad f_i^{r'}(x_k') = 1, f_i^{r'}(x_j') = 0 \]
Then $x_k \notin Q_{00}$ and $x_j \notin Q_{11}$ since $f_{i, 00}^r$ and $f_{i, 11}^{r}$ are the constant $0$ and $1$ functions, respectively. Additionally, if $x_j \in Q_{00}$, then $x_k \in Q_{00}$ since $r'_{1} = r'_{2} = 0$, but this contradicts the fact that $f_i^{r'}(x_k') = 1$, so $x_j \notin Q_{00}$. Similarly, $x_k \notin Q_{11}$. 

Additionally, if $x_k \in Q_{01}$ (or $Q_{10}$) and $x_j \in Q_{10}$ (or $Q_{01}$), $x_k'$ and $x_j'$ are incomparable, so $x_{k}' \oplus r'$ and $x_{j}'\oplus r'$ are incomparable. Also, for any $i \in [n]$, either $f_{i,01}^r$ or $f_{i,10}^r$ is monotone, so it suffices to consider pairs $(x_k, x_j)$ where either both $x_k, x_j \in Q_{01}$, or both $x_k, x_j \in Q_{10}$. Consider the case $f_{i,10}^r$ is monotone, since the other case is symmetric. Therefore, it suffices to show that with probability $1-o(1)$ over the choice of $\ii \sim [n]$, $B$ does not observe any violations to monotonicity for $f_{\ii, 01}^r$ from queries in $Q_{01}$. 

Similarly to \cite{FLNRRS}, consider the graph of the queries where $x_j$ and $x_k$ are connected if $x_j \oplus r$ and $x_k\oplus r$ are comparable. Additionally, consider a spanning forest $T$ over this graph. For any $i \in [n]$, if $f_{i,01}^r(x_j) \neq f_{i,01}^r(x_k)$ when $x_j$ and $x_k$ are connected in $T$, then there exists an edge in $T$, $(y, z)$, where $f_{i,01}^r(y) \neq f_{i,01}^r(z)$. Thus, it suffices to upper-bound the probability that some edge $(y, z)$ in $T$ has $f_{i,01}^r(y) \neq f_{i,01}^r(z)$, and this only happens when $y\oplus r$ and $z\oplus r$ differ at index $i$.

 We have:
\[ \mathop{\Pr}_{\ii \sim [n]}\big[\exists\hspace{0.05cm} (y, z) \in T: f_{\ii,01}^{r}(y) \neq f_{\ii,01}^r(z)\big] \leq \dfrac{q \cdot 2\log n}{n} \]
since the two end points of each edge have hamming distance at most $2\log n$ (recall our choice for $r$). We union bound over at most $q$ edges in $T$ to conclude that with probability at least $1 - 2q\log n/n$ over the draw $\ii \sim [n]$, $B$ does not observes a violation to monotonicity for $f_{\ii,01}^r$ in $Q_{01}$. When $\smash{q = {n}/{\log^2 n}}$, this probability is at least $1 - o(1)$.
\end{proof}

\section{Non-Adaptive Monotonicity Lower Bound}\label{sec:non-mono}

In this section, we present the proof that \emph{non-adaptive} monotonicity testing requires $\tilde{\Omega}(\sqrt{n})$ queries. The previous best non-adaptive lower bound for testing monotonicity is from \cite{CDST15}, where they show that for any $c > 0$, testing monotonicity requires $\Omega(n^{1/2 - c})$ many queries. Since this lower bound matches the known upper bound from \cite{KMS15}, our result is tight up to poly-logarithmic factors. The following distribution and proof is very similar to the work in \cite{BB15}. 

We use distributions over Boolean functions very similar to the distributions used in \cite{BB15}. A function $\ff \sim \Dy$ is drawn using the following procedure:
\begin{flushleft}\begin{enumerate}
\item Sample $\TT\sim\calE$ ($\calE$ is the same distribution over terms used in Section~\ref{sec:unate}). Then $\TT$ is used to define the multiplexer map $\bGamma = \bGamma_{\TT} \colon \{0, 1\}^n \to [N] \cup \{ 0^*, 1^* \}$. 
\item Sample  $\HH = (\hh_{i} \colon i \in [N])$ from a distribution $\Ey$, where each $\hh_{i} \colon \{0, 1\}^n \to \{0,1\}$ is a random dictatorship Boolean function, i.e., $\hh_{i}(x) = x_k$ with $k$ sampled independently and uniformly at random from $[n]$.
\item Finally, $\ff \colon \{0, 1\}^n \to \{0, 1\}$ is defined as follows: $\ff(x) = 1$ if $|x| > (n/2) + \sqrt{n}$; $\ff(x) = 0$ if $|x| < (n/2) - \sqrt{n}$; if $(n/2) -\sqrt{n} \leq |x|\leq (n/2) + \sqrt{n}$, we have
\[ \ff(x) = \left\{ \begin{array}{cc} 0 & \bGamma(x) = 0^* \\
						1 & \bGamma(x) = 1^* \\
						\hh_{\bGamma(x)}(x) & \text{otherwise (i.e., $\bGamma(x) \in [N]$)} \end{array} \right. \]
\end{enumerate}\end{flushleft}
A function $\ff \sim\Dn$ is drawn using the same procedure, with the only difference being that $\HH = (\hh_{i} \colon i \in [N])$ is drawn from $\En$ (instead of $\Ey$): each $\hh_i(x) = \overline{x_k}$ is a random anti-dictatorship Boolean function with $k$ drawn independently and uniformly from $[n]$. 

Similarly to Section~\ref{sec:mono}, the truncation allows us to show lower bounds against algorithms that query strings in the middle layers. The following two lemmas are easy extensions of Lemma~\ref{monotone:lem} and Lemma~\ref{nonmonotone:lem} in Section~\ref{sec:mono}.

\begin{lemma}\label{lem:mono}
Every function in the support of $\Dy$ is monotone. 
\end{lemma}

\begin{lemma}\label{lem:nonmono}
A function $\ff \sim \Dn$ is $\Omega(1)$-far from monotone with probability $\Omega(1)$. 
\end{lemma}

Below, we fix $q = \sqrt{n} / \log^2 n$. Recall from Section~\ref{sec:nonadaptive} that a non-adaptive, deterministic algorithm $B$ is a set of $q$ query strings $x_1, \dots, x_q$, as well as a decision procedure which outputs ``accept'' or ``reject'' given $f(x_k)$ for each $k \in [q]$. Thus, in order to prove the lower bound, it suffices to prove the following lemma:

\begin{lemma} \label{non-adaptive-mono-bound}
Let $B$ be any non-adaptive deterministic algorithm with oracle access to $f$ making $q = \sqrt{n} / \log^2 n$ queries. Then
\[ \mathop{\Pr}_{\ff\sim\Dy}[B \text{ accepts }\ff] \leq \mathop{\Pr}_{\ff\sim\Dn}[B\text{ accepts }\ff] + o(1) \] 
\end{lemma}

We follow in a similar fashion to Subsection~\ref{sec:unate-oracle} by considering a stronger oracle model that results more than just $f(x) \in \{0, 1\}$. In particular, we will use the oracle model from Subsection~\ref{sec:unate-oracle}, where on query $x \in \{0, 1\}^n$, the oracle reveals the signature of $x$ with respect to $(T, H)$ as described in Definition~\ref{def:unate-sig}. From Lemma~\ref{lem:simul}, this new oracle is at least as powerful as the standard oracle. Recall the definitions of the 5-tuple $(I;P;R;A;\rho)$ from Subsection~\ref{sec:unate-oracle}. To summarize, the algorithm $B$ with oracle access to the signatures with respect to $(T, H)$ works in the following way:
\begin{flushleft}\begin{enumerate}
\item Query the oracle with queries $Q = \{x_1, \dots, x_q\} \subset \{0, 1\}^n$.
\item Receive the full signature map of $Q$ with respect to $(T, H)$, and build the 5-tuple $(I;P;R;A;\rho)$.
\item Output ``accept'' or ``reject".
\end{enumerate}\end{flushleft}

We think of an algorithm $B$ as a list of possible outcome, $L = \{ \ell_1, \ell_2, \dots \}$, where each outcome corresponds to an execution of the algorithm. Thus, each $\ell_i$ is labelled with a full-signature map of $Q$ (and therefore, a 5-tuple) as well as ``accept'' or ``reject''. These possible outcomes are similar in nature to the leaves in Section~\ref{sec:mono} and Section~\ref{sec:unate}. 

We proceed in a similar fashion to Section~\ref{sec:mono} and Section~\ref{sec:unate}, by first identifying some \emph{bad outcomes}, and then proving that for the remaining good outcomes, $B$ cannot distinguish between $\Dy$ and $\Dn$. Note that since our algorithm is non-adaptive, $B$ is not a tree; thus, there are no edges like in Section~\ref{sec:mono} and Section~\ref{sec:unate}. For the remainder of the section, we let $\alpha > 0$ be a large constant. 

\begin{definition}\label{def:bad-events}
For a fixed 5-tuple, $(I;P;R;A;\rho)$, we say the tuple is \emph{bad} if:
\begin{flushleft}\begin{itemize}
\item For some $i \in I$, there exists $x, y \in P_i$ where $|\{ k \in [n] \mid x_k = y_k = 1\}| \leq (n/2) - \alpha \sqrt{n} \log n$.
\item For some $i \in I$, $P_i$ is inconsistent (recall definition of inconsistent from Definition~\ref{def:incons}).
\end{itemize}\end{flushleft}
\end{definition}
We will say an outcome $\ell$ is bad if the 5-tuple at $\ell$, given by $(I;P;R;A;\rho)$ from the full signature map at $\ell$ is bad. Thus, we may divide the outcomes into $L_{B}$, consisting of the bad outcomes, and $L_G$, consisting of the good outcomes. Similarly to Section~\ref{sec:mono} and Section~\ref{sec:unate}, Lemma~\ref{non-adaptive-mono-bound} follows from the following two lemmas.

\begin{lemma}\label{lem:prune-non-adaptive}
Let $B$ be a non-adaptive $q$-query algorithm. Then
\[ \mathop{\Pr}_{\TT\sim\calE, \HH\sim\Ey}[ (\TT, \HH) \text{ results an outcome in } L_B] = o(1). \]
\end{lemma}

We prove the following lemma for good outcomes. 
\begin{lemma}
For any non-adaptive, $q$-query algorithm $B$, if $\ell \in L_G$ is a good outcome, 
\[ \mathop{\Pr}_{\TT\sim\calE, \HH\sim\Ey}[(\TT, \HH) \text{ results in outcome $\ell$}] \leq (1 + o(1)) \mathop{\Pr}_{\TT\sim\calE, \HH\sim\En}[(\TT, \HH) \text{ results in outcome $\ell$}]. \]
\end{lemma}

\begin{proof}
Fix a good outcome $\ell \in L_G$, and let $\phi \colon Q \to \frakP$ be the associated full signature map and $(I;P;R;A;\rho)$ be the associated $5$-tuple. Since $(I;P;R;A;\rho)$ is not bad:
\begin{flushleft}\begin{itemize}
\item For all $i \in I$, and $x, y \in P_i$, $|\{k \in [n] \mid x_k = y_k = 1\}| \geq (n/2) - \alpha \sqrt{n} \log n$; hence, by Lemma 19 in \cite{BB15}, 
\[ \Big| |A_{i, 1}| - |A_{i, 0}| \Big| \leq O(|P_i| \sqrt{n} \log n) \]
\item For all $i \in I$, $P_i$ is either $1$-consistent, or $0$-consistent. We use the $\rho_i$ to denote the value $\rho_i(x)$ shared by all $x \in P_i$.
\end{itemize}\end{flushleft}
Consider a fixed $T$ in the support of $\calE$ such that the probability of $(T, \HH)$ resulting in outcome $\ell$ is positive when $\HH\sim\Ey$. Then it suffices to show that 
\[ \dfrac{\Pr_{\HH\sim\En}[(T, \HH) \text{ results in outcome $\ell$}]}{\Pr_{\HH\sim\Ey}[(T, \HH) \text{ results in outcome $\ell$}]} \geq 1 - o(1). \]
We know that $T$ matches the full signature $\phi$ at $\ell$. Now, to match the $a_x$ and $b_x$ for each $x \in Q$ given in $\phi$, $H$ (from either $\Ey$ and $\En$) needs to satisfy the following condition:
\begin{flushleft}\begin{itemize}
\item If $H = (h_i \colon i \in [N])$ is from the support of $\Ey$, then the dictator variable of each $h_i$, $i \in I$, is in $A_{i, \rho_i}$.
\item If $H = (h_i \colon i \in [N])$ is from the support of $\En$, then the dictator variable of each $h_i$, $i \in I$, is in $A_{i, 1-\rho_i}$.
\item If $i \notin I$, there is no condition posed on $h_i$. 
\end{itemize}\end{flushleft}
As a result, we have:
\begin{align*}
\dfrac{\Pr_{\HH\sim\En}[(T, \HH) \text{ results in outcome $\ell$}]}{\Pr_{\HH\sim\Ey}[(T, \HH) \text{ results in outcome $\ell$}]} &= \prod_{i \in I} \left(\dfrac{|A_{i, 1-\rho_i}|}{|A_{i, \rho_i}|} \right) \\
			  &\geq \prod_{i \in I} \left(1 - \dfrac{\big||A_{i, \rho_i}| - |A_{i, 1-\rho_i}|\big|}{|A_{i, \rho_i}|} \right) \\
			  &\geq \prod_{i \in I} \left(1 - O\left( \dfrac{|P_i|\log n}{\sqrt{n}} \right) \right) = 1 - o(1),
\end{align*}
when $q = \sqrt{n} / \log^2 n$.
\end{proof}

We now prove Lemma~\ref{lem:prune-non-adaptive}, which allows us to only consider good outcomes. 

\begin{proof}[Proof of Lemma~\ref{lem:prune-non-adaptive}]
We first handle the first case of bad outcomes: some $i \in I$ has $x, y \in P_i$ where $|\{ k \in [n] \mid x_k = y_k = 1\} \leq (n/2) - \alpha \sqrt{n} \log n$. This case is almost exactly the same as Lemma~16 of \cite{BB15}. Since the probability some $\TT\sim\calE$ is sampled with the above event happening is at most:
\[ 2^{\sqrt{n}} q^2 \left( \dfrac{(n/2) - \alpha \sqrt{n} \log n}{n} \right)^{\sqrt{n}} = q^2 \left(1 - \alpha n^{-1/2} \log n \right)^{\sqrt{n}} \leq q^2 n^{-\alpha} = o(1) \]
since $\alpha > 0$ is a large constant and $q^2 \leq n$. Thus, by Lemma~19 in \cite{BB15}, all $i \in I$ satisfy
\[ \Big|[n] \setminus A_{i, 0} \setminus A_{i, 1}\Big| \leq O(|P_i| \sqrt{n} \log n). \]
For the second case, in order for some $P_i$ to be inconsistent, $\hh_i(x) = x_k$ sampled according to $\Ey$ must have $k \in [n] \setminus A_{i, 0} \setminus A_{i, 1}$. Thus, taking a union bound over all possible $i \in I$, the probability over $\HH\sim\Ey$ of resulting in an outcome where some $i \in I$ is inconsistent is at most
\[ \sum_{i\in I} \left(\dfrac{\big|[n]\setminus A_{i, 0} \setminus A_{i, 1}\big|}{n} \right) \leq \sum_{i \in I} \left( \dfrac{ O(|P_i|\sqrt{n} \log n)}{n}\right) = o(1) \]
since $\sum_{i\in I} |P_i| \leq 2q = 2\sqrt{n} / \log^2 n$.
\end{proof}

\section{Tightness of Distributions for Monotonicity}\label{sec:alg}







In this section, we provide the reader with some intuition of why the analyses of \cite{BB15} and this paper are tight. In particular, we sketch one-sided algorithms to find violating pairs in the far-from-monotone functions from the distributions considered. We maintain this discussion at a high level.

\subsection{An $O(n^{1/4})$-query algorithm for distributions of \cite{BB15}}
\label{sec:bb-alg}

Belovs and Blais define a pair of distributions $\dyes^*$ and $\dno^*$ over 
   functions of $n$ variables. To describe $\dyes^*$ and $\dno^*$, recall Talagrand's random DNF \cite{Talagrand} (letting $\smash{N=2^{\sqrt{n}}}$): A function $f$ drawn from $\Tal$ is 
  the disjunction of $N$ terms $T_i$, $i\in [N]$, where each~$T_i$ is the conjunction of 
  $\sqrt{n}$ variables sampled independently and
  uniformly from $[n]$.

\def\SS{\boldsymbol{S}}

Next we use $\Tal$ to define~$\Talpm$. To draw a function $\gg$ from $\Talpm$, one samples an $\ff$ from $\Tal$ and a random $\sqrt{n}$-subset $\SS$ of $[n]$.\hspace{0.04cm}\footnote{Formally,
  $\SS$ is sampled by including each element of $[n]$ independently with probability $1/\sqrt{n}$.} Then $\smash{\gg(x)=\ff(x^{(\SS)})}$, where $x^{(\SS)}$ is the string obtained from $x$ by 
  flipping each coordinate in $\SS$.
 Equivalently variables in $T_i\cap \SS$ appear
  negated in the conjunction of $T_i$. 
The $\dyes^*$ distribution is then the truncation of $\Tal$, and the $\dno^*$ distribution is~the truncation of $\Talpm$. Every $\ff\sim \dyes^*$ is monotone by definition;
  \cite{BB15} shows that~$\gg\sim\dno^*$~is far from monotone using the
   extremal noise sensitivity property of Talagrand functions \cite{MosselOdonnell:03}.

We now sketch a $O(n^{1/4})$-query one-sided algorithm 
  that rejects $\gg\sim \dno^*$ with high probability.
Note that the description below is not a formal analysis; the goal is to 
  discuss the main idea behind the algorithm. Let $g$ be a function in the support of $\dno^*$ defined by $T_i$ and $S$ with $T_i'=T_i\setminus S$. 
Then the algorithm starts by sampling a random $x \in \{0, 1\}^n$ in the middle layers with $g(x)=1$.
It is likely ($\Omega(1)$ probability by a simple calculation)
  that: 
\begin{enumerate}
\item $x$ satisfies a unique term $T_k'$~among all $T_i'$'s.
\item $T_k\cap S$ contains a unique $\ell\in [n]$ (by 1). 
\item $T_k=T_k'\cup \{\ell\}$ and $x$ has $x_\ell=0$ (since $g(x)=1$).
 \end{enumerate} 
Assume this is the case, and 
  let $A_0$ and~$A_1$ denote the set of $0$-indices and~$1$-indices~of $x$,
  respectively. 
Then $T_k'\subseteq A_1$ and $\ell\in A_0$.

The first stage of the algorithm goes as follows:
\begin{flushleft}\begin{enumerate}
\item[] \textbf{Stage 1.} Repeat the following for $\smash{n^{1/4}}$ times: Pick a random subset $\smash{R\subset A_1}$ of size $\smash{\sqrt{n}}$ and 
  query $\smash{g(x^{(R)})}$.
By 1) and 2) above, $\smash{g(x^{(R)}))=1}$ if and only if $R\cap T_k'=\emptyset$, which
  happens with $\Omega(1)$ probability.
Let $A_1'$ denote $A_1$ after removing those indices of $R$ with $\smash{g(x^{(R)}))=1}$ encountered.
Then we have  $T_k'\subset A_1'$ and most likely, $C=A_1\setminus A_1'$ has size $\Theta(n^{3/4})$.
\end{enumerate}\end{flushleft}

After the first stage, the algorithm has shrunk $A_1$ by $\Theta(n^{3/4})$  
  while still making sure that variables of $T_k'$ lie in $A_1'$.
In the second stage, the algorithm takes advantage of the smaller $A_1$ to search for~$\ell$ in $A_0$,
  with each query essentially covering $\Theta(n^{3/4})$ indices of $A_0$:
\begin{flushleft}\begin{enumerate}
\item[] \textbf{Stage 2.} Randomly partition $A_0$ into $O(n^{1/4})$ many disjoint 
  parts $A_{0,1},A_{0,2},\ldots$, each of size $|C|=\Theta(n^{3/4})$.
For each $A_{0,j}$, query $g(x^{(A_{0,j}\cup C)})$.
For each $A_{0,j}$ with $\ell\notin A_{0,j}$, $g$ must return $1$;
for the $A_{0,h}$ with $\ell\in A_{0,h}$, $g$ returns $0$ with $\Omega(1)$ probability\hspace{0.04cm}\footnote{Informally speaking, this is because the values of $g(x)$ and $g(y)$
  essentially become independent when $x$ and $y$~are far from each other.} and
  when this happens, the algorithm has found a $O(n^{3/4})$-size subset $A_{0,h}$ of $A_0$ containing $\ell$. Let $y = x^{(A_{0, j}\cup C)}$.
\end{enumerate}\end{flushleft}
Note that the algorithm cannot directly query $g(x^{(A_{0,j})})$ since the new string will be outside of
  the middle layers (unless $|A_{0,j}|=O(\sqrt{n})$, in which case one needs $\Omega(\sqrt{n})$
  queries to cover $A_0$). This~is only achieved by flipping $A_{0,j}$ and $C$ at the same time (in different directions) 
  and this is the reason why we need the first stage to shrink $A_1$. In the last stage, the algorithm will find a violation for $y$, by providing $z \prec y$ with $g(z) = 1$. 
  \begin{flushleft}\begin{enumerate}
\item[] \textbf{Stage 3.} Randomly partition $A_{0,h}$ into $O(n^{1/4})$ many disjoint 
  parts $\Delta_{1}, \Delta_2,\ldots$, each of size $O(\sqrt{n})$.
For each $\Delta_i$, query $g(y^{(\Delta_i)})$. When $\ell \in \Delta_i$, $g(y^{(\Delta_i)}) = 1$ with probability $\Omega(1)$, and $y^{(\Delta_i)} \prec y$. 
\end{enumerate}\end{flushleft}

\subsection{An $O(n^{1/3})$-query algorithm for our distributions}

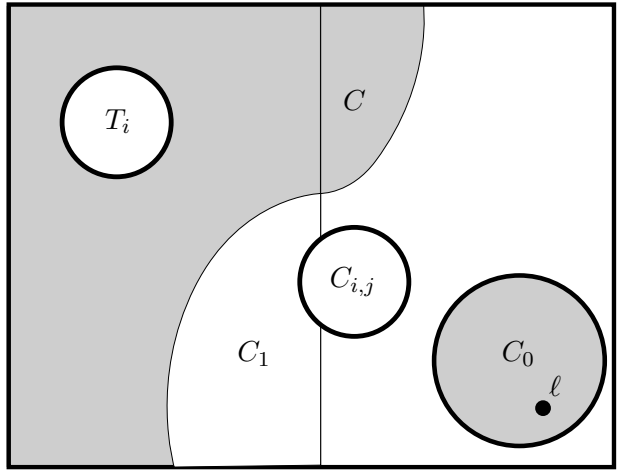
\begin{figure}
\label{fig:algo}
\centering
\begin{picture}(210,200)
    \put(0,0){\includegraphics[width=0.5\linewidth]{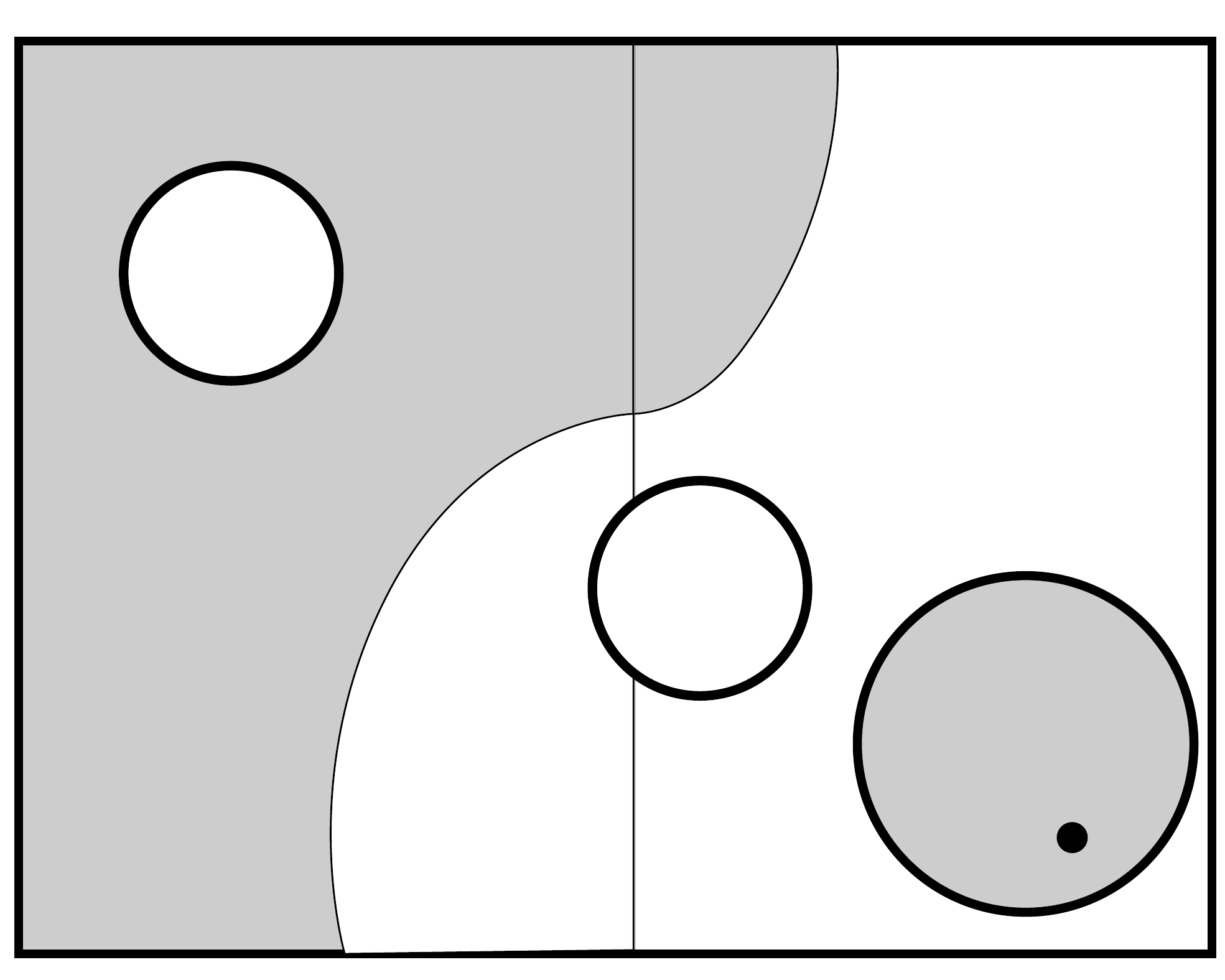}}
    \put(40, 133){$T_i$}
    \put(90, 45){$C_1$}
    \put(130, 140){$C$}
    \put(125, 73){$C_{i, j}$}
    \put(190, 45){$C_0$}
    \put(208, 32){$\ell$}
  \end{picture}\vspace{0.18cm}
  \caption{A visual representation of the algorithm for finding violations in the two-level Talagrand construction. The whole rectangle represents the set $[n]$, which is shaded for coordinates which are set to $1$, and clear for coordinates which are set to $0$. $T_i$ is the unique term satisfied and $C_{i, j}$ is the unique clause falsified. The functions $h_{i, j}$ is an anti-dictator of coordinate $\ell$. The sets illustrated represent the current knowledge at the end of Stage 3 of the algorithm. Note that $|C_1| = \Theta(n^{5/6})$, $|C| = \Theta(n^{2/3})$, $|C_0| = n^{5/6}$, $|T_i| = |C_{i, j}| = \Theta(\sqrt{n})$. }\label{fig:algo}
\end{figure}

The idea sketched above can be applied to our far from monotone distribution $\dno$ from Section~\ref{sec:mono}. It is slightly more complicated, since now the algorithm must attack two levels of Talagrand, which will incur the query cost of $\tilde{O}(n^{1/3})$ rather than $O(n^{1/4})$. Similarly to Subsection~\ref{sec:bb-alg} above, we will give a high level description, and not a formal analysis. The goal is to show the main obstacle one faces in improving the lower bound. 

Assume $g$ is in the support of $\dno$. The algorithm works in stages and follows a similar pattern to the one described in Subsection~\ref{sec:bb-alg} above. We may assume the algorithm has a string $x \in \{0, 1\}^n$ where $x$ satisfies a unique term $T_i$, and falsifies no clauses, so $g(x) = 1$ (this happens with $\Omega(1)$ probability for a random $x$). 
\begin{flushleft}\begin{enumerate}
\item[] \textbf{Stage 1.} Repeat the following for $\smash{n^{1/3}}$ times: Pick a random subset $\smash{R\subset A_1}$ of size $\smash{\sqrt{n}}$\\ and 
  query $\smash{g(x^{(R)})}$.
Let $A_1'$ denote $A_1$ after removing those indices of $R$ with $\smash{g(x^{(R)}))=1}$ encountered.
Then we have  $T_i\subset A_1'$ and most likely, $C_1=A_1\setminus A_1'$ has size $\Theta(n^{5/6})$.
\end{enumerate}\end{flushleft}
The following stages will occur $n^{1/6}$ many times, and each makes $n^{1/6}$ many queries. 
\begin{flushleft}\begin{enumerate}
\item[] \textbf{Stage 2.} Pick a random subset $\smash{C_0 \subset A_0}$ of size $\smash{n^{5/6}}$. Let $y = x^{(C_1 \cup C_0)}$ and query $g(y)$. With probability $\Omega(1)$, $g(y)$ satisfies the unique term $T_i$ (as did $x$), falsifies a unique clause $C_{i, j}$, and $h_{i,j}(y) = 0$. Additionally, with probability $\Omega(n^{-1/6})$, $h_{i, j}(y) = \overline{y_{\ell}}$, where $\ell \in C_0$.
\end{enumerate}\end{flushleft}
Assume that $\ell \in C_0$, which happens with $\Omega(n^{-1/6})$ probability. In the event this happens, we will likely find a violation. 
\begin{flushleft}\begin{enumerate} 
\item[] \textbf{Stage 3.} Repeat the following for $\smash{n^{1/6}}$ times: Pick a random subset $\smash{R\subset A_0 \setminus C_0}$ of size $\smash{\sqrt{n}}$ and 
  query $\smash{g(y^{(R)})}$. Let $A_0'$ denote $A_{0} \setminus C_0$ after removing those indices of $R$ with $g(y^{(R)}) = 0$. Let $C = (A_0 \setminus C_0) \setminus A_0'$, where very likely $|C| = \Theta(n^{2/3})$. Our sets satisfy the following three conditions: 1) $T_i \subset A_1'$, 2) $C_{i, j} \subset A_0' \cup C_1 \setminus C_0$, and 3) $\ell \in C_0$. See Figure~\ref{fig:algo} for a visual representation of these sets.  
\item[] \textbf{Stage 4.} Partition $C_0$ into $O(n^{1/6})$ many disjoint parts $C_{0, 1}, C_{0, 2}, \dots$, each of size $\Theta(n^{2/3})$ and query $g(y^{(C_{0, j} \cup C)})$. For each $C_{0, j}$ with $\ell \notin C_{0, j}$ and no new terms are satisfied, $g$ must return $0$. If for some sets $C_{0, j}$, $g$ returns 1, then either $\ell \in C_{0,j}$ and no new terms are satisfied, or new terms are satisfied; however, we can easily distinguish these cases with a statistical test. 
\end{enumerate}\end{flushleft}
The final stage is very similar to the final stage of Subsection~\ref{sec:bb-alg}. After Stage 4, we assume we have found a set $C_{0, j}$ containing $\ell$. We further partition $C_{0, j}$ (when $g(y^{(C_{0,j}\cup C)}) = 1$) into $O(n^{1/6})$ parts of size $\sqrt{n}$ to find a violation. One can easily generalize the above algorithm sketch to $O(1)$-many levels of Talagrand. This suggests that the simple extension of our construction to $O(1)$ many levels (which still gives a far-from-monotone function) cannot achieve lower bounds better than $n^{1/3}$.


\section{Discussion and Open Problems}

While our two-level Talagrand functions for monotonicity testing looked
  promising at first sight,~a few 
   issues remain, which allow an algorithm to find a violating pair
  with $O(n^{1/3})$ queries 
  (see Section~\ref{sec:alg}). However, for the problem of testing unateness, a different and simpler 
    pair of distributions allows us to overcome the $n^{1/3}$ obstacle for monotonicity and establish an $\tilde{\Omega}(\sqrt{n})$ lower bound for unateness.  The multiplexer maps of Section~\ref{sec:unate} turn out to be more resilient to the kinds of attacks sketched in Section~\ref{sec:alg}, so one can imagine adapting them to the monotonicity testing setting. This leads us to the following conjecture:
\begin{conjecture}
Adaptivity does not help for monotonicity testing.
\end{conjecture}

With regards to testing unateness, 
  our adaptive $\tilde{\Omega}(\sqrt{n})$ lower bound exploited the existence of more resilient multiplexer   maps. Although 
preliminary work suggests that the pair of \mbox{distributions} employed in our lower bound proof for unateness \emph{can} be distinguished with $O(\sqrt{n})$ queries, 
  it looks promising to us that small modifications to these distributions may yield lower bounds asymptotically higher than $\sqrt{n}$. This leads us to the following conjecture:
\begin{conjecture}
Testing unateness is strictly harder than testing monotonicity.
\end{conjecture}

\section*{Acknowledgments}

We thank Rocco Servedio and Li-Yang Tan for countless discussions and suggestions.
This work~is supported in part by NSF CCF-1149257, CCF-1423100  and the NSF Graduate Research Fellowship under Grant No. DGE-16-44869.

\begin{flushleft}
\bibliographystyle{alpha}
\bibliography{all,allrefs2,allrefs,odonnell-bib}
\end{flushleft}

\appendix

\section{A claim about products}
Recall Bernoulli's inequality: For every real number $a\ge 1$ and real number $x\ge -1$, we have
$$
(1+x)^a\ge 1+ax,
$$
and for every real number $0\le a\le 1$ and real number $x\ge -1$, we have
$$
(1+x)^a\le 1+ax.
$$

We prove the following claim used in Section~\ref{sec:goodleaves}. 

\begin{claim}\label{cl:real-nums}
Let $t \leq n^{2/3}$ and $c_1, \dots, c_t \in \R$ be numbers with $|c_i| \leq {\log^2 n}/{\sqrt{n}}$. We have
\[ \prod_{i\in [t]} \big( 1 - c_i \big) \geq \big(1 - o(1)\big) \cdot \left(1  - \sum_{i\in [t]} c_i\right), \]
where the asymptotic notation is with respect to $n$. 
\end{claim}

\begin{proof}
Let $\beta=\log^2n/\sqrt{n}$.
Assume without loss of generality that $$c_1, \dots, c_k \geq 0\quad\text{and}\quad 
c_{k+1}, \dots, c_t < 0$$ for some $k \leq t$. Let $\delta_i =c_i/\beta$ for $i \leq k$ and $\tau_j = - c_j /\beta$ for $j>k$. Thus, $\delta_i, \tau_j \in [0, 1]$ and $$\sum_{i\in [t]} c_i = \beta \left(\sum_{i\le k} \delta_i - \sum_{j>k} \tau_j \right).$$ 
Let $\Delta = \sum_{i\le k} \delta_i - \sum_{j>k} \tau_i$. By Bernoulli's inequality, we also have
\[ 1 - c_i \geq \left(1 - \beta\right)^{\delta_i} \quad \text{and} \quad 1 - c_j \geq \left(1 + \beta\right)^{\tau_j}. \]
As a result, it remains to show that
\[ \left(1 - \beta \right)^{\sum_{i\le k} \delta_i} \cdot \left( 1 + \beta\right)^{\sum_{j>k} \tau_j} \geq (1 - o(1)) \left(1 - \beta\Delta \right). \]
We consider two cases: $\Delta>0$ or $\Delta\le 0$.
If $\Delta > 0$, we have
\begin{align*}
\left(1 - \beta \right)^{\sum_i \delta_i} \cdot \left( 1 +\beta\right)^{\sum_j \tau_i} &= \left(1 - \beta \right)^{\Delta}\cdot  \left( 1 - \beta^2\right)^{\sum_j \tau_j}
\ge (1-o(1))\cdot (1-\beta)^{\Delta} 
\end{align*}
using $\beta^2=\log^4/n$ and $\sum_j \tau_j \leq n^{2/3}$. 
When $\Delta\ge 1$ it follows by Bernoulli's inequality that~$(1-\beta)^\Delta$ $\ge 1-\beta\Delta$
  and we are done.
When $0<\Delta<1$, we have from $\beta=o(1)$ and $\beta\Delta=o(1)$ that
$$
(1-\beta)^\Delta>1-\beta\ge (1-o(1))\cdot (1-\beta\Delta).
$$

The case when $\Delta \le 0$ is similar:
\begin{align*}
\left(1 - \beta \right)^{\sum_i \delta_i} \cdot \left( 1+\beta\right)^{\sum_{j} \tau_i} &= \left(1 + \beta \right)^{-\Delta}\cdot \left( 1 - \beta^2\right)^{\sum_{i} \delta_i} \geq (1 - o(1))
  \cdot (1+\beta)^{-\Delta}. 
\end{align*}
When $\Delta\le -1$, it follows from Bernoulli's inequality that
  $(1+\beta)^{-\Delta}\ge 1-\beta\Delta$ and we are done.
If $-1<\Delta\le 0$, 
  we have from $- \beta\Delta=o(1)$ that
$ 
(1+\beta)^{-\Delta}>1>(1-o(1))\cdot (1-\beta\Delta).
$ 
\end{proof}

\end{document}